\def\llncs{0}
\def\fullpage{1}
\def\anonymous{0}
\def\authnote{1}
\def\notxfont{0}
\def\submission{0}

\def\extendedabstract{0}

\ifnum\submission=1
\def\llncs{1}
\fi

\ifnum\llncs=1
\documentclass[envcountsect,a4paper,runningheads,10pt]{llncs}
\else
	\documentclass[letterpaper,hmargin=1.05in,vmargin=1.05in,11pt]{article}
			\ifnum\fullpage=1
		\usepackage{fullpage}
		\fi
\fi

\usepackage{CJKutf8}

\usepackage[%
  colorlinks=true,
  citecolor=blue,
  pagebackref=true
]{hyperref}

\usepackage{amsmath, amsfonts, amssymb, mathtools,amscd}

\usepackage{amsthm}

\usepackage{lmodern}
\usepackage[T1]{fontenc}
\usepackage[utf8]{inputenc}

\usepackage{etex}

\usepackage{arydshln} 
\usepackage{url}
\usepackage{ifthen}
\usepackage{bm}
\usepackage{multirow}
\usepackage[dvips]{graphicx}

\usepackage[usenames]{color}
\usepackage{xcolor,colortbl} 
\usepackage{threeparttable}
\usepackage{comment}
\usepackage{paralist,verbatim}
\usepackage{cases}
\usepackage{booktabs}
\usepackage{braket}
\usepackage{cancel} 
\usepackage{ascmac} 
\usepackage{framed}
\usepackage{authblk}
\usepackage{pifont}
\usepackage{qcircuit}
\usepackage{tikz}
\usetikzlibrary{cd}
\definecolor{darkblue}{rgb}{0,0,0.6}
\definecolor{darkgreen}{rgb}{0,0.5,0}
\definecolor{maroon}{rgb}{0.5,0.1,0.1}
\definecolor{dpurple}{rgb}{0.2,0,0.65}

\usepackage[capitalise,noabbrev]{cleveref}
\usepackage[absolute]{textpos}
\usepackage[final]{microtype}
\usepackage[absolute]{textpos}
\usepackage{everypage}
\DeclareMathAlphabet{\mathpzc}{OT1}{pzc}{m}{it}

\usepackage{algorithmic}
\usepackage{algorithm}
\usepackage{here}

\usepackage{thm-restate}

\usepackage[normalem]{ulem}

\newtheoremstyle{thicktheorem}%
{\topsep}
{\topsep}
{\itshape}{}%
{\bfseries}%
{.}
{ }%
{\thmname{#1}\thmnumber{ #2}%
		\thmnote{ (#3)}%
}

\newtheoremstyle{remark}
{\topsep}
{\topsep}
	{}
	{}
	{}
	{.}
	{ }
	{\textit{\thmname{#1}}\thmnumber{ #2}
			\thmnote{ (#3)}%
	}

\ifnum\llncs=0
	\theoremstyle{thicktheorem}
	\newtheorem{theorem}{Theorem}[section]
	\newtheorem{lemma}[theorem]{Lemma}
	\newtheorem{corollary}[theorem]{Corollary}
	
	\newtheorem{definition}[theorem]{Definition}

	\theoremstyle{remark}
	\newtheorem{claim}[theorem]{Claim}
	\newtheorem{remark}[theorem]{Remark}

\else
\fi
	\crefname{theorem}{Theorem}{Theorems}
	\crefname{assumption}{Assumption}{Assumptions}
	\crefname{construction}{Construction}{Constructions}
	\crefname{corollary}{Corollary}{Corollaries}
	\crefname{conjecture}{Conjecture}{Conjectures}
	\crefname{definition}{Definition}{Definitions}
	\crefname{exmaple}{Example}{Examples}
	\crefname{experiment}{Experiment}{Experiments}
	\crefname{counterexample}{Counterexample}{Counterexamples}
	\crefname{lemma}{Lemma}{Lemmata}
	\crefname{observation}{Observation}{Observations}
	\crefname{proposition}{Proposition}{Propositions}
	\crefname{remark}{Remark}{Remarks}
	\crefname{claim}{Claim}{Claims}
	\crefname{fact}{Fact}{Facts}
	\crefname{note}{Note}{Notes}

\ifnum\llncs=1
 \crefname{appendix}{App.}{Appendices}
 \crefname{section}{Sec.}{Sections}
\else
\fi

\ifnum\llncs=1
\renewcommand*{\backref}[1]{}
\else
	\renewcommand*{\backref}[1]{(Cited on page~#1.)}
	\ifnum\notxfont=1
	\else
		\usepackage{newtxtext}
	\fi
\fi

\usepackage{fancyhdr}

\ifnum\authnote=0  
\newcommand{\mor}[1]{}
\newcommand{\shogo}[1]{}
\newcommand{\takashi}[1]{}
\newcommand{\fuyuki}[1]{}
\newcommand{\minki}[1]{}
\newcommand{\yaoting}[1]{}

\else
\newcommand{\mor}[1]{$\ll$\textsf{\color{red} Tomoyuki: { #1}}$\gg$}
\newcommand{\takashi}[1]{$\ll$\textsf{\color{orange} Takashi: { #1}}$\gg$}
\newcommand{\shogo}[1]{$\ll$\textsf{\color{darkgreen} Shogo: { #1}}$\gg$}
\newcommand{\fuyuki}[1]{$\ll$\textsf{\color{darkblue} Fuyuki: { #1}}$\gg$}
\newcommand{\minki}[1]{$\ll$\textsf{\color{darkblue} Minki: { #1}}$\gg$}

\newcommand{\yaoting}[1]{$\ll$\textsf{\color{magenta} Yao-Ting: { #1}}$\gg$}

\DeclareRobustCommand{\Erase}{\bgroup\markoverwith{\textcolor{red}{\rule[.5ex]{2pt}{0.4pt}}}\ULon}

\fi

\newcommand{\Choi}{\text{Choi-Jamiołkowski} }

\newcommand{\CNOT}[1]{\text{C}_{#1}\text{NOT}}

\newcommand{\unitaryPSPACE}{\mathbf{UnitaryPSPACE}}
\newcommand{\pureunitaryPSPACE}{\mathbf{pureUnitaryPSPACE}}

\newcommand{\statePSPACE}{\mathbf{StatePSPACE}}

\newcommand{\HRI}{\mathsf{HRI}}

\newcommand{\SpanSpace}{\text{span}}

\newcommand{\Tr}{\mathrm{Tr}}

\newcommand{\cA}{\mathcal{A}}
\newcommand{\cB}{\mathcal{B}}

\newcommand{\cD}{\mathcal{D}}
\newcommand{\cE}{\mathcal{E}}
\newcommand{\cF}{\mathcal{F}}

\newcommand{\cH}{\mathcal{H}}
\newcommand{\cI}{\mathcal{I}}
\newcommand{\cK}{\mathcal{K}}

\newcommand{\cM}{\mathcal{M}}
\newcommand{\cN}{\mathcal{N}}
\newcommand{\cO}{\mathcal{O}}
\newcommand{\cP}{\mathcal{P}}
\newcommand{\cQ}{\mathcal{Q}}

\newcommand{\cS}{\mathcal{S}}
\newcommand{\cT}{\mathcal{T}}
\newcommand{\cU}{\mathcal{U}}
\newcommand{\cV}{\mathcal{V}}

\newcommand{\identitymap}{\mathrm{id}}

\def\makeuppercase#1{
\expandafter\newcommand\csname tl#1\endcsname{\widetilde{#1}}
}

\def\makelowercase#1{
\expandafter\newcommand\csname tl#1\endcsname{\widetilde{#1}}
}

\newcommand{\N}{\mathbb{N}}

\newcommand{\R}{\mathbb{R}}

\newcommand{\Unitaries}{\mathbb{U}}

\newcommand{\States}{\mathbb{S}}

\newcommand{\Linear}{\text{L}}

\newcommand{\symetric}{\text{sym}}

\newcommand{\regF}{\mathbf{F}}

\newcommand{\regC}{\mathbf{C}}
\newcommand{\regE}{\mathbf{E}}
\newcommand{\regR}{\mathbf{R}}
\newcommand{\regZ}{\mathbf{Z}}

\newcommand{\regB}{\mathbf{B}}
\newcommand{\regA}{\mathbf{A}}
\newcommand{\regX}{\mathbf{X}}
\newcommand{\regY}{\mathbf{Y}}

\newcommand{\secp}{\lambda}

\newcommand{\Adv}{\mathsf{Adv}}

\newenvironment{boxfig}[2]{\begin{figure}[#1]\fbox{\begin{minipage}{0.97\linewidth}
                        \vspace{0.2em}
                        \makebox[0.025\linewidth]{}
                        \begin{minipage}{0.95\linewidth}
            {{
                        #2 }}
                        \end{minipage}
                        \vspace{0.2em}
                        \end{minipage}}}{\end{figure}}

\newcommand{\bit}{\{0,1\}}

\newcommand{\negl}{{\mathsf{negl}}}

\newcommand{\poly}{{\mathrm{poly}}}

\DeclareMathOperator*{\Exp}{\mathbb{E}}

\newcommand{\AlgInput}{\textbf{Input: }}
\newcommand{\AlgOutput}{\textbf{Output: }}
\newcommand{\AlgOracle}{\textbf{Oracle access: }}

\usetikzlibrary{decorations.markings}
\tikzset{
  cross/.style={
    postaction={decorate,decoration={markings,
    mark=at position 0.45 with {\draw[-,line width=1pt] (-10pt,-10pt) -- (10pt,10pt);\draw[-,line width=1pt] (-10pt,10pt) -- (10pt,-10pt);}}}
  }
}

\makeatletter
\DeclareRobustCommand
  \myvdots{\vbox{\baselineskip4\p@ \lineskiplimit\z@
    \hbox{.}\hbox{.}\hbox{.}}}
\makeatother

\newcommand{\ketbra}[2]{\lvert #1 \rangle\mkern-3mu \langle #2 \rvert}


\title{Black-Box Separation Between Pseudorandom Unitaries, Pseudorandom Isometries, and Pseudorandom Function-Like States}

\ifnum\anonymous=1
\ifnum\llncs=1
\author{\empty}\institute{\empty}
\else
\author{}
\fi
\else
\ifnum\llncs=1
\author{
}
\institute{
	Yukawa Institute for Theoretical Physics, Kyoto University, Kyoto, Japan 
}
\else
\author[1]{Aditya Gulati}
\author[1]{Yao-Ting Lin}
\author[2]{Tomoyuki Morimae}
\author[2]{Shogo Yamada \thanks{A part of the work was done when the last author visited UCSB.}}
\affil[1]{{\small University of California, Santa Barbara, CA, USA}
\authorcr{\small \{adityagulati,yao-ting\_lin\}@ucsb.edu}
}
\affil[2]{{\small Yukawa Institute for Theoretical Physics, Kyoto University, Kyoto, Japan}
\authorcr{\small \{tomoyuki.morimae,shogo.yamada\}@yukawa.kyoto-u.ac.jp}
}

\fi 
\fi

\date{\today}

\begin{document}
\begin{CJK}{UTF8}{ipxm}

\maketitle
\begin{abstract}
Pseudorandom functions (PRFs) are one of the most fundamental primitives in classical cryptography.
On the other hand, in quantum cryptography, it is possible that
PRFs do not exist but their quantum analogues could exist, and still enabling many applications including
SKE, MACs, commitments, multiparty computations, and more.
Pseudorandom unitaries (PRUs) [Ji, Liu, Song, Crypto 2018], pseudorandom isometries (PRIs) [Ananth, Gulati, Kaleoglu, Lin, Eurocrypt 2024], and pseudorandom function-like state generators (PRFSGs) [Ananth, Qian, Yuen,
Crypto 2022] are major quantum analogs of PRFs. 
PRUs imply PRIs, and PRIs imply PRFSGs, but the converse implications remain unknown.
An important open question is whether these natural quantum analogues of PRFs are equivalent.
In this paper, we partially resolve this question by ruling out black-box constructions of them:
\begin{enumerate}
    \item There are no black-box constructions of $O(\log\secp)$-ancilla PRUs from PRFSGs.
    \item There are no black-box constructions of $O(\log\secp)$-ancilla PRIs with $O(\log\secp)$ stretch from PRFSGs.
    \item There are no black-box constructions of $O(\log\secp)$-ancilla PRIs with $O(\log\secp)$ stretch from PRIs with $\Omega(\secp)$ stretch.
\end{enumerate}
Here, $O(\log\secp)$-ancilla means that the generation algorithm uses at most $O(\log\secp)$ ancilla qubits.
PRIs with $s(\secp)$ stretch is PRIs mapping $\secp$ qubits to $\secp+s(\secp)$ qubits.
To rule out the above black-box constructions, we construct a unitary oracle that separates them.
For the separations, we construct an adversary based on the quantum singular value transformation, 
which would be independent of interest and should be useful for other oracle separations in quantum cryptography.  
\end{abstract}

\ifnum\submission=0
\clearpage
\newpage
\setcounter{tocdepth}{2}
\tableofcontents
\newpage
\fi

\section{Introduction}
Pseudorandom functions (PRFs)~\cite{JACM:GolGolMic86} are among the most fundamental primitives
in classical cryptography. 
PRFs formalize the hardness of distinguishing certain functions from truly random functions,
and have numerous important applications including IND-CPA secret-key encryption (SKE)~\cite{JACM:GolGolMic86}
and EUF-CMA message authentication codes (MAC) \cite{C:GolGolMic84}. 
Moreover, PRFs are existentially equivalent to one-way functions (OWFs)~\cite{JACM:GolGolMic86,SIAMCOMP:HILL99,SIAMCOMP:GolKraLub93,STOC:Levin85}, which indicates that
PRFs are existentially equivalent to all Minicrypt primitives and are implied by almost all computationally-secure cryptographic primitives.

In quantum cryptography, 
on the other hand, it is possible that
PRFs do not exist but quantum analogs of PRFs could exist \cite{C:JiLiuSon18,C:AnaQiaYue22,TCC:BBSS23,EC:AGKL24,TCC:LQSYZ24,TCC:BraMag24,TQC:Kre21,KreQiaTal24}, and
many applications are still possible from them~\cite{C:JiLiuSon18,C:AnaQiaYue22,C:MorYam22}.
Pseudorandom unitaries (PRUs)~\cite{C:JiLiuSon18}, pseudorandom isometries (PRIs)~\cite{EC:AGKL24}, and pseudorandom function-like state generators (PRFSGs)~\cite{C:AnaQiaYue22} are major quantum analogs of PRFs.
A PRU is a family $\{U_k\}_k$ of unitaries implementable in quantum polynomial-time (QPT) 
that are computationally indistinguishable from Haar random unitaries.
A PRI is a family $\{\cI_k\}_k$ of QPT implementable isometries
that are computationally indistinguishable from Haar random isometries.\footnote{Here, Haar random isometry acts as $\ket{\psi}\mapsto U(\ket{\psi}\ket{0...0})$, where $U$ is Haar random unitary. }
A PRFSG is a QPT
algorithm that, on input a classical key $k$ and a classical bit string $x$, outputs a quantum state $|\phi_k(x)\rangle$ that is computationally indistinguishable from Haar random states.
PRUs, PRIs, and PRFSGs could exist even if PRFs do not exist \cite{TQC:Kre21,KreQiaTal24}.
Moreover, PRFSGs imply various primitives and applications~\cite{C:JiLiuSon18,C:MorYam22,AC:Yan22,C:AnaQiaYue22,Ac:MorYamYam24,STOC:KhuTom24}.
\ifnum\extendedabstract=1
Moreover, they imply various primitives and applications~\cite{C:JiLiuSon18,C:MorYam22,AC:Yan22,C:AnaQiaYue22,Ac:MorYamYam24,STOC:KhuTom24}.
\fi

PRUs imply PRIs, and PRIs imply PRFSGs~\cite{TCC:AGQY22}. 
However, although they are natural quantum analogs of PRFs, it remains an open question whether the reverse implication holds.
This naturally raises the following question:
\begin{center}
    {\it{Are PRUs, PRIs, and PRFSGs equivalent?}}
\end{center}
Given their crucial roles in quantum cryptography, an important open problem is to determine whether these natural quantum analogues of PRFs are equivalent.

\subsection{Our Results}
\label{subsec:our_result}

In this paper, we partially resolve the above open problem by ruling out
black-box constructions for restricted cases.
The first result is the following:
\begin{restatable}{theorem}{BlackBox}\label{Intro_thm:main1}
     There is no black-box construction of non-adaptive and $O(\log\secp)$-ancilla PRUs from PRFSGs.
\end{restatable}

Here, a black-box construction is defined as follows~\cite{TCC:ColMut24,ChenColSat24}.
\begin{definition} [Black-Box Construction of Non-Adaptive PRUs from PRFSGs]
\label{def:BB}
    We say that non-adaptive PRUs can be constructed from PRFSGs in a black-box way
    if there exist QPT algorithms $C^{(\cdot,\cdot)}$ and $R^{(\cdot,\cdot)}$ such that
    both of the following two conditions are satisfied:
    \begin{enumerate}
        \item
        Black-box construction: For any QPT algorithm $G$ satisfying the correctness of PRFSGs\footnote{We say that a QPT algorithm satisfies the correctness of PRFSGs if it takes bit strings $k$ and $x$ as input and outputs a pure state.} and for any its unitary implementation $\tilde{G}$,\footnote{In general $G$ is a CPTP map. The CPTP map $G$ can be
        implemented by applying a unitary $\tilde{G}$ on a state and tracing out some qubits. A unitary implementation of $G$ is such a unitary $\tilde{G}$.} $C^{\tilde{G},\tilde{G}^\dag}$ satisfies the correctness of non-adaptive PRUs.\footnote{We say that a QPT algorithm satisfies the correctness of non-adaptive PRUs if it takes a classical bit string $k$ and a quantum state as input and applies a unitary on the input state.}
        \item Black-box security reduction: For any QPT algorithm $G$ satisfying the correctness of PRFSGs, any its unitary implementation $\tilde{G}$,
        any adversary $\cA$ that breaks the security of $C^{\tilde{G},\tilde{G}^\dag}$, 
        and any unitary implementation $\tilde{\cA}$ of $\cA$, 
        it holds that $R^{\tilde{\cA},\tilde{\cA}^\dag}$ breaks the security of $G$.
    \end{enumerate}
\end{definition}
In this definition, unitary implementations and their inverses are queried\footnote{In this paper, we do not consider a query to controlled-operation, transpose, and complex conjugate.}.
There are other variants of black-box constructions. For example, only unitary implementations are queried, and their inverses are not queried.
Alternatively, instead of unitary implementations, isometry implementations are queried.
\cref{def:BB} contains these variants~\cite{TCC:ColMut24,ChenColSat24}, and therefore
it captures general black-box constructions.

Non-adaptive PRUs are a weaker variant of PRUs where the adversary can query the oracle only non-adaptively.
$O(\log \secp)$-ancilla PRUs are PRUs that can be implemented in QPT using at most $O(\log \secp)$ ancilla qubits.
Because ($O(\log \secp)$-ancilla) PRUs imply non-adaptive (and $O(\log \secp)$-ancilla) PRUs, we have the following as a corollary:
\begin{corollary}
     There is no black-box construction of $O(\log \secp)$-ancilla PRUs from PRFSGs.
\end{corollary}
We also point out that
PRFSGs in \cref{Intro_thm:main1} are quantumly-accessible and adaptively-secure ones,
which is the strongest version of PRFSGs in the following sense:
Recall that a PRFSG is a QPT algorithm $G$ that, on input a classical key $k$ and a classical bit string $x$,
outputs a quantum state $|\phi_k(x)\rangle$ that is computationally indistinguishable from Haar random states.
More precisely, the computational indistinguishability means that
for any QPT adversary $\cA$,
\begin{align}
\left|\Pr_{k\gets\bit^\secp}[1\gets\cA^{G(k,\cdot)}(1^\secp)]-\Pr_{\cH}[1\gets\cA^{\cH}(1^\secp)]\right|\le\negl(\secp),    
\end{align}
where $\cH$ is the following oracle: sample a Haar random state $|\psi_x\rangle$ for each $x$ independently
in advance,
and when $x$ is queried, return $|\psi_x\rangle$.
Quantumly-accessible means that
$\cA$ can query superpositions of $x$.
Adaptively-secure means that
$\cA$ can query the oracle adaptively. 
We can define variants of PRFSGs where the queries are only non-adaptive or classical ones.
Clearly,
quantumly-accessible and adaptively-secure PRFSGs are stronger than them.

Next, we show separation for PRIs.

\begin{theorem}\label{Intro_thm:main2}
    There are no black-box constructions of non-adaptive and $O(\log \secp)$-ancilla PRUs from PRIs with $\Omega(\secp)$ stretch.
\end{theorem}

\begin{theorem}\label{Intro_thm:main3}
    There is no black-box construction of non-adaptive and $O(\log \secp)$-ancilla PRIs with $O(\log\secp)$ stretch from PRIs with $\Omega(\secp)$ stretch.
\end{theorem}

Here, $\{\cI_k\}_k$ is called PRI with $s(\secp)$ stretch if it is a family of QPT implementable isometries from $\secp$ qubits to $\secp+s(\secp)$ qubits and it is computationally indistinguishable from Haar random isometries.
The notion of black-box construction used in \cref{Intro_thm:main2,Intro_thm:main3} is defined similarly to \cref{def:BB}.  
Non-adaptive PRIs are a weaker variant of PRIs in which the adversary is restricted to making only non-adaptive oracle queries.  
$O(\log \secp)$-ancilla PRIs with $s$ stretch are PRIs with $s$ stretch that can be implemented in QPT using at most $s(\secp)+O(\log \secp)$ ancilla 
qubits.\footnote{The use of $s$ ancilla qubits is nessesarry for PRIs with $s$ stretch because it maps $\secp$ qubits to $\secp+s(\secp)$ qubits.}
Because ($O(\log \secp)$-ancilla) PRIs imply non-adaptive (and $O(\log \secp)$-ancilla) PRIs, we have the following as a corollary:
\begin{corollary}
     There is no black-box construction of ancilla-free PRUs from PRIs with $\Omega(\secp)$ stretch.
     In addition, there is no black-box construction of $O(\log \secp)$-ancilla PRIs with $O(\log \secp)$ stretch from PRIs with $\Omega(\secp)$ stretch.
\end{corollary}

Our main results are \cref{Intro_thm:main1,Intro_thm:main2,Intro_thm:main3}, but they are derived from the following 
technical results:
\ifnum\extendedabstract=0
\begin{theorem}[\cref{thm:main}, Informal]\label{Intro_thm:PRU_vs_PRFSG}
     There exists a unitary oracle $\cO$ such that
     PRFSGs exist but non-adaptive and $O(\log \secp)$-ancilla PRUs do not exist
     relative to $\cO$ and $\cO^\dagger$.
\end{theorem}

\begin{theorem}[\cref{thm:PRI_vs_PRFSG}, Informal]\label{Intro_thm:PRI_vs_PRFSG}
     There exists a unitary oracle $\cO$ such that
     PRFSGs exist but non-adaptive and $O(\log \secp)$-ancilla PRIs with $O(\log\secp)$ stretch do not exist
     relative to $\cO$ and $\cO^\dagger$.
\end{theorem}

\begin{theorem}[\cref{thm:PRI_vs_PRI}, Informal]\label{Intro_thm:PRI_vs_PRI}
     There exists a unitary oracle $\cO$ such that
     PRIs with $\Omega(\secp)$ stretch exist but non-adaptive and $O(\log \secp)$-ancilla PRIs with $O(\log\secp)$ stretch do not exist
     relative to $\cO$ and $\cO^\dagger$.
\end{theorem}
\fi
\ifnum\extendedabstract=1
\begin{theorem}[Informal]\label{Intro_thm:PRU_vs_PRFSG}
     There exists a unitary oracle $\cO$ such that
     PRFSGs exist but non-adaptive and ancilla-free PRUs do not exist
     relative to $\cO$ and $\cO^\dagger$.
\end{theorem}

\begin{theorem}[Informal]\label{Intro_thm:PRI_vs_PRFSG}
     There exists a unitary oracle $\cO$ such that
     PRFSGs exist but non-adaptive and ancilla-free PRIs with $O(\log\secp)$ stretch do not exist
     relative to $\cO$ and $\cO^\dagger$.
\end{theorem}

\begin{theorem}[Informal]\label{Intro_thm:PRI_vs_PRI}
     There exists a unitary oracle $\cO$ such that
     PRIs with $\Omega(\secp)$ stretch exist but non-adaptive and ancilla-free PRIs with $O(\log\secp)$ stretch do not exist
     relative to $\cO$ and $\cO^\dagger$.
\end{theorem}
\fi
\ifnum\extendedabstract=0
Here the existence and non-existence of primitives relative to oracles mean the following.
\begin{definition}
\label{def:oraclesepa}
Let $\cO$ be a unitary oracle.
We say that a primitive exists relative to $\cO$ and $\cO^\dagger$   
if there exists a QPT algorithm $C^{(\cdot,\cdot)}$ such that both of the following two conditions are satisfied:
\begin{itemize}
    \item 
    $C^{\cO,\cO^\dagger}$ satisfies the correctness of the primitive.
    \item
    $C^{\cO,\cO^\dagger}$ satisfies the security of the primitive against any QPT adversary $\cA^{\cO,\cO^\dagger}$
    that can query $\cO$ and $\cO^\dagger$.
\end{itemize}
\end{definition}
\fi

\ifnum\extendedabstract=0
In \cref{sec:black-box_construction}, we will explain that this oracle separation implies the impossibility of the
black-box construction.
A high-level overview of our proofs of \cref{Intro_thm:PRU_vs_PRFSG,Intro_thm:PRI_vs_PRFSG,Intro_thm:PRI_vs_PRI} 
will be explained 
in \cref{sec:technicaloverview}.
\fi

\ifnum\extendedabstract=0
In the above definition, \cref{def:oraclesepa}, the qurey to both $\cO$ and $\cO^\dagger$ are allowed.
If the query to only $\cO$ is allowed, and that to $\cO^\dagger$ is not allowed,
what we can rule out is not the black-box construction of \cref{def:BB},
but a more restricted one where $C$ queries only $\tilde{G}$ and $R$ queries only $\tilde{\cA}$. 
Because our \cref{Intro_thm:PRU_vs_PRFSG} shows the oracle separation in terms of \cref{def:oraclesepa},
we can exclude the general black-box constructions,
which is an important advantage of our results\footnote{However, our oracle separations do not rule out the case when $C$ queries $\widetilde{G}^\top,\bar{\widetilde{G}}$, or $R$ queries $\cA^\top,\bar{\cA}$, where $(\cdot)^\top$ denotes the transpose, and $\bar{(\cdot)}$ denotes the complex conjugate.
This is because queries to $\cO^\top$ or $\bar{\cO}$ are not allowed in \cref{Intro_thm:PRU_vs_PRFSG,Intro_thm:PRI_vs_PRFSG,Intro_thm:PRI_vs_PRI}.
We expect that our separations can be extended to the case when the query to $\cO^\top$ and $\bar{\cO}$ are allowed using a similar technique in \cite{C:Zhandry25}.}.
\fi

\cref{Intro_thm:PRU_vs_PRFSG,Intro_thm:PRI_vs_PRFSG,Intro_thm:PRI_vs_PRI} also indicate that
all primitives that are known to be implied by PRFSGs or PRIs 
(such as PRSGs, private-key quantum money, OWSGs, OWPuzzs, EFI pairs, SKE, commitments, MAC, etc.)
also exist relative to $\cO$ and $\cO^\dagger$.
However, some caution is needed for
the existence of IND-CPA SKE with quantum ciphertexts and EUF-CMA MAC (with unclonable tags), 
because known constructions of these primitives from PRFSGs \cite{C:AnaQiaYue22} query the inverse of a unitary implementation of PRFSGs.
One advantage of our result, \cref{Intro_thm:PRU_vs_PRFSG}, is that the security of PRFSGs holds against adversaries that query not only $\cO$ but also $\cO^\dagger$.
Because of this advantage, the known constructions of
IND-CPA SKE with quantum ciphertexts and EUF-CMA MAC (with unclonable tags) from PRFSGs
automatically imply their existence relative to $\cO$ and $\cO^\dagger$. 
\ifnum\extendedabstract=0
\ifnum\submission=1
(For details, see \cref{subsec:many_primitives}.)
\fi
\ifnum\submission=0
(For details, see \cref{subsec:many_primitives}.)
\fi
\fi

\ifnum\extendedabstract=0
Several new ideas and techniques are used to show \cref{Intro_thm:PRU_vs_PRFSG}, many of which are of independent interest
and should be useful for other applications in quantum cryptography. 
In particular, for the oracle separation, we construct a direct attack to PRUs.
To the best of our knowledge, this is the first time that a direct attack to PRUs has been constructed.
All previous results that break PRUs first reduced PRUs to PRSGs and then broke PRSGs \cite{TQC:Kre21,TCC:AGQY22} by using the shadow tomography \cite{Shadow2,NatPhys:HKP20}.
\fi
\ifnum\extendedabstract=0
As we will explain later, our key idea for constructing the direct attack to PRUs is to leverage the quantum singular-value transformation (QSVT) \cite{STOC:GSLW19}.
\fi
\ifnum\extendedabstract=1
Our key idea for constructing the direct attack to PRUs is to leverage the quantum singular-value transformation (QSVT) \cite{STOC:GSLW19}.
\fi
To our knowledge, this is the first time that QSVT
has been used to separate quantum cryptographic primitives. 
We believe that QSVT should be useful for other applications in quantum cryptography.

\ifnum\submission=0
Finally, our results are summarized in \cref{fig}.

\usetikzlibrary{positioning} 
\usetikzlibrary{calc} 
\usetikzlibrary {quotes}
\tikzset{>=latex} 

\tikzstyle{normalimply}=[->,black,line width=1.6]
\tikzstyle{curveimply}=[->,black,line width=1.6,bend right]
\tikzstyle{curveseparation}=[->,black,dashed,line width=1.6, bend right]
\tikzstyle{thiswork}=[->,red,dashed,line width=1.6]
\begin{figure}[h]
\begin{center}
    \begin{tikzpicture}[scale=0.9,every edge quotes/.style = {font=\footnotesize,fill=white}]
      \def\h{-2.0} 
      \def\w{2.6} 

       \node[] (PRF) at (4*\w,0.1*\h) {PRFs};
       \node[] (PRFSG) at (4*\w,5*\h) {PRFSGs};
       \node[] (PRU) at (4*\w,1.3*\h) {PRUs};
       \node[] (shortPRI) at (4*\w,2.5*\h) {PRIs with $O(\log\secp)$ stretch};
       \node[] (longPRI) at (4*\w,4*\h) {PRIs with $\Omega(\secp)$ stretch};
       \node[] (many) at (4*\w,6*\h) {PRSGs, SKE, unclonable MAC, UPSGs, private money, OWSGs, OWpuzzs, EFI, etc.};
        \node[] (non-adaptivePRU) at (2.5*\w,2*\h) {Non-adaptive PRUs};
        \node[] (non-adaptivePRI) at (6.6*\w,3*\h) {Non-adaptive PRIs with $O(\log\secp)$ stretch};

        \draw[curveimply] (PRF) edge["\cite{MaHsi24}" left] 
        (PRU);
        \draw[curveseparation](PRU)
        edge["\cite{TQC:Kre21,KreQiaTal24}" right]
        (PRF);
        \draw[thiswork](longPRI)
        edge["\text{\cref{Intro_thm:main3}}"]
        (non-adaptivePRI);
        \draw[thiswork](longPRI)
        edge["\text{\cref{Intro_thm:main2}}"]
        (non-adaptivePRU);
        \draw[thiswork](PRFSG)
        edge["\text{\cref{Intro_thm:main1}}"]
        (non-adaptivePRI);
        \draw[normalimply] (PRFSG) edge[] (many);
        \draw[normalimply=black] (PRU) edge["\text{Trivial}"] (non-adaptivePRU);
        \draw[normalimply=black] (PRU) edge["\text{Trivial}"] (shortPRI);
        \draw[normalimply=black] (shortPRI) edge["\text{Trivial}"] (non-adaptivePRI);
        \draw[normalimply=black] (shortPRI) edge["\text{Trivial}"] (longPRI);
        \draw[normalimply] (longPRI) edge["\cite{TCC:AGQY22}"] (PRFSG);
    \end{tikzpicture}
\end{center}
\caption{A summary of our results and known results.  
An arrow from primitive A to primitive B indicates that A implies B.  
A dashed arrow from A to B indicates that there is no black-box construction from A to B.  
A red dashed arrow from A to B indicates that there is no black-box construction from A to $O(\log \secp)$-ancilla B.
SKE means IND-CPA SKE with quantum ciphertexts. Unclonable MAC means EUF-CMA MAC with unclonable tags. Private money means privet-key quantum money schemes.}
\label{fig}
\end{figure}
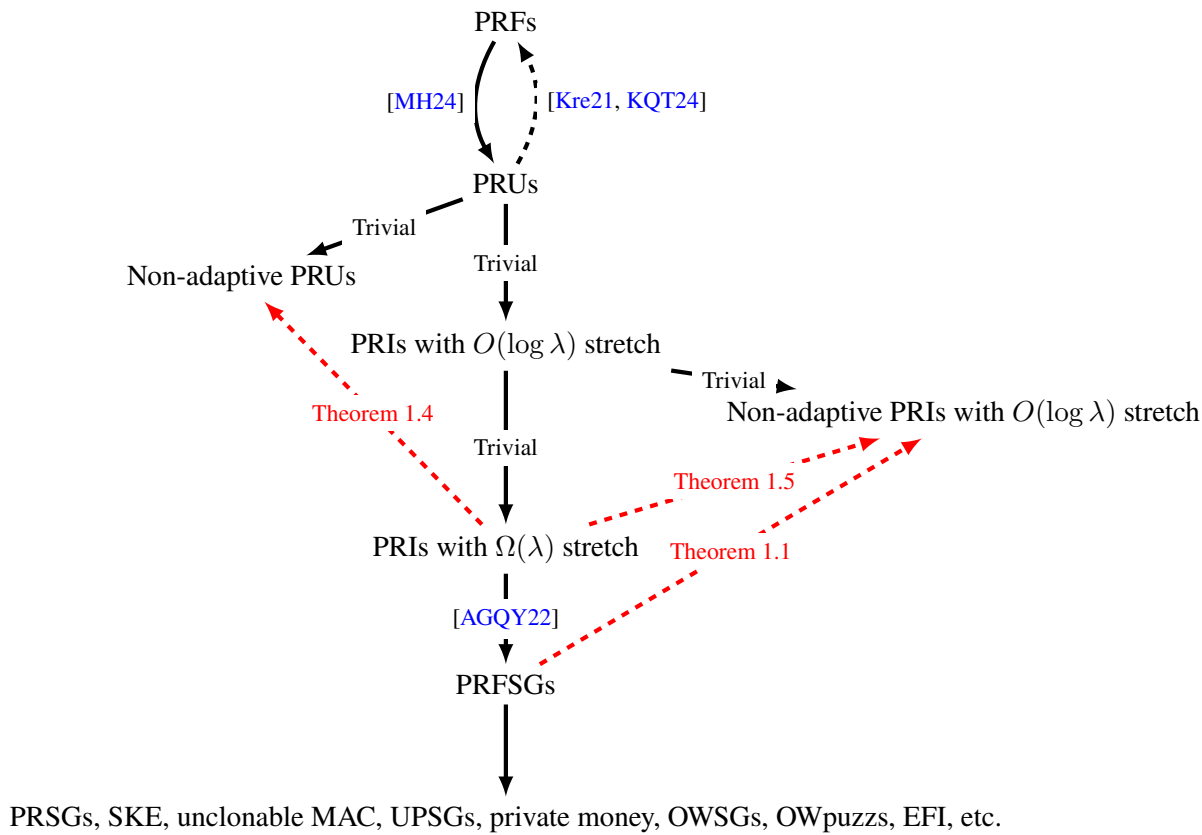
\fi
\ifnum\submission=1
In the supplemental material, we give a figure summarizing our results in \cref{fig}.
\fi

\ifnum\extendedabstract=0
\subsection{Technical Overview}
\label{sec:technicaloverview}
As we have mentioned, our main results \cref{Intro_thm:main1,Intro_thm:main2,Intro_thm:main3} are obtained from the technical result
\cref{Intro_thm:PRU_vs_PRFSG,Intro_thm:PRI_vs_PRFSG,Intro_thm:PRI_vs_PRI}.
In this subsection, we will overview the high-level idea of the proof of \cref{Intro_thm:PRU_vs_PRFSG}.
\cref{Intro_thm:PRI_vs_PRFSG,Intro_thm:PRI_vs_PRI} can be shown in a similar idea.

Our goal is to construct a unitary oracle $\cO$
such that
PRFSGs exist but PRUs do not
relative to $\cO$ and $\cO^\dagger$.
Our oracle $\cO$ consists of two oracles $\cS$ and $\cU$.
$\cS$ is used to construct PRFSGs. 
$\cU$ is an oracle that solves $\unitaryPSPACE$-complete problem~\cite{ITCS:RosYue22,BosEfrMetPorQiaYue23}, which is
used to break PRUs.

\paragraph{Constructing PRFSGs.}
The oracle $\cS\coloneqq\{\cS_{k,x}\}_{k,x}$ is a set of unitary oracles $\cS_{k,x}$. 
Each $\cS_{k,x}$ works as follows:
\begin{itemize}
\item 
Sample a Haar random state $|\psi_{k,x}\rangle$.
    \item 
    If the input is $|0\rangle|0...0\rangle$, return $|1\rangle|\psi_{k,x}\rangle$.
    \item 
    If the input is $|1\rangle|\psi_{k,x}\rangle$, return $|0\rangle|0...0\rangle$.
    \item 
    For other inputs, do nothing.
\end{itemize}
In other words, the oracle $\cS_{k,x}$ ``swaps'' $|0\rangle|0...0\rangle$ and $|1\rangle|\psi_{k,x}\rangle$.
Similar oracles were considered in \cite{BMMMY24,BosChenNeh24,ChenColSat24}.
With this $\cS$, we construct a PRFSG $G^{\cS}$ as follows:
\begin{enumerate}
    \item 
    On input $(k,x)$, query $|0\rangle|0...0\rangle$ to $\cS_{k,x}$ in order to get $|1\rangle|\psi_{k,x}\rangle$.
    \item 
    Output $|\psi_{k,x}\rangle$.
\end{enumerate}
Intuitively, it is clear that $G^{\cS}$ satisfies the security of PRFSGs, because
its output is Haar random states.
In fact, we can show the security of $G^{\cS}$ based on the proof template of \cite{TQC:Kre21}.
However, the template cannot be directly used in this case, and
some careful re-investigations are required. 
First, define
\begin{align}
\mathsf{Adv}(\cS)\coloneqq
\Pr_{k}[1\gets\cA^{\cS_k,\cS}]-\Pr_{\{|\vartheta_x\rangle\}_x\gets \sigma}[1\gets\cA^{\cT,\cS}],
\label{shimeshitai}
\end{align}
where $\cS_k\coloneqq\{\cS_{k,x}\}_x$, $\sigma$ is a Haar random state measure and
$\cT$ is the following oracle.
\begin{enumerate}
    \item 
    For each $x$, sample a Haar random state $|\vartheta_x\rangle$ independently in advance.
    \item 
    When $x$ is queried, it returns $|\vartheta_x\rangle$.
\end{enumerate}
This $\mathsf{Adv}(\cS)$ is the advantage of the adversary $\cA$ in the security game of PRFSGs when $\cS$ is chosen.
(Actually, we have to consider not $\cA^{\cS}$ but $\cA^{\cS,\cS^\dagger,\cU,\cU^\dagger}$. However,
because $\cS^\dagger=\cS$ and $\cU^\dagger$ can be simulated by querying $\cU$, we have only to consider
$\cA^{\cS,\cU}$. Moreover, our argument below is about unbounded $\cA$, and therefore we can also ignore $\cU$.) 
Our goal is to show that $|\mathsf{Adv}(\cS)|$ is small with high probability over $\cS$.
To that end, we first show
\begin{align}
|\Exp_\cS\mathsf{Adv}(\cS)|\le\negl(\secp).
\label{shimeshitai2}
\end{align}
This is shown by using the technique of \cite{TQC:Kre21}, which is 
based on the BBBV theorem~\cite{SICOMP:BBBV97}.
However, the technique of \cite{TQC:Kre21} cannot be directly used because in his case
each $\cS_k$ is a Haar random unitary, while in our case not. Fortunately, we can confirm that his proof also holds
for our $\cS_k$, and therefore we can show \cref{shimeshitai2} in a similar way.
From \cref{shimeshitai2}, we want to show our goal that $|\mathsf{Adv}(\cS)|$ is small with high probablity over $\cS$
via the following concentration inequality:\footnote{$\mathsf{Adv}(\cS)$ is a function of unitary,
because each Haar random state in $\cS$ can be replaced with a Haar random unitary applied on $|0...0\rangle$.}
\begin{align}
\Pr_{\cS}[|\mathsf{Adv}(\cS)-\Exp_{\cS'}\mathsf{Adv}(\cS')|\ge \delta]\le e^{-O(\delta/L^2)}.    
\end{align}
This concentration inequality holds if
$\mathsf{Adv}(\cS)$ is $L$-Lipshitz. 
By a straightforward calculation, we confirm it.
In this way, we can show that our constructed PRFSG $G^\cS$ is secure against $\cA^{\cS}$.

\paragraph{Breaking PRUs.}
Here we describe how to break ancilla-free PRUs, since the argument generalizes to $O(\log \secp)$-ancilla PRUs with postselection.
Let $\cS$ be the unitary oracle introduced above.
Let $\cU$ be the $\unitaryPSPACE$-complete oracle.
Let $F^{\cS,\cU}$ be a QPT algorithm that, on input $k$ and a state,
applies a unitary $U_k$ on the state. (Actually, we have to consider $F^{\cS,\cS^\dagger,\cU,\cU^\dagger}$, but as we have mentioned above, we can ignore $\cS^\dagger$ and $\cU^\dagger$.)
Our goal is to show that $F^{\cS,\cU}$ cannot satisfy the security of PRUs. In other words, we construct a QPT adversary 
$\cA^{\cS,\cU}$ that distinguishes the query to $F^{\cS,\cU}$ and that to the oracle that applies Haar random unitaries.

To show it, we first define
two states 
\begin{align}
\rho_0\coloneqq\Exp_{U\gets\mu}(U^{\otimes \ell}\otimes I)|\Phi\rangle\langle\Phi|(U^{\otimes \ell}\otimes I)^\dag,
\end{align}
and
\begin{align}
\rho_1\coloneqq\Exp_{k\gets\cK_\secp}(U_k^{\otimes \ell}\otimes I)|\Phi\rangle\langle\Phi|(U_k^{\otimes \ell}\otimes I)^\dag.
\end{align}
Here, $\mu$ is the Haar measure over $\secp$-qubit unitaries,
$U_k$ is a $\secp$-qubit unitary,
$\ell\coloneqq\lceil\log|\cK_\secp|\rceil$, $\cK_\secp$ is the key space,
$|\Phi\rangle\coloneqq\frac{1}{\sqrt{2^{\ell\secp}}}\sum_{x\in\bit^{\ell\secp}}|x\rangle|x\rangle$ is the $\ell\secp$-qubit maximally entangled state.
The adversary can generate $\rho_0$ 
if it queries the Haar random oracle, while
it can generate $\rho_1$ if it queries $F^{\cS,\cU}$. 
Therefore, if $\rho_0$ and $\rho_1$ can be distinguished 
by using the $\unitaryPSPACE$-complete oracle $\cU$, the adversary can break the PRU.

The question is therefore how to distinguish $\rho_0$ and $\rho_1$ by using
$\cU$?
There is one issue here.
Each $U_k$ can depend on $\cS$, because
the PRU generator $F^{\cS,\cU}$ can query $\cS$, while the $\unitaryPSPACE$-complete oracle $\cU$ (that $\cA$ queries) is independent of $\cS$. 
When $\cS$ acts on small number of qubits, $\cA$ can get its classical information by
querying $\cA$'s $\cS$ many times and doing the process tomography~\cite{FOCS:HKOT23}, and can send the classical information to $\cU$ as is done in \cite{TQC:Kre21}.
However, when $\cS$ acts on large number of qubits, this strategy does not work,
because the process tomography is no longer efficient.
Our key observation is that when $\cS$ acts on large number of qubits,
it almost does not cause any effect, because $\cS$, which swaps only $|0...0\rangle$ and Haar random states,
is almost the identity operation.\footnote{This step of our proof requires the ancilla-free condition. As mentioned earlier, we can relax this condition to the \( O(\log \secp) \)-ancilla condition by using postselection. Whether the same result can be established without this assumption remains an open question.}

Hence, we introduce $\{U_k'\}_k$ that is the same as $\{U_k\}_k$ except that
$\cS$ acting on small number of qubits is simulated by the classical information obtained
via the process tomography. Then, if we define
another state
\begin{align}
\rho_2\coloneqq\Exp_{k\gets\cK_\secp}((U_k')^{\otimes \ell}\otimes I)|\Phi\rangle\langle\Phi|((U_k')^{\otimes \ell}\otimes I)^\dag,
\end{align}
our goal is to distinguish $\rho_0$ and $\rho_2$
by using $\cU$ since $\rho_1$ is statistically close to $\rho_2$.
Let $Q$ be the projection onto the support of $\rho_2$.
Clearly, $\Tr[Q\rho_2]=1$.
On the other hand, as we will explain later, $\Tr[Q\rho_0]$ is negligible.
Therefore, if we can implement $\cQ$, we can distinguish $\rho_0$ and $\rho_2$.
The fact that $\Tr[Q\rho_0]$ is negligible can be shown from the following lemma.
\begin{lemma}[\cref{lem:Haar_Choi_has_negligible_overlap}, Informal]\label{Intro_lem:design_needs_large_support} 
    For the above $Q$ and $\rho_0$,
    \begin{align}
        \Tr[Q\rho_0]
        \le 
        \negl(\secp)
        \label{eq:design_needs_large_support}.
    \end{align}
\end{lemma}

How can we implement $Q$ with $\cU$?
Our novel idea is to use the singular-value discrimination (SVD) algorithm, which is a concrete example of 
quantum singular-value transformation (QSVT) \cite{STOC:GSLW19}.
Let $M$ be any positive matrix such that there exists a unitary $V$ that satisfies
$M=(\bra{0...0}\otimes I)V(\ket{0...0}\otimes I)$.
Such an encoding of a matrix to a unitary is called a block encoding~\cite{STOC:GSLW19}.
The SVD algorithm, which can query $V$, can solve the following promise problem:
Given a single-copy quantum state $\xi$, decide
whether the support of $\xi$ is in the support of $M$ or 
the support of $\xi$ is orthogonal to that of $M$.
Because our goal is to implement $Q$, which decides whether a given state is in the support of $\rho_2$,
we have only to take $M=\rho_2$.
It is known that a block encoding of any quantum state $\xi$ can be efficiently
implementable by using a unitary $W$ satisfying that $W\ket{0...0}$ is a purification of $\xi$~\cite{ICALP:vApGil19}. 
Thus if we take $M=\rho_2$, we can efficiently implement its block encoding.
Then if we run the SVD algorithm on input $\rho_0$ or $\rho_2$, we can distinguish $\rho_0$ and $\rho_2$.
The SVD algorithm can be realized in quantum polynomial space~\cite{STOC:GSLW19},
and therefore $Q$ can be realized by querying $\cU$.

Recall that our original goal is to distinguish $\rho_0$ from $\rho_1$.
Since $\rho_1$ is statistically close to $\rho_2$, the above algorithm can also distinguish $\rho_0$ from $\rho_1$.
In summary, therefore, a QPT adversary $\cA^{\cS,\cU}$ 
can break the PRU $F^{\cS,\cU}$.

\fi
\ifnum\submission=0
\subsection{Related Works}

\paragraph{Comparison with the concurrent work \cite{BHMV25}.}
The concurrent work by \cite{BHMV25} also separates non-adaptive and ancilla-free PRUs from PRFSGs. 
Their separation oracle, which they call unitary common Haar function-like state (CHFS) oracle, is the same as our separation oracle.
However, they use a different technique to break non-adaptive and ancilla-free PRUs. 
\fi

\ifnum\submission=0
\paragraph{Comparison with the previous works.}
There are several works that separate Microcrypt primitives.
In addition to the works we have already mentioned~\cite{TQC:Kre21,ChenColSat24,BosChenNeh24,BMMMY24},
there are other three papers.
\cite{TCC:ColMut24} separated multi-query secure quantum digital signatures from PRUs. 
\cite{TCC:AnaGulLin24} separated quantum computation classical communication (QCCC) primitives from PRFSGs.
\cite{GMMY24} constructed oracles such that QCCC primitives exist but 
$\mathbf{BQP}=\mathbf{QCMA}$, and quantum lightning \cite{EC:Zhandry19b} exist but $\mathbf{BQP}=\mathbf{QMA}$ relative to the oracles.
These separation results are incomparable with our work.
\fi
\ifnum\submission=1
\paragraph{Comparison with the previous works.}
There are several works that separate Microcrypt primitives~\cite{TQC:Kre21,ChenColSat24,BosChenNeh24,BMMMY24,TCC:ColMut24,TCC:AnaGulLin24,GMMY24,C:GolZha25}.
These separation results are incomparable with our work.
\fi

\section{Preliminaries}\label{sec:preliminaries}

\subsection{Basic Notations}
\label{sec:basic_notations}
This paper uses the standard notations of quantum computing and cryptography. 
For bit strings $x$ and $y$, $(x,y)$ denotes their concatenation.
We use $\secp$ as the security parameter.
$[n]$ means the set $\{1,2,...,n\}$.
For any set $S$, $x\gets S$ means that an element $x$ is sampled uniformly at random from the set $S$.
We write $\negl$ as a negligible function
and $\poly$ as a polynomial. 
QPT stands for quantum polynomial time.
For an algorithm $\cA$, $y\gets \cA(x)$ means that the algorithm $\cA$ outputs $y$ on input $x$.

We use $I\coloneqq\ketbra{0}{0}+\ketbra{1}{1}$ as the identity on a single qubit.
For the notational simplicity, we sometimes write $I^{\otimes n}$ just as $I$ when the dimension is clear from the context. 
For a vector $\ket{\psi}$, we define its norm as $\|\ket{\psi}\|\coloneqq\sqrt{\braket{\psi|\psi}}$.
For any matrix $A$, we define the $p$-norm $\|A\|_p\coloneqq(\Tr[(A^\dag A)^{p/2}])^{1/p}$.
In particular, we call it the trace norm when $p=1$ and the Frobenious norm when $p=2$.
For any matrix $A$, the operator norm $\|\cdot\|_\infty$ is defined as $\|A\|_\infty\coloneqq\max_{\ket{\psi}}\sqrt{\bra{\psi}A^\dag A\ket{\psi}}$, where the maximization is taken over all pure states $\ket{\psi}$.
$\identitymap$ denotes the identity channel, i.e., $\identitymap(\rho)=\rho$ for any state $\rho$.
For two channels $\cE$ and $\cF$ that take $d$ dimensional states as inputs, we say $\|\cE-\cF\|_\diamond\coloneqq\max_{\ket{\psi}}\|(\mathrm{id}\otimes\cE)(\ket{\psi}\bra{\psi})-(\mathrm{id}\otimes\cF)(\ket{\psi}\bra{\psi})\|_1$ is the diamond norm between $\cE$ and $\cF$, where the maximization is taken over all $d^2$ dimensional pure states.

$\Linear(d)$ denotes the set of all $d\times d$ matrices.
The set (or group) of $d$-dimensional unitary matrices and states are denoted by $\Unitaries(d)$ and $\States(d)$, respectively.
$\mu_d$ and $\sigma_d$ denote the Haar measure over $\Unitaries(d)$ and $\States(d)$, respectively.
For $U\in\Unitaries(d)$, c-$U\coloneqq\ketbra{0}{0}\otimes I+\ketbra{1}{1}\otimes U$ is the controlled-$U$.
$\ket{\Omega_d}\coloneqq\frac{1}{\sqrt{d}}\sum_{x\in[d]}\ket{x}\ket{x}$ is the maximally entangled state. 
For $U\in\Unitaries(d)$, $U(\cdot)U^\dag$ denotes the channel that maps $\rho\mapsto U\rho U^\dag$ for all $d$-dimensional state $\rho$. 

\subsection{Useful Facts}
Here we introduce several useful facts that we will use.
\begin{lemma}[Gentle Measurement Lemma \cite{winter1999coding,watrous2018theory}]\label{lem:gentle_measurement}
    Let $\rho$ be a quantum state, and $0\le\epsilon\le1$. Let $M$ be a matrix such that $0\le M\le I$ and
    \begin{align}
        \Tr[M\rho]\ge1-\epsilon.
    \end{align}
    Then,
    \begin{align}
        \bigg\|\rho-\frac{\sqrt{M}\rho\sqrt{M}}{\Tr[M\rho]}\bigg\|_1\le \sqrt{\epsilon}.
    \end{align}
\end{lemma}

The following is a well-known fact. For its proof, see \cite{MatrixAnalysis}.

\begin{lemma}[H\"{o}lder's Inequality \cite{MatrixAnalysis}]\label{lem:Holder}
    Let $A$ and $B$ be matrices of the same size. Then, for any $p>0$, we have $\|AB\|_p\le\|A\|_p\|B\|_\infty.$
    Moreover, we have $\|AB\|_\infty\le\|A\|_\infty\|B\|_\infty$
\end{lemma}

\ifnum\submission=0
We need the following lemma when we use the concentration inequality which we introduce later.
\begin{lemma}[Lemma 28 in \cite{TQC:Kre21}]\label{lem:Lipschitz}
    Let $\cA^{U}$ be a quantum algorithm that makes $T$ queries to $U\in\Unitaries(d)$ and its inverse. Then, $f(U)=\Pr[1\gets\cA^{U}]$ is $2T$-Lipschitz in the Frobenius norm, i.e., $|f(U)-f(V)|\le2T\|U-V\|_2$ for all $U,V\in\Unitaries(d)$.
\end{lemma}

By applying the triangle inequality, we have the following.

\begin{lemma}\label{lem:f+g_Lipschitz}
    Let $f,g:\Unitaries(d)\to\R$ be $L$-Lipschitz functions in the Frobenious norm, i.e.,
    $|f(U)-f(V)|\le L\|U-V\|_2$ and $|g(U)-g(V)|\le L\|U-V\|_2$ for any $U,V\in\Unitaries(d)$. Then, $f+g$ is $2L$-Lipschitz. Namely, for any $U,V\in\Unitaries(d)$,
    \begin{align}
        |f(U)+g(U)-f(V)-g(V)|\le2L\|U-V\|_2.
    \end{align}
\end{lemma}

Regarding the Frobenius norm, we use the following.

\begin{lemma}\label{lem:Uket0-Vket0_is_less_than_U-V}
    Let $U,V\in\Unitaries(d)$. Then, for any $\ket{\psi}\in\States(d)$,
    \begin{align}
        \|U\ketbra{\psi}{\psi}U^\dag-V\ketbra{\psi}{\psi}V^\dag\|_2
        \le2\|U-V\|_2
    \end{align}
\end{lemma}

\begin{proof}[Proof of \cref{lem:Uket0-Vket0_is_less_than_U-V}]
    We obtain the inequality as follows:
    \begin{align}
        &\|U\ketbra{\psi}{\psi}U^\dag-V\ketbra{\psi}{\psi}V^\dag\|_2
        \notag\\
        \le&\|U\ketbra{\psi}{\psi}U^\dag-V\ketbra{\psi}{\psi}U^\dag\|_2+\|V\ketbra{\psi}{\psi}U^\dag-V\ketbra{\psi}{\psi}V^\dag\|_2
        \tag{By the triangle inequality}\\
        =&\|(U-V)\ketbra{\psi}{\psi}U^\dag\|_2+\|V\ketbra{\psi}{\psi}(U-V)^\dag\|_2
        \notag\\
        =&\|(U-V)\ketbra{\psi}{\psi}U^\dag\|_2+\|(U-V)\ketbra{\psi}{\psi}V^\dag\|_2
        \tag{By $\|A^\dag\|_2=\|A\|_2$}\\
        \le&\|U-V\|_2\|\ketbra{\psi}{\psi}U^\dag\|_\infty
        +\|(U-V)\|_2\|\ketbra{\psi}{\psi}V^\dag\|_\infty
        \tag{By H\"older's inequality, \cref{lem:Holder}}\\
        \le&2\|U-V\|_2,
    \end{align}
    where in the last inequality we have used $\|\ketbra{\psi}{\psi}W\|_\infty\le\|\ketbra{\psi}{\psi}\|_\infty\|W\|_\infty\le1$ for any $W\in\Unitaries(d)$ from \cref{lem:Holder}.
\end{proof}

The following lemma is implicitly shown in \cite{TQC:Kre21}.\footnote{In \cite{TQC:Kre21}, he showed \cref{lem:BBBV} only for the case when $\cD$ is the Haar measure. However, his proof does not rely on any property of the Haar measure because its proof essentially depends on the BBBV theorem \cite{SICOMP:BBBV97}. Thus, we can obtain the same claim for a general distribution $\cD$ from his original proof.}

\begin{lemma}[Lemma 31 in \cite{TQC:Kre21}]\label{lem:BBBV}
    Let $\cD$ be a distribution over $\Unitaries(d)$. Suppose that $\cA$ is a quantum algorithm that queries $U=(U_1,..,U_N)\in\Unitaries(d)^N$ and $O\in\Unitaries(d)$. We see $U$ as $\sum_{n\in[N]}\ketbra{n}{n}\otimes U_n$. For fixed $U\in\Unitaries(d)^N$, define
    \begin{align}
        \Adv(\cA,U)\coloneqq
        \Pr_{k\gets[N]}[1\gets\cA^{U_k,U}]-\Pr_{O\gets\cD}[\cA^{O,U}].
    \end{align}
    Then, there exists a constant $c>0$ such that, for any $T$-query algorithm $\cA$,
    \begin{align}
        \bigg|
         \Exp_{U\gets\cD^N}[\Adv(\cA,U)]
        \bigg|
        \le\frac{cT^2}{N}.
    \end{align}
    Here, $U\gets\cD^N$ denotes that each $U_k$ is independently sampled from $\cD$.
\end{lemma}
\fi

We will use the following process tomography algorithm.

\begin{theorem}[\cite{FOCS:HKOT23}]\label{thm:process_tomography_HKOT23}
    There exists a quantum algorithm $\cA$ that, given a black-box access to $Z\in\Unitaries(D)$, satisfies the following:
    \begin{itemize}
        \item Accuracy: On input $\epsilon,\mu\in(0,1)$, $\cA$ outputs a classical description of a unitary $Z'$ such that
        \begin{align}
            \Pr_{Z'\gets\cA}[\|Z(\cdot)Z^\dag-Z'(\cdot)Z'^\dag\|_\diamond\le\epsilon]\ge1-\eta.
        \end{align}
        \item Query complexity: $\cA$ makes $O(\frac{D^2}{\epsilon}\log\frac{1}{\eta})$ queries to $Z$.
        \item Time complexity: The time complexity of $\cA$ is $\poly(D,\frac{1}{\epsilon},\log\frac{1}{\eta})$.
    \end{itemize}
\end{theorem}

\begin{lemma}[Borel-Cantelli]\label{lem:Borel-Cantelli}
    Suppose that $\{X_n\}_{n\in\N}$ is a sequence of random variables such that $X_n\in\bit$. If $\sum_{n\in\N}\Exp[X_n]$ is finite, then
    \begin{align}
        \Pr\bigg[\sum_{n\in\N}X_n=\infty\bigg]=0.
    \end{align}
\end{lemma}

\subsection{The Haar Measure}

We use the properties of the Haar measure. Recall that $\mu_d$ is the Haar measure over $\Unitaries(d)$, and $\sigma_d$ is that over $\States(d)$. 
\ifnum\submission=0
We will use the following concentration property.
\begin{theorem}[Theorem 5.17 in \cite{Mec}]\label{thm:Haar_concentration}
     Given $d_1,\ldots,d_k \in \N$, let $X = \Unitaries(d_1)\times\cdots\times\Unitaries(d_k)$. Let $\mu = \mu_{d_1} \times \cdots \times \mu_{d_k}$ be the product of Haar measures on $X$. Suppose that $f: X \to \R$ is $L$-Lipschitz in the $\ell^2$-sum of Frobenius norm, i.e., for any $U=(U_1,...,U_k)\in X$ and $V=(V_1,...,V_k)\in X$, 
     we have $|f(U)-f(V)|\le L\sqrt{\sum_i\|U_i-V_i\|^2_2}$. Then for every $\delta > 0$,
    \begin{align}
       \Pr_{U\gets \mu}\left[f(U)\ge\Exp_{V \gets \mu}[f(V)]+\delta\right] \le \exp\left(-\frac{(d-2)\delta^2}{24L^2}\right),
    \end{align}
    where $d \coloneqq \min\{d_1,\ldots,d_k\}$.
\end{theorem}
\fi

The following is a well-known fact \cite{Harrow13,Quantum:Mele24}.

\begin{lemma}\label{lem:Haar_states}
    Let $\ell,d\in\N$. Then,
    \begin{align}
        \Exp_{\ket{\psi}\gets\sigma_d}\ketbra{\psi}{\psi}^{\otimes\ell}=\frac{\Pi_\symetric}{\binom{d+\ell-1}{\ell}}.
    \end{align}
    Here, $\Pi_\symetric$ is the following projection:
    \begin{align}
        \Pi_\symetric=\frac{1}{\ell!}\sum_{\pi\in S_\ell}R_{\pi},
    \end{align}
    where $S_\ell$ denotes the permutation group over $\ell$ elements, and $R_\pi$ is permutation unitary such that $R_\pi\ket{x_1,...,x_k}=\ket{x_{\pi^{-1}(1)},...,x_{\pi^{-1}(k)}}$ for all $x_1,...,x_k\in[d]$ and $\pi\in S_\ell$.
\end{lemma}
We often use the expectation of unitaries' action over some distribution. Thus, we define the following for the notational simplicity.

\begin{definition}\label{def:map_wrp_to_nu}
    Let $\ell,d\in\N$.
    For a distribution $\nu$ over $\Unitaries(d)$, we define 
    \begin{align}
        \cM_{\nu,\ell}(\cdot)\coloneqq
        \Exp_{U\gets\nu}U^{\otimes\ell}(\cdot)U^{\dag\otimes\ell}.
    \end{align}
\end{definition}

\ifnum\submission=0
We use a relationship between Haar random states and Haar random Choi–Jamiołkowski states.

\begin{lemma}[Implicitly Shown in \cite{Harrow23}]\label{lem:Haar_state_vs_Haar_choi_state}
    Let $\ell,d\in\N$ such that $d\ge\ell^2$. Then,
    \begin{align}
         \bigg\|
         (\cM_{\mu_d,\ell}\otimes \identitymap)(\ketbra{\Omega_{d^\ell}}{\Omega_{d^\ell}})
         -
         \Exp_{\ket{\psi}\gets\sigma_{d^2}}\ketbra{\psi}{\psi}^{\otimes\ell}
         \bigg\|_1
         \le O\bigg(\frac{\ell^2}{d}\bigg)
    \end{align}
    where $\ket{\Omega_{d^\ell}}=\frac{1}{d^{\ell/2}}\sum_{x\in[d^\ell]}\ket{x}\ket{x}$ is the maximally entangled state.
\end{lemma}
\fi

\subsection{Unitary Complexity}

For the separation, we use unitary complexity classes. First, we remind the definition of $\unitaryPSPACE$.\footnote{
Note that
in \cite{BosEfrMetPorQiaYue23}, 
\cref{eq:def_unitaryPSPACE}
is replaced with
$\bigcap_{p\in\poly}\unitaryPSPACE_{1/p}$.
}

\begin{definition}[$\unitaryPSPACE$ \cite{ITCS:RosYue22}]\label{def:unitaryPSPACE}
    Let $\delta:\N\to\R$ be a function. We define $\unitaryPSPACE_\delta$ to be the set of a sequence of unitaries $\{U_n\}_{n\in\N}$ such that $\{U_n\}_{n\in\N}$ is quantum polynomial-space implementable\footnote{Here, intermediate measurements are allowed.} 
    with error $\delta$. Namely, there exists a quantum polynomial-space algorithm $C(\cdot,\cdot)$ such that
    \begin{align}
        \|C(1^n,\cdot)-U_n(\cdot)U_n^\dag\|_\diamond\le\delta(n)
    \end{align}
    for all sufficiently large $n\in\N$. We also define
    \begin{align}
        \unitaryPSPACE\coloneqq\bigcap_{p\in\poly}\unitaryPSPACE_{2^{-p}}.\label{eq:def_unitaryPSPACE}
    \end{align}
\end{definition}

In the definition of $\unitaryPSPACE$, intermediate measurements are allowed. 
\cite{FOCS:MetYue23} introduced a variant of $\unitaryPSPACE$, which is called $\pureunitaryPSPACE$ where intermediate measurements are not allowed.\footnote{
It is trivial that $\pureunitaryPSPACE\subseteq\unitaryPSPACE$. However, the other direction is not trivial, because
a quantum polynomial-space algorithm $C$ that implements $U_n$ with an exponentially small error 
could perform exponentially many intermediate measurements, but postponing these measurements requires exponentially many ancilla qubits.}

\begin{definition}[$\pureunitaryPSPACE$ \cite{FOCS:MetYue23}]
    Let $\delta:\N\to\R$ be a function. We define $\pureunitaryPSPACE_\delta$ to be the set of a sequence of unitaries $\{U_n\}_{n\in\N}$ such that $\{U_n\}_{n\in\N}$ is quantum polynomial-space implementable with error $\delta$ and without any intermediate measurement. Namely, there exists a polynomial-space family of unitary circuit $\{C_n\}_n$ such that
    \begin{align}
        \bigg\|
          C_n\ket{\psi}\ket{0...0}-(U_n\ket{\psi})\ket{0...0}
        \bigg\|
        \le\delta(n)
    \end{align}
    for all sufficiently large $n\in\N$ and all pure state $\ket{\psi}$. Here, $\ket{0...0}$ is a state on the ancilla register. We also define
    \begin{align}
        \pureunitaryPSPACE\coloneqq\bigcap_{p\in\poly}\pureunitaryPSPACE_{2^{-p}}.
    \end{align}
\end{definition}

\begin{remark}\label{remark:contrlizatin_of_pureUnitaryPSPACE}
    It is clear that if $\{U_n\}_{n\in\N}\in\pureunitaryPSPACE$, we have $\{\text{c-}U_n\}_{n\in\N}\in\pureunitaryPSPACE$ and $\{U_n^\dag\}_{n\in\N}\in\pureunitaryPSPACE$.
\end{remark}

 In \cite{BosEfrMetPorQiaYue23}, they showed that a certain problem, SUCCINCTUHLMANN, is $\unitaryPSPACE$-complete.
 For this paper, we need only the fact that $\unitaryPSPACE$ has a complete problem, which is formalized as follows.\footnote{
 There are two remarks. First,
 \cite{BosEfrMetPorQiaYue23} showed that SUCCINCTUHLMANN is in $\unitaryPSPACE$. However, their proof can be easily modified so that
 the problem is in $\pureunitaryPSPACE$. 
 Second, in \cite{BosEfrMetPorQiaYue23}, their definition of the reduction and $\unitaryPSPACE$ allows the inverse polynomial error, while we only allow the exponentially small error in \cref{eq:hardness_of_unitaryPSPACE} and \cref{eq:def_unitaryPSPACE}. However, we can easily see their original proof of Theorem 7.14 \cite{BosEfrMetPorQiaYue23} also works in our case.
 }
 

\begin{lemma}[\cite{BosEfrMetPorQiaYue23}]\label{lem:unitaryPSPACE_has_a_complete_problem}
    $\unitaryPSPACE$ has a complete problem. Namely, there exists a sequence of unitaries $\cU=\{U_n\}_n$ that satisfies the following.
    \begin{itemize}
        \item 
        $\cU\in\unitaryPSPACE$. Moreover, $\cU\in\pureunitaryPSPACE$.
        \item For any polynomial $p$ and $\cV=\{V_n\}_n\in\unitaryPSPACE$, there exists a QPT algorithm $\cA^{(\cdot)}$ such that
        \begin{align}
            \|\cA^{\cU}(1^n,\cdot)-V_n(\cdot)V_n^\dag\|_\diamond\le2^{-p(n)}\label{eq:hardness_of_unitaryPSPACE}
        \end{align}
        for all sufficiently large $n\in\N$. 
    \end{itemize}
\end{lemma}

\subsection{Quantum Singular Value Transformation and Block Encoding}
To break PRUs, we use a quantum singular value transformation (QSVT) \cite{STOC:GSLW19}. Especially, we use the following singular value discrimination algorithm.

\begin{theorem}[Singular Value Discrimination Algorithm \cite{STOC:GSLW19}]\label{thm:singular_value_discrimination}
    Let $0\le a<b\le1$. Suppose that $M\in\Linear(d)$ can be written $M=\widetilde{\Pi}U\Pi$ with some $U\in\Unitaries(d)$ and projections $\Pi,\widetilde{\Pi}\in\Linear(d)$. Let $\xi$ be a given unknown state promised that
    \begin{itemize}
        \item the support of $\xi$ is contained in the subspace $W_0$, which is the subspace spanned by the right singular vectors of $M$ with singular value at most $a$ or
        \item the support of $\xi$ is contained in the subspace $W_1$, which is the subspace spanned by the right singular vectors of $M$ with singular value at least $b$.
    \end{itemize}
    Then, for each $\eta>0$, there exists an algorithm $\cD$ satisfying the following:
    \begin{itemize}
        \item on input a single copy of $\xi$, $\cD$ distinguishes between the first case or the second case with probability at least $1-\eta$;
        \item $\cD$ uses $U,U^\dag,\CNOT{\Pi},\CNOT{\widetilde{\Pi}}$ and other single-qubit gates 
        $$
        O\bigg(
         \frac{1}{\max\{b-a,\sqrt{1-a^2}-\sqrt{1-b^2}\}}
         \log\bigg(\frac{1}{\eta}\bigg)
         \bigg)
        $$ times, and uses a single ancilla qubit. Here, $\CNOT{\Pi}\coloneqq\Pi\otimes X+(I-\Pi)\otimes I$ and $\CNOT{\widetilde{\Pi}}$ is defined in the same way.
    \end{itemize}
\end{theorem}

When we apply the singular value discrimination algorithm, we need to encode a matrix into a unitary circuit. This technique is referred to as block encoding.

\begin{definition}[Block Encoding \cite{STOC:GSLW19}]\label{def:block_encoding}
    Let $M\in\Linear(d)$.
    We say that $U\in\Unitaries(2^ad)$ is an $(\alpha,\epsilon,a)$-block encoding of $M$ for some $\alpha\ge1,\epsilon\ge0$ and $a\in\N$ if it satisfies
    \begin{align}
        \|M-\alpha(\bra{0^a}\otimes I)U(\ket{0^a}\otimes I)\|_\infty\le\epsilon.
    \end{align}
\end{definition}

\ifnum\submission=0
In general, it is not clear that we can space-efficiently implement a block encoding unitary of any matrix $M$.
The following lemma ensures that we can implement a block encoding unitary of a density operator if we can generate its purification.

\begin{lemma}[Lemma 12 in \cite{ICALP:vApGil19}]\label{lem:block-encoding_of_state}
    Let $U$ be a unitary over registers $\regA$ and $\regB$, where $\regA$ and $\regB$ are $n$-qubit register and $m$-qubit register, respectively. Define $\rho_{\regA}\coloneqq\Tr_{\regB}[(U\ketbra{0...0}{0...0}U^\dag)_{\regA\regB}]$. Then, there exists a $(1,0,n+m)$-block encoding unitary $V$ of $\rho$, where $V$ is implementable with single use of $U$ and $U^\dag$, and $n+1$ two-qubit gates.
\end{lemma}
\fi

\subsection{Cryptographic Primitives}

We recall PRUs defined by \cite{C:JiLiuSon18}.

\begin{definition}[Pseudorandom Unitaries \cite{C:JiLiuSon18}]\label{def:PRU}
    We define that an algorithm $G$ is a pseudorandom unitary generator (PRU) if it satisfies the following:
    \begin{itemize}
        \item Correctness: Let $\secp$ be the security parameter. Let $\cK_\secp$ denote the key-space at most $\poly(\secp)$ bits.
         $G$ is a QPT algorithm such that $G(k,\ket{\psi})=U_k\ket{\psi}$ for any $\secp$-qubit state $\ket{\psi}$.
         \item Pseduorandomness: For any uniform QPT algorithm $\cA^{(\cdot)}$,
        \begin{align}
            \bigg|\Pr_{k\gets\cK_\secp}[1\gets\cA^{U_k}(1^\secp)]-\Pr_{U\gets\mu_{2^\secp}}[1\gets\cA^{U}(1^\secp)]
            \bigg|\le\negl(\secp).\label{eq:security_def_of_PRU}
        \end{align}
    \end{itemize}
    If \cref{eq:security_def_of_PRU} holds for any non-adaptively-querying adversary, we call $G$
    non-adaptive PRU.\footnote{Here, non-adaptive query means that the adversary queries
    $U^{\otimes\poly(\secp)}$ only once.}
    If $G(k,\cdot)$ uses at most $c$ ancilla qubits to implement $U_k$ for all $k \in \cK_\secp$, we call $G$ a $c$-ancilla PRU.
\end{definition}
\begin{remark}
We could define PRUs secure against    
non-uniform adversaries, but in this paper we can break PRUs against uniform
adversaries, and therefore we provide only the definition of the latter.
\end{remark}

\begin{definition}[Pseduorandom Isometries \cite{EC:AGKL24}]
    Let $s:\N\mapsto\N$ be a function such that $s(\secp)\le\poly(\secp)$.
    We define that an algorithm $G$ is a pseudorandom isometry generator with $s$ stretch (PRI) if it satisfies the following:
    \begin{itemize}
        \item Correctness: Let $\secp$ be the security parameter. Let $\cK_\secp$ denote the key-space at most $\poly(\secp)$ qubits.
         $G$ is a QPT algorithm such that $G(k,\ket{\psi})=\cI_k\ket{\psi}$ for any $\secp$-qubit state $\ket{\psi}$, where $\cI_k$ is an isometry that maps $\secp$ qubits to $\secp+s(\secp)$ qubits.
         \item Pseduorandomness: For any uniform QPT algorithm $\cA^{(\cdot)}$,
        \begin{align}
            \bigg|\Pr_{k\gets\cK_\secp}[1\gets\cA^{\cI_k}(1^\secp)]-\Pr_{U\gets\mu_{2^{\secp+s(\secp)}}}[1\gets\cA^{\cI_U}(1^\secp)]
            \bigg|\le\negl(\secp),\label{eq:security_def_of_PRI}
        \end{align}
        where, for each $U\in\Unitaries(2^{\secp+s(\secp)})$, $\cI_U$ is the isometry that maps $\secp$-qubit state $\ket{\psi}$ to $(\secp+s(\secp))$-qubit state $U(\ket{\psi}\ket{0^s})$.\footnote{Without loss of generality, the ancilla state can be $\ket{0^s}$ due to the right invariance of the Haar measure.} 
        
    \end{itemize}
    If \cref{eq:security_def_of_PRI} holds for any non-adaptively-querying adversary, we call $G$
    non-adaptive PRI with $s$ stretch.\footnote{Here, non-adaptive query means that the adversary queries
    $\cI^{\otimes\poly(\secp)}$ only once.}
    If $G(k,\cdot)$ uses at most $s+c$ ancilla qubits to implement $\cI_k$ for all $k \in \cK_\secp$, we call $G$ a $c$-ancilla PRI with $s$ stretch.\footnote{To implement an isometry from $\secp$ qubits to $(\secp+s)$ qubits, we need at least $s$ ancilla qubits.}
\end{definition}

Quantumly-accessible adaptively-secure PRFSGs were defined in \cite{TCC:AGQY22}.

\begin{definition}[Quantumly-accessible adaptively-secure PRFSGs \cite{TCC:AGQY22}]\label{def:PRFSG}
    We define that an algorithm $G$ is a quantumly-accessible adaptively-secure PRFSG if it satisfies the following: 
    \begin{itemize}
        \item Correctness: Let $\secp\in\N$ be the security parameter. Let $q$ be a polynomial. Let $\cK_\secp$ denote the key-space at most $\poly(\secp)$ bits.
        $G$ is a QPT algorithm that takes a key $k\in\cK_\secp$ and a bit string $x$ as input, and outputs a pure $q(\secp)$-qubit state $\ket{\phi_k(x)}$.

        \item Quantumly-accessible adaptive security: For any QPT adversary $\cA^{(\cdot)}$ and any bit sting $y$ whose length is at most polynomial of $\secp$,
        \begin{align}
            \bigg|\Pr_{k\gets\cK_\secp}[1\gets\cA^{
            G(k,\cdot)}(1^\secp,y)]-\Pr_{\{\ket{\vartheta_x}\}\gets\sigma}[1\gets\cA^{\cH_{\{\ket{\vartheta_x}\}}}(1^\secp,y)]\bigg|\le\negl(\secp),
        \end{align}
        where $\{\ket{\vartheta_x}\}\gets\sigma$ denotes that each $\ket{\vartheta_x}$ is independently chosen from the Haar measure $\sigma_{2^q}$.
        Here, the actions of $G(k,\cdot)$ and $\cH_{\{\ket{\vartheta_x}\}}$ are defined as follows:
        \begin{itemize}
            \item $G(k,\cdot):$ It applies $\ket{x}\mapsto\ket{x}\ket{\phi_k(x)}$ coherently.\footnote{When a superposition $\sum_x\alpha_x|x\rangle\ket{\xi_x}$ is queried, it outputs $\sum_x\alpha_x|x\rangle|\phi_k(x)\rangle\ket{\xi_x}$. In general
            it is not possible when the junk states depending on $x$ appear, but in our construction of PRFSG, junk states are independent of $x$.}
            \item $\cH_{\{\ket{\vartheta_x}\}}:$ It applies
            $\ket{x}\mapsto\ket{x}\ket{\vartheta_x}$ coherently.
        \end{itemize}
    \end{itemize}
\end{definition}

In this paper, we often omit the term "quantumly-accessible adaptively-secure".

\begin{remark}
Here we provide the definition of PRFSGs secure against non-uniform adversaries with classical advice because
we can construct it. We can also consider the security against all non-uniform adversaries with quantum advice, but it is not clear whether our construction \cref{const:PRFSGs} satisfies the security. We leave it to the future work.
\end{remark}

\section{Separation Oracle}
In this section, we define an oracle that separates between PRUs and PRFSGs. 
Our separation oracle is defined as follows:

\begin{definition}[Separation Oracle]\label{def:unitary_orcle}
    We define an oracle $\cO\coloneqq(\cS,\cU)$ as follows:
    \begin{itemize}
        \item 
        For each $n\in\N$ and $m\in\bit^n$, sample $\ket{\psi_{n,m}}\in\States(2^n)$ from the Haar measure $\sigma_{2^n}$. Then, define the $(n+1)$-qubit swapping unitary
        \begin{align}
            \cS_{n,m}\coloneqq\ketbra{0}{1}\otimes\ketbra{0^n}{\psi_{n,m}}+\ketbra{1}{0}\otimes\ket{\psi_{n,m}}\bra{0^{n}}+I_\bot^{n,m}
        \end{align}
        for each $n\in\N$ and $m\in\bit^n$.
        Here, $I_\bot^{n,m}$ is the identity on the subspace orthogonal to $\SpanSpace\{\ket{0}\ket{0^{n}},\ket{1}\ket{\psi_{n,m}}\}$.
        We define $\cS\coloneqq\{\cS_{n}\}_{n\in\N}$, where 
        $\cS_n\coloneqq\sum_{m\in\bit^n}\ketbra{m}{m}\otimes
        \cS_{n,m}$ is a $(2n+1)$-qubit unitary.
        \item $\cU\coloneqq\{U_n\}_{n\in\N}$ is the $\unitaryPSPACE$ complete problem in \cref{lem:unitaryPSPACE_has_a_complete_problem}.
    \end{itemize}
\end{definition}

In this work, we allow not only the query to $\cO$ but also the query to the inverse of $\cO$.
When we write $\cA^\cO$ for an algorithm $\cA$, it can query $\cO$ and the inverse of $\cO$.

\begin{remark}\label{remark:considering_only_forward_query}
    For any $\ket{\psi_{n,m}}\in\States(2^n)$, $\cS_{n,m}=\cS_{n,m}^\dag$ by its definition. Thus, when we consider an algorithm $\cA^\cS$, it suffices to consider the forward query to $\cS$ regardless of the choice of $\cS$.
\end{remark}

Our goal is to show the following.

\begin{theorem}\label{thm:main}
    With probability $1$ over the choice of $\cO$ defined in \cref{def:unitary_orcle}, the following are satisfied:
    \begin{enumerate}
        \item Quantumly-accessible adaptively-secure PRFSGs exist relative to $\cO$.\label{item:PRFSGs_exist}
        \item Non-adaptive, $O(\log\secp)$-ancilla PRUs do not exist relative to $\cO$.\label{item:PRUs_do_not_exist}
    \end{enumerate}
\end{theorem}

We prove the existence of PRFSGs in \cref{sec:PRFSGs}, and the non-existence of PRUs in \cref{sec:PRUs}.
\section{Constructing PRFSGs}
\label{sec:PRFSGs}
In this section, we show that quantumly-accessible adaptively-secure PRFSGs exist relative to $\cO$. We construct PRFSGs as follows:

\begin{definition}\label{const:PRFSGs}
    Let $\cO=(\cS,\cU)$ be the oracle in \cref{def:unitary_orcle}. Relative to $\cO$, we define a QPT algorithm $G^{\cO}$ as follows:
    \begin{enumerate}
        \item Let $k,x\in\bit^\secp$ be an input.\footnote{In our explicit construction, the length of the secret key is the same as that of the input bit string. However, our security proof works if they are different. On the other hand, the number of qubits of output states must be larger than the length of the secret key and input bit string. In particular, it is not clear whether short PRFSGs exist or not relative to our oracle. We leave it to further work.} Here, $k$ is a secret key and $x$ is an input bit string.
        \item Prepare $\ket{(k,x)}\ket{0}\ket{0^{2\secp}}$. Here $(k,x)$ denotes the concatination of $k$ and $x$.
        \item Obtain $\ket{(k,x)}\ket{1}\ket{\psi_{2\secp,(k,x)}}$ by querying $\ket{(k,x)}\ket{0}\ket{0^{2\secp}}$ to $\cS_{2\secp}$. 
        \item Output $\ket{\psi_{2\secp,(k,x)}}$.
    \end{enumerate} 
\end{definition}

\ifnum\submission=0
The goal of this section is to prove the following:
\fi
\ifnum\submission=1
We can prove the following. 
Since the proof is the same as \cite{TQC:Kre21}, we give it in the supplemental material.
\fi

\begin{restatable}{theorem}{PRFSGs}\label{thm:PRFSGs_relative_to_unitary_oracle}
    With probability $1$ over the randomness of $\cO$ (defined in \cref{def:unitary_orcle}), \cref{const:PRFSGs} is a quantumly-accessible adaptively-secure PRFSGs relative to $\cO$.
\end{restatable}

\ifnum\submission=0
\ifnum\submission=1
\section{Security proof for PRFSGs}
In this section, we prove \cref{thm:PRFSGs_relative_to_unitary_oracle}.
\fi
Our strategy is the same as \cite{TQC:Kre21} which shows PRUs exist relative to exponentially many Haar random unitary oracles.
As a first step, we need that swap oracles are indistinguishable from independent swap oracles on average. 
\ifnum\submission=0
This is formalized as follows and directly follows from \cref{lem:BBBV}.
\fi
\ifnum\submission=1
This is formalized as follows:
\fi

\begin{lemma}\label{lem:indistinguishability_of_swap_oracle}
    Let $\cA^{(\cdot,\cdot)}$ be an algorithm. Let $\secp\in\N$. 
    For each fixed $\cS$ defined in \cref{def:unitary_orcle},
    we define
    \begin{align}
        \Adv(\cA,\cS_{2\secp})\coloneqq
         \Pr_{k\gets\bit^\secp}[1\gets\cA^{\cT_{2\secp,k},\cS_{2\secp}}]
         -\Pr_{\ket{\vartheta_1},...,\ket{\vartheta_{2^\secp}}\gets\sigma_{2^{2\secp}}}[1\gets\cA^{\cT_{\{\ket{\vartheta_x}\}},\cS_{2\secp}}],
    \end{align}
    where
    \begin{itemize}
        \item 
        $\cT_{2\secp,k}
        \coloneqq\sum_{x\in\bit^\secp}\ketbra{x}{x}
        \otimes\cS_{2\secp,(k,x)}$.
        \item for $\ket{\vartheta_1},...,\ket{\vartheta_{2^\secp}}$, 
        $\cT_{\{\ket{\vartheta_x}\}}
        \coloneqq
        \sum_{x\in\bit^\secp}\ketbra{x}{x}
        \otimes\cT_{\ket{\vartheta_x}}$. Here, for $\ket{\vartheta}\in\States(2^{2\secp})$, we define $\cT_{\ket{\vartheta}}\coloneqq
        \ketbra{0}{1}\otimes\ketbra{0^{2\secp}}{\vartheta}
        +\ketbra{1}{0}\otimes\ket{\vartheta}\bra{0^{2\secp}}+I^{\ket{\vartheta}}_\bot$, where $I^{\ket{\vartheta}}_\bot$ is the identity on the subspace orthogonal to $\SpanSpace\{\ket{0}\ket{0^{2\secp}},\ket{1}\ket{\vartheta}\}$.
    \end{itemize}
    Then, there exists a constant $c>0$ such that, for any algorithm $\cA^{(\cdot,\cdot)}$ that makes $T$ queries in total,
    \begin{align}
        \bigg|
        \Exp_{\cS_{2\secp}\gets\sigma}[\Adv(\cA,\cS_{2\secp})]
        \bigg|
        \le\frac{cT^2}{2^\secp},
    \end{align}
    where $\cS_{2\secp}\gets\sigma$ denotes that, for each $m\in\bit^{2\secp}$, $\ket{\psi_{2\secp,m}}$ is drawn from the Haar measure $\sigma_{2^{2\secp}}$ independently.
\end{lemma}

\ifnum\submission=1
For the proof, we need the following lemma.\footnote{In \cite{TQC:Kre21}, he showed \cref{lem:BBBV} only for the case when $\cD$ is the Haar measure. However, his proof does not rely on any property of the Haar measure because its proof essentially depends on the BBBV theorem \cite{SICOMP:BBBV97}. Thus, we can obtain the same claim for a general distribution $\cD$ from his original proof.}

\begin{lemma}[Lemma 31 in \cite{TQC:Kre21}]\label{lem:BBBV}
    Let $\cD$ be a distribution over $\Unitaries(d)$. Suppose that $\cA$ is a quantum algorithm that queries $U=(U_1,..,U_N)\in\Unitaries(d)^N$ and $O\in\Unitaries(d)$. We see $U$ as $\sum_{n\in[N]}\ketbra{n}{n}\otimes U_n$. For fixed $U\in\Unitaries(d)^N$, define
    \begin{align}
        \Adv(\cA,U)\coloneqq
        \Pr_{k\gets[N]}[1\gets\cA^{U_k,U}]-\Pr_{O\gets\cD}[\cA^{O,U}].
    \end{align}
    Then, there exists a constant $c>0$ such that, for any $T$-query algorithm $\cA$,
    \begin{align}
        \bigg|
         \Exp_{U\gets\cD^N}[\Adv(\cA,U)]
        \bigg|
        \le\frac{cT^2}{N}.
    \end{align}
    Here, $U\gets\cD^N$ denotes that each $U_k$ is independently sampled from $\cD$.
\end{lemma}
\fi

\begin{proof}[Proof of \cref{lem:indistinguishability_of_swap_oracle}]
    Note that $\cS_{2\secp}=\sum_{k\in\bit^\secp}\ketbra{k}{k}\otimes\cT_{2\secp,k}$. Thus, this claim follows from \cref{lem:BBBV} with $N=2^\secp,$  and $\cD=\sigma$.
\end{proof}

Next, we want to show that $\Adv(\cA,\cS_{2\secp})$ is negligible with overwhelming probability over the choice of $\cS_{2\secp}$ by invoking the concentration inequality (\cref{thm:Haar_concentration}). 
For that goal, we view $\Adv(\cA,\cS_{2\secp})$ as a function of $U\in\Unitaries(2^{2\secp})^{2^{2\secp}}$, and need to show that it satisfies Lipshcitz condition in \cref{thm:Haar_concentration}. 
This is formalized as follows:

\begin{restatable}{lemma}{Lipschitz}\label{lem:Lipschitz_swap_oracle}
    Let $n\in\N$.
    For $U=(U_1,...,U_{2^n})\in\Unitaries(2^n)^{2^n}$, we define
    \begin{align}
        \Tilde{\cS_n}(U)\coloneqq\sum_{m\in\bit^n}\ketbra{m}{m}\otimes 
        \bigg(
         \ketbra{0}{1}\otimes\ketbra{0^n}{0^n}U_m^\dag+\ketbra{1}{0}\otimes U_m\ketbra{0^n}{0^n}+I_\bot(U_m)
        \bigg),
    \end{align}
    where each $I_\bot(U_m)$ is the identity on the subspace orthogonal to $\SpanSpace\{\ket{0}\ket{0^{n}},\ket{1}U_m\ket{0^n}\}$.
    Then, for any algorithm $\cA^{(\cdot)}$ that makes $T$ queries, $f(U)\coloneqq\Pr[1\gets\cA^{\Tilde{\cS_n}(U)}]$ is $8T$-Lipschitz in the $\ell^2$-sum of the Forbenious norm. Namely, for any $U=(U_1,...,U_{2^n}),V=(V_1,..,V_{2^n})\in\Unitaries(2^n)^{2^n}$,
    \begin{align}
        |f(U)-f(V)|\le 8T\sqrt{\sum_{m\in\bit^n}\|U_m-V_m\|^2_2}.
    \end{align}
\end{restatable}

\ifnum\submission=1
For the proof, we need the following lemma.
\begin{lemma}[Lemma 28 in \cite{TQC:Kre21}]\label{lem:Lipschitz}
    Let $\cA^{U}$ be a quantum algorithm that makes $T$ queries to $U\in\Unitaries(d)$ and its inverse. Then, $f(U)=\Pr[1\gets\cA^{U}]$ is $2T$-Lipschitz in the Frobenius norm, i.e., $|f(U)-f(V)|\le2T\|U-V\|_2$ for all $U,V\in\Unitaries(d)$.
\end{lemma}

By applying the triangle inequality, we have the following.

\begin{lemma}\label{lem:f+g_Lipschitz}
    Let $f,g:\Unitaries(d)\to\R$ be $L$-Lipschitz functions in the Frobenious norm, i.e.,
    $|f(U)-f(V)|\le L\|U-V\|_2$ and $|g(U)-g(V)|\le L\|U-V\|_2$ for any $U,V\in\Unitaries(d)$. Then, $f+g$ is $2L$-Lipschitz. Namely, for any $U,V\in\Unitaries(d)$,
    \begin{align}
        |f(U)+g(U)-f(V)-g(V)|\le2L\|U-V\|_2.
    \end{align}
\end{lemma}

Now we are ready to prove \cref{lem:Lipschitz_swap_oracle}.
\fi

\begin{proof}[Proof of \cref{lem:Lipschitz_swap_oracle}]
    Note that, for each $m\in\bit^n$, we can write $I_\bot(U_m)$ as follows:
    \begin{align}
        I_\bot(U_m)
        =I-\ketbra{0}{0}\otimes\ketbra{0^n}{0^n}-\ketbra{1}{1}\otimes U_m\ketbra{0^n}{0^n}U_m^\dag. \label{eq:explicit_form_of_swap_oracle}
    \end{align}
    Thus, we have
    \begin{align}
        &|f(U)-f(V)|\notag\\
        \le& 2T\|\Tilde{\cS_n}(U)-\Tilde{\cS_n}(V)\|_2 \tag{By \cref{lem:Lipschitz}}\\
        =&2T\bigg\|\sum_{m\in\bit^n}\ketbra{m}{m}\otimes
         \bigg(
          \ketbra{0}{1}\otimes\ketbra{0^n}{0^n}(U_m-V_m)^\dag+\ketbra{1}{0}\otimes (U_m-V_m)\ketbra{0^n}{0^n}
          \notag\\
          &
          -\ketbra{1}{1}\otimes U_m\ketbra{0^n}{0^n}U_m+\ketbra{1}{1}\otimes V_m\ketbra{0^n}{0^n}V_m^\dag
         \bigg)
        \bigg\|_2\tag{By \cref{eq:explicit_form_of_swap_oracle}}\\
        \le&2T\bigg\|\sum_{m\in\bit^n}\ketbra{m}{m}\otimes
         \bigg(\ketbra{0}{1}\otimes \ketbra{0^n}{0^n}(U_m-V_m)^\dag+\ketbra{1}{0}\otimes (U_m-V_m)\ketbra{0^n}{0^n}
         \bigg)
         \bigg\|_2
         \notag\\
         &+2T\bigg\|
         \sum_{m\in\bit^n}\ketbra{m}{m}\otimes\ketbra{1}{1}\otimes
          \bigg(
           U_m\ketbra{0^n}{0^n}U_m- V_m\ketbra{0^n}{0^n}V_m^\dag
          \bigg)
         \bigg\|_2,\label{eq:f(U)-f(V)}
    \end{align}
    where the last inequality follows from the triangle inequality. The first term in \cref{eq:f(U)-f(V)} is estimated as follows:

    \begin{align}
        &2T\bigg\|\sum_{m\in\bit^n}\ketbra{m}{m}\otimes
         \bigg(\ketbra{0}{1}\otimes \ketbra{0^n}{0^n}(U_m-V_m)^\dag+\ketbra{1}{0}\otimes (U_m-V_m)\ketbra{0^n}{0^n}
         \bigg)
         \bigg\|_2
         \notag\\
        \le&4T\bigg\|\sum_{m\in\bit^n}\ketbra{m}{m}\otimes
        \ketbra{1}{0}\otimes (U_m-V_m)\ketbra{0^n}{0^n}
         \bigg\|_2
        \tag{By $\|A^\dag\|_2=\|A\|_2$ and the triangle inequality}\\
        =&4T\sqrt{\sum_{m\in\bit^n}
        \bigg\|\ketbra{1}{0}\otimes
         (U_m-V_m)\ketbra{0^n}{0^n}
        \bigg\|^2_2}
        \tag{By $\|\sum_m\ketbra{m}{m}\otimes A_m\|_2=\sqrt{\sum_m\|A_m\|^2_2}$}\\
        =&4T\sqrt{\sum_{m\in\bit^n}
        \bigg\|
         (U_m-V_m)\ketbra{0^n}{0^n}
        \bigg\|^2_2}
        \tag{By $\|A\otimes B\|_2=\|A\|_2\|B\|_2$ and $\|\ketbra{1}{0}\|_2=1$}\\
        \le&4T\sqrt{\sum_{m\in\bit^n}\|U_m-V_m\|^2_2},\label{eq:ineq1_for_f(U)-f(V)}
    \end{align}
    where the last inequality follows from the H\"older's inequality (\cref{lem:Holder}) and $\|\ketbra{0^n}{0^n}\|_\infty=1$.
    On the other hand, the second term in \cref{eq:f(U)-f(V)} is estimated as follows:
    \begin{align}
        &2T\bigg\|
         \sum_{m\in\bit^n}\ketbra{m}{m}\otimes\ketbra{1}{1}\otimes
          \bigg(
           U_m\ketbra{0^n}{0^n}U_m-V_m\ketbra{0^n}{0^n}V_m^\dag
          \bigg)
         \bigg\|_2
         \notag\\
         =&2T\sqrt{\sum_{m\in\bit^n}
         \bigg\|
          U_m\ketbra{0^n}{0^n}U_m-V_m\ketbra{0^n}{0^n}V_m^\dag
         \bigg\|_2^2}
         \tag{By $\|\sum_m\ketbra{m}{m}\otimes A_m\|_2=\sqrt{\sum_m\|A_m\|^2_2}$}\\
         \le&4T\sqrt{\sum_{m\in\bit^n}
         \|U_m-V_m\|_2^2},\label{eq:ineq2_for_f(U)-f(V)}
    \end{align}
    \ifnum\submission=0
    where the last inequality follows from \cref{lem:Uket0-Vket0_is_less_than_U-V}.
    \fi
    \ifnum\submission=1
    where the last inequality follows from the following lemma.
    We give the proof later.
    \begin{lemma}\label{lem:Uket0-Vket0_is_less_than_U-V}
    Let $U,V\in\Unitaries(d)$. Then, for any $\ket{\psi}\in\States(d)$,
    \begin{align}
        \|U\ketbra{\psi}{\psi}U^\dag-V\ketbra{\psi}{\psi}V^\dag\|_2
        \le2\|U-V\|_2
    \end{align}
    \end{lemma}
    \fi
    By combining \cref{eq:f(U)-f(V),eq:ineq1_for_f(U)-f(V),eq:ineq2_for_f(U)-f(V)}, we have
    \begin{align}
        |f(U)-f(V)|\le8T\sqrt{\sum_{m\in\bit^n}
         \|U_m-V_m\|_2^2},
    \end{align}
    which concludes the proof.
\end{proof}

\ifnum\submission=1
To complete the proof, we give the proof of \cref{lem:Uket0-Vket0_is_less_than_U-V}

\begin{proof}[Proof of \cref{lem:Uket0-Vket0_is_less_than_U-V}]
    We obtain the inequality as follows:
    \begin{align}
        &\|U\ketbra{\psi}{\psi}U^\dag-V\ketbra{\psi}{\psi}V^\dag\|_2
        \notag\\
        \le&\|U\ketbra{\psi}{\psi}U^\dag-V\ketbra{\psi}{\psi}U^\dag\|_2+\|V\ketbra{\psi}{\psi}U^\dag-V\ketbra{\psi}{\psi}V^\dag\|_2
        \tag{By the triangle inequality}\\
        =&\|(U-V)\ketbra{\psi}{\psi}U^\dag\|_2+\|V\ketbra{\psi}{\psi}(U-V)^\dag\|_2
        \notag\\
        =&\|(U-V)\ketbra{\psi}{\psi}U^\dag\|_2+\|(U-V)\ketbra{\psi}{\psi}V^\dag\|_2
        \tag{By $\|A^\dag\|_2=\|A\|_2$}\\
        \le&\|U-V\|_2\|\ketbra{\psi}{\psi}U^\dag\|_\infty
        +\|(U-V)\|_2\|\ketbra{\psi}{\psi}V^\dag\|_\infty
        \tag{By H\"older's inequality, \cref{lem:Holder}}\\
        \le&2\|U-V\|_2,
    \end{align}
    where in the last inequality we have used $\|\ketbra{\psi}{\psi}W\|_\infty\le\|\ketbra{\psi}{\psi}\|_\infty\|W\|_\infty\le1$ for any $W\in\Unitaries(d)$ from \cref{lem:Holder}.
\end{proof}
\fi

\ifnum\submission=1
We need the following concentration property.
\begin{theorem}[Theorem 5.17 in \cite{Mec}]\label{thm:Haar_concentration}
     Given $d_1,\ldots,d_k \in \N$, let $X = \Unitaries(d_1)\times\cdots\times\Unitaries(d_k)$. Let $\mu = \mu_{d_1} \times \cdots \times \mu_{d_k}$ be the product of Haar measures on $X$. Suppose that $f: X \to \R$ is $L$-Lipschitz in the $\ell^2$-sum of Frobenius norm, i.e., for any $U=(U_1,...,U_k)\in X$ and $V=(V_1,...,V_k)\in X$, 
     we have $|f(U)-f(V)|\le L\sqrt{\sum_i\|U_i-V_i\|^2_2}$. Then for every $\delta > 0$,
    \begin{align}
       \Pr_{U\gets \mu}\left[f(U)\ge\Exp_{V \gets \mu}[f(V)]+\delta\right] \le \exp\left(-\frac{(d-2)\delta^2}{24L^2}\right),
    \end{align}
    where $d \coloneqq \min\{d_1,\ldots,d_k\}$.
\end{theorem}
\fi

With \cref{lem:indistinguishability_of_swap_oracle,lem:Lipschitz_swap_oracle} at hand, we can argue $\Adv(\cA,\cS_{2\secp})$ is negligible with high probability.

\begin{lemma}\label{lem:Adv_is_negligible_whp}
    Let $c$ be a constant in \cref{lem:indistinguishability_of_swap_oracle}.
    Suppose that $\cA^{(\cdot,\cdot)}$ is an algorithm that makes $T$ queries in total. Then, for any $p\ge\frac{cT^2}{2^{\secp}}$,
    \begin{align}
        \Pr_{\cS_{2\secp}\gets\sigma}[|\Adv(\cA,\cS_{2\secp})|\ge p]
        \le 2\exp\bigg(
         -\frac{(2^{2\secp}-2)(p-cT^22^{-\secp})^2}{6144T^2}
        \bigg),
    \end{align}
    where $\Adv(\cA,\cS_{2\secp})$ is defined in \cref{lem:indistinguishability_of_swap_oracle}, and $\cS_{2\secp}\gets\sigma$ denotes that, for each $m\in\bit^{2\secp}$, $\ket{\psi_{2\secp,m}}$ is drawn from the Haar measure $\sigma_{2^{2\secp}}$ independently.
\end{lemma}

\begin{proof}[Proof of \cref{lem:Adv_is_negligible_whp}]
    Recall that 
    $\cS_{2\secp}=\sum_{m\in\bit^{2\secp}}
    \ketbra{m}{m}\otimes
    \cS_{2\secp,m}$ 
    and each $S_{2\secp,m}$ is a unitary that swaps between $\ket{0}\ket{0^{2\secp}}$ and $\ket{1}\ket{\psi_{2\secp,m}}$. It is clear that choosing $\ket{\psi_{2\secp,m}}$ from the Haar measure $\sigma_{2^{2\secp}}$ independently is exactly the same as setting $\ket{\psi_{2\secp,m}}\coloneqq U_m\ket{0^{2\secp}}$, where $U_m$ is chosen from the Haar measure $\mu_{2^{2\secp}}$ independently. Thus, let $\Tilde{\cS}_{2\secp}\gets\mu$ denote that, for each $m\in\bit^{2\secp}$, $\ket{\psi_{2\secp,m}}$ is defined by $U_m\ket{0^{2\secp}}$ and $U_m\gets\mu_{2^{2\secp}}$.
    Recall that
    \begin{align}
        \Adv(\cA,\Tilde{\cS}_{2\secp})=
         \Pr_{k\gets\bit^\secp}[1\gets\cA^{\cT_{2\secp,k},\Tilde{\cS}_{2\secp}}]
         -\Pr_{\ket{\vartheta_1},...,\ket{\vartheta_{2^\secp}}\gets\sigma_{2^{2\secp}}}[1\gets\cA^{\cT_{\{\ket{\vartheta_x}\}},\Tilde{\cS}_{2\secp}}].
    \end{align}
    We can see that both teams are $8T$-Lipschitz functions as follows:
    \begin{itemize}
        \item $\Pr_{k\gets\bit^\secp}[1\gets\cA^{\cT_{2\secp,k},\Tilde{\cS}_{2\secp}}]:$ define the following algorithm $\cA^{\Tilde{\cS}_{2\secp}}_1$. 
         \begin{enumerate}
             \item Choose $k\gets\bit^\secp$.
             \item Simulate $\cA^{\cT_{2\secp,k},\Tilde{\cS}_{2\secp}}$ with the query access to $\Tilde{\cS}_{2\secp}$. When $\cA$ queries a register $\regA$ to $\cT_{2\secp,k}$, simulate its query by preparing $\ket{k}$ on an ancilla register $\regR$ and querying $\regR\regA$ to $\Tilde{\cS}_{2\secp}$.
             \item Output what $\cA$ outputs.
         \end{enumerate}
        It is clear that $\Pr[1\gets\cA_1^{\Tilde{\cS}_{2\secp}}]=\Pr_{k\gets\bit^\secp}[1\gets\cA^{\cT_{2\secp,k},\Tilde{\cS}_{2\secp}}]$, and $\cA_1^{\Tilde{\cS}_{2\secp}}$ makes $T$ queries. Thus, from \cref{lem:Lipschitz_swap_oracle}, we can view that $\Pr_{k\gets\bit^\secp}[1\gets\cA^{\cT_{2\secp,k},\Tilde{\cS}_{2\secp}}]$ is an $8T$-Lipschitz function in the $\ell^2$-sum of Frobenious norm.

        \item $\Pr_{\ket{\vartheta_1},...,\ket{\vartheta_{2^\secp}}\gets\sigma_{2^{2\secp}}}[1\gets\cA^{\cT_{\{\ket{\vartheta_x}\}},\Tilde{\cS}_{2\secp}}]:$ define the following algorithm $\cA^{\Tilde{\cS}_{2\secp}}_2$. 
         \begin{enumerate}
             \item Choose $\ket{\vartheta_1},...,\ket{\vartheta_{2^\secp}}\gets\sigma_{2^{2\secp}}$ along with their classical descriptions.
             \item Simulate 
             $\cA^{\cT_{\{\ket{\vartheta_x}\}},\Tilde{\cS}_{2\secp}}$.
             When $\cA$ queries $\cT_{\{\ket{\vartheta_x}\}}$, $\cA_2$ applies $\cT_{\{\ket{\vartheta_x}\}}$ by using the classical descriptions of $\ket{\vartheta_1},...,\ket{\vartheta_{2^\secp}}$. 
             \item Output what $\cA^{\cT_{\{\ket{\vartheta_x}\}},\Tilde{\cS}_{2\secp}}$ outputs.
         \end{enumerate}
        It is clear that $\Pr[1\gets\cA_2^{\Tilde{\cS}_{2\secp}}]
        =\Pr_{\ket{\vartheta_1},...,\ket{\vartheta_{2^\secp}}\gets\sigma_{2^{2\secp}}}
        [1\gets\cA^{\cT_{\{\ket{\vartheta_x}\}},\Tilde{\cS}_{2\secp}}]$, and $\cA_2^{\Tilde{\cS}_{2\secp}}$ makes $T$ queries. Thus, from \cref{lem:Lipschitz_swap_oracle}, we can view that $\Pr_{\ket{\vartheta_1},...,\ket{\vartheta_{2^\secp}}\gets\sigma_{2^{2\secp}}}
        [1\gets\cA^{\cT_{\{\ket{\vartheta_x}\}},\Tilde{\cS}_{2\secp}}]$ is an $8T$-Lipschitz function in the $\ell^2$-sum of Frobenious norm.
    \end{itemize}
    Thus, from \cref{lem:f+g_Lipschitz}, we can view $\Adv(\cA,\Tilde{\cS}_{2\secp})$ is a function that maps $\Unitaries(2^{2\secp})^{2^{2\secp}}\to\R$ and it is a $16T$-Lipschitz in the $\ell^2$-sum of Frobenious norm. Therefore,
    \begin{align}
        &\Pr_{\Tilde{\cS}_{2\secp}\gets\mu}[\Adv(\cA,\Tilde{\cS}_{2\secp})\ge p]\notag\\
        \le&\Pr_{\Tilde{\cS}_{2\secp}\gets\mu}
         \bigg[
          \Adv(\cA,\Tilde{\cS}_{2\secp})\ge \Exp_{\Tilde{\cS'}_{2\secp}\gets\mu}[\Adv(\cA,\Tilde{\cS'}_{2\secp})]
          +p-cT^22^{-\secp}
         \bigg]\tag{By \cref{lem:indistinguishability_of_swap_oracle}}\\
         \le&\exp
         \bigg(
          -\frac{(2^{2\secp}-2)(p-cT^22^{-\secp})^2}{6144T^2}
         \bigg)\tag{By \cref{thm:Haar_concentration} with $d=2^{2\secp},\delta=p-cT^22^{-\secp}, L=16T$}.
    \end{align}
    We can obtain the same bound for $\Pr_{\Tilde{\cS}_{2\secp}\gets\mu}[\Adv(\cA,\Tilde{\cS}_{2\secp})\le -p]$ in the same way, so we obtain our claim.
\end{proof}

Now we are ready to prove \cref{thm:PRFSGs_relative_to_unitary_oracle}. We restate it here for the reader's convenience.

\PRFSGs*

\begin{proof}[Proof of \cref{thm:PRFSGs_relative_to_unitary_oracle}]
    First, we recall \cref{const:PRFSGs}: For a fixed $\cO=(\cS,\cU)$ and inputs $k,x\in\bit^\secp$, $G^\cO(k,x)=\ket{\psi_{2\secp,(k,x)}}$. From \cref{const:PRFSGs}, it is clear $G$ is a QPT algorithm. Thus, it suffices to show that $G$ satisfies quantumly-accessible adaptive security.

    From \cref{remark:considering_only_forward_query}, for the oracle $\cS$, it suffices to consider the forward query.
    Recall that, for fixed $\cO$ and $k\in\bit^\secp$, quantum query to $G^\cO(k,\cdot)$ is defined as follows: let $\sum_x\alpha_x\ket{x}_\regX\ket{\xi_x}_\regZ$ be an overall state of the adversary $\cA$. When $\cA$ queries the register $\regX$, return $\sum_x\alpha_x\ket{x}_\regX\ket{\psi_{2\secp,(k,x)}}_\regY\ket{\xi_x}_\regZ$.
    Similarly, when $\cA$ queries quantumly the register $\regX$ to the ideal oracle $\cH_{\{\ket{\vartheta_x}\}}$, it returns $\sum_x\alpha_x\ket{x}_\regX\ket{\vartheta_{x}}_\regY\ket{\xi_x}_\regZ$.

    To show the security, we want to invoke \cref{lem:Adv_is_negligible_whp}.
    However, since the above oracles $G^\cO(k,\cdot)$ and $\cH_{\{\ket{\vartheta_x}\}}$ are different from $\cT_{2\secp,k}$ and $\cT_{2\secp,\{\ket{\vartheta_x}\}}$, we cannot invoke \cref{lem:Adv_is_negligible_whp} directly.
    For that purpose,
    we construct an algorithm $\cB^{(\cdot),\cS_{2\secp}}$ with the query access to $\cT_{2\secp,k}$ or $\cT_{2\secp,\{\ket{\vartheta_x}\}}$ that simulates $\cA^{(\cdot),\cS_{2\secp}}$ with the query access to $G^\cO(k,\cdot)$ or $\cH_{\{\ket{\vartheta_x}\}}$, respectively:
    \begin{enumerate}
        \item $\cB$ simulates $\cA$. When $\cA$ queries the first oracles, $\cB$ simulates as follows.
        \begin{itemize}
            \item When $\cA$ queries the register $\regX$ to $\cS_{2\secp}$, $\cB$ queries it to $\cS_{2\secp}$.
            \item When $\cA$ queries the register $\regX$ to the first oracle ($G^\cO(k,\cdot)$ or $\cH_{\{\ket{\vartheta_x}\}}$), $\cB$ prepares $\ket{0}_\regA\ket{0^{2\secp}}_\regY$. Then, $\cB$ queries the registers $\regX,\regA$ and $\regY$ to the first oracle ($\cT_{2\secp,k}$ or $\cT_{2\secp,\{\ket{\vartheta_x}\}}$), and removes the register $\regA$. 
        \end{itemize}
         \item $\cB$ outputs what $\cA$ outputs.
    \end{enumerate}
    It is clear that 
    \begin{align}
        \Pr_{k\gets\bit^\secp}[1\gets\cB^{\cT_{2\secp,k},\cS_{2\secp}}]
        =\Pr_{k\gets\bit^\secp}[1\gets\cA^{G^\cO(k,\cdot),\cS_{2\secp}}]
    \end{align}
    and
    \begin{align}
        \Pr_{\ket{\vartheta_1},...,\ket{\vartheta_{2^\secp}}\gets\sigma_{2^{2\secp}}}[1\gets\cB^{\cT_{\{\ket{\vartheta_x}\}},\cS_{2\secp}}]
        =\Pr_{\ket{\vartheta_1},...,\ket{\vartheta_{2^\secp}}\gets\sigma_{2^{2\secp}}}[1\gets\cA^{\cH_{\{\ket{\vartheta_x}\}},\cS_{2\secp}}].
    \end{align}
    Thus, for any adversary $\cA$ with classical advice $y$,
    \begin{align}
        &\Pr_{\cS_{2\secp}\gets\sigma}
        \bigg[\bigg|
        \Pr_{k\gets\bit^\secp}[1\gets\cA^{\ket{G(k,\cdot)},\cS_{2\secp}}(1^\secp,y)]-\Pr_{\ket{\vartheta_1},...,\ket{\vartheta_{2^\secp}}\gets\sigma_{2^{2\secp}}}[1\gets\cA^{\cH_{\{\ket{\vartheta_x}\}},\cS_{2\secp}}(1^\secp,y)]
        \bigg|\ge\frac{1}{2^{\secp/2}}
        \bigg]\notag\\
        =&\Pr_{\cS_{2\secp}\gets\sigma}
        \bigg[\bigg|
        \Pr_{k\gets\bit^\secp}[1\gets\cB^{\cT_{2\secp,k},\cS_{2\secp}}(1^\secp,y)]
        -\Pr_{\ket{\vartheta_1},...,\ket{\vartheta_{2^\secp}}\gets\sigma_{2^{2\secp}}}[1\gets\cB^{\cT_{\{\ket{\vartheta_x}\}},\cS_{2\secp}}(1^\secp,y)]
        \bigg|\ge\frac{1}{2^{\secp/2}}
        \bigg]
        \notag
        \\
        \le&2\exp\bigg(
         -\frac{(2^{2\secp}-2)(2^{-\secp/2}-cT^22^{-\secp})^2}{6144T^2}
        \bigg)\tag{By \cref{lem:Adv_is_negligible_whp} with $p=2^{-\secp/2}$}\\
        \le&\exp\bigg(
        -O\bigg(
         \frac{2^\secp}{T^2}
        \bigg)
        \bigg).\label{eq:PRFSGs_adversary_cannot_win_with_fixed_advice}
    \end{align}
    Hence, for any $T$-query adversary $\cA$ and any polynomial $q$,
    \begin{align}
        &\Pr_{\cS_{2\secp}\gets\sigma}
        \bigg[\text{there exists a }y\in\bit^q \text{ s.t. }
        \notag\\
        &\;\;\bigg|
        \Pr_{k\gets\bit^\secp}[1\gets\cA^{\ket{G(k,\cdot)},\cS_{2\secp}}(1^\secp,y)]-\Pr_{\ket{\vartheta_1},...,\ket{\vartheta_{2^\secp}}\gets\sigma_{2^{2\secp}}}[1\gets\cA^{\cH_{\{\ket{\vartheta_x}\}},\cS_{2\secp}}(1^\secp,y)]
        \bigg|\ge\frac{1}{2^{\secp/2}}
        \bigg]\notag\\
        \le&2^q
        \exp\bigg(
        -O\bigg(
         \frac{2^\secp}{T^2}
        \bigg)
        \bigg)\tag{By the union bound and \cref{eq:PRFSGs_adversary_cannot_win_with_fixed_advice}}\\
        \le&\negl(\secp).
    \end{align}
    Since all $\cS_n$ with $n\neq2\secp$ and $\cU$ are independent of $\cS_{2\secp}$, the above inequality holds even if $\cA$ queries $\cS_n$ with $n\neq2\secp$ and $\cU$. Therefore, by applying the Borel-Cantelli lemma (\cref{lem:Borel-Cantelli}) with $\sum_\secp\negl(\secp)\le\sum_\secp\secp^{-2}\le\infty$, no adversaries with classical advice can distinguish \cref{const:PRFSGs} from independent Haar random states with at least $2^{-\secp/2}$ advantage relative to $\cO$ with probability $1$ over the randomness of $\cO$. This concludes the proof.
\end{proof}

\begin{remark}
    As in the case of \cite{TQC:Kre21}, it is not clear whether \cref{const:PRFSGs} is secure even against adversaries with quantum advice or not. In \cite{TQC:Kre21}, he gives an idea to extend the security proof against adversaries with quantum advice.
    The idea seems to work even in our case, but we leave it to future work.
    For details of the idea, see \cite{TQC:Kre21}.
\end{remark}

\fi
\section{Breaking PRUs}
\label{sec:PRUs}

In this section, we show that
with probability $1$ over $\cO$, non-adaptive and $O(\log\secp)$-ancilla PRUs do not exist relative to $\cO$.
For its proof, we construct a QPT adversary that breaks PRUs relative to $\cO$. 

\subsection{Construction of Adversary}
\label{subsec:adversary_breaking_PRUs}
Let $G^\cO$ be a QPT algorithm that satisfies the correctness condition for $c$-ancilla PRUs. 
In particular, we consider the case $c(\secp) = O(\log \secp)$.
For such $G^\cO$,
let $\{U_k\}_{k\in\cK_\secp}$ be the unitary implemented by $G^\cO$ on input $k\in\cK_\secp$, where $\cK_\secp$ denotes the key-space.
We define the following map:
\begin{align}
    \cM_{\{U_k\},\ell}(\cdot)=\Exp_{k\gets\cK_\secp}U^{\otimes\ell}_k(\cdot)U_k^{\dag\otimes\ell}.
\end{align}
Before constructing the adversary, we introduce some lemmas. 

\begin{restatable}{lemma}{ChoiStates}\label{lem:approx_random_Choi_state}
    Let $T(\secp)$ be a polynomial.
    Let $\epsilon>0$ and $c,d\in\N$.
    Let $\{U_k\}_{k\in\cK_\secp}$ be an ensemble of $\secp$-qubit unitaries each of which is QPT implementable by making $T$ queries to $\cO$ and using $c$ ancilla qubits, where $\cO$ is defined in \cref{def:unitary_orcle}. 
    For all $n\in[d]$, let $\cS'_{n}$ be any unitary satisfying 
    \begin{align}
        \|\cS_{n}(\cdot)\cS_{n}^\dag-\cS_{n}'(\cdot)\cS_{n}^{'\dag}\|_\diamond\le\epsilon.
    \end{align}
    Then, for any polynomial $\ell$, there exists a family 
    $\{V_k\}_{k\in\cK_\secp}$ 
    of $(\secp+c)$-qubit unitaries 
    such that each $V_k$ is QPT implementable with classical descriptions of $\cS'_n$ for all $n\in[d]$ and query access to the $\unitaryPSPACE$-complete oracle $\cU$ such that it satisfies
    \begin{align}
        &\bigg\|
         (\cM_{\{U_k\},\ell}\otimes \identitymap)(\ketbra{\Omega_{2^{\secp\ell}}}{\Omega_{2^{\secp\ell}}})-
         (\cE_{\{V_k\},\ell}\otimes \identitymap)(\ketbra{\Omega_{2^{\secp\ell}}}{\Omega_{2^{\secp\ell}}})
        \bigg\|_1
        \le O(\ell T\epsilon)+O\bigg(\frac{2^{c/2}\ell T}{2^{d/2}}\bigg).
        \label{eq:goal;approx_random_Choi_state;ancilla}
    \end{align}
    Here, $\cE_{\{V_k\},\ell}$ is a CPTP map acting on $\secp\ell$ qubits defined as follows:
    \begin{align}
        \cE_{\{V_k\},\ell}((\cdot)_\regA)\coloneqq
        \Tr_{\regB}
        \bigg[
         \Exp_{k\gets\cK_\secp}V^{\otimes\ell}_{k,\regA\regB}((\cdot)_\regA\otimes\ketbra{0^{c}}{0^c}^{\otimes\ell}_{\regB})
         V^{\dag\otimes\ell}_{k,\regA\regB}
        \bigg],
    \end{align}
    where $\regA$ is a $\secp\ell$-qubit register, and $\regB$ is a $c\ell$-qubit register.
\end{restatable}


\begin{restatable}{lemma}{PRUdistinguisher}\label{lem:PRU_distinguisher}
    Let $\{V_k\}_{k\in\cK_\secp}$ be a family of $(\secp+c)$-qubit unitaries in \cref{lem:approx_random_Choi_state}.
    Let $\ell(\secp)\coloneqq\lceil\log|\cK_\secp|\rceil$.
    Then, there exists a QPT algorithm $\cD^{\cU}$ that, on input classical descriptions of $\cS'_n$ for all $n\in[d]$, distinguishes $(\cE_{\{V_k\},\ell}\otimes \identitymap)(\ketbra{\Omega_{2^{\secp\ell}}}{\Omega_{2^{\secp\ell}}})$ from $(\cM_{\{U_k\},\ell,}\otimes \identitymap)(\ketbra{\Omega_{2^{\secp\ell}}}{\Omega_{2^{\secp\ell}}})$
    with advantage at least $1-\negl(\secp)$.
\end{restatable}

With these ingredients at hand, we can now give an adversary to break PRUs, which implies the second item in \cref{thm:main}:


\begin{algorithm}
    \caption{Adversary distinguishing $\{U_k\}_{k\in\cK_\secp}$ from Haar random unitary.}
    \label{alg:break_PRU} 
    \vspace{2mm}
    \AlgOracle 
    The algorithm has query access to
    \begin{itemize}
        \item a unitary $W\in\Unitaries(2^\secp)$ which is whether $U_k$ or a Haar random unitary.
        
        \item the oracle $\cO=(\cS,\cU)$ and its inverse defined in \cref{def:unitary_orcle};
        
    \end{itemize}
    \AlgInput The algorithm takes the security parameter $1^\secp$ as input.

    Define $\ell\coloneqq\lceil\log|\cK_\secp|\rceil$ and $d\coloneqq 2\log(\ell Tp)+c$, where $p$ is a polynomial, and $T$ is the number of queries to $\cO$ to implement $U_k$.

    \vspace{2mm}
    \begin{enumerate}
        \item \label{algstep:process_tomography}
        For $n\in[d]$, run the process tomography algorithm in \cref{thm:process_tomography_HKOT23} on inputs $\epsilon\coloneqq\frac{1}{\ell Tp}$ and $\eta\coloneqq2^{-\secp-1}$ for $\cS_{n}$ to get a classical description of $\cS'_{n}$. Note that $\|\cS_{n}(\cdot)\cS_{n}^\dag-\cS'_{n}(\cdot)\cS'^\dag_{n}\|_\diamond\le\epsilon$ holds with probability at least $1-2^{-\secp}$ over the randomness of the process tomography algorithm.
        
        \item Prepare $(U^{\otimes\ell}\otimes I)\ket{\Omega_{2^{\secp\ell}}}$ by querying $U$.\label{algstep:prepare_Choi_state}

        \item \label{algstep:measurement}
        Let $\{V_k\}_{k\in\cK_\secp}$ be a family of unitaries in \cref{lem:approx_random_Choi_state}. Note that each $V_k$ is QPT implementable with access to $\cU$ and classical descriptions of $S'_{n}$ for all $n\in[d]$, where such classical descriptions are obtained in the step \ref{algstep:process_tomography}.
        Let $\cD^{(\cdot)}$ be a QPT algorithm in \cref{lem:PRU_distinguisher} for $\{V_k\}_{k\in\cK_\secp}$. 
        By querying $\cU$, run $\cD^\cU$ on input $(U^{\otimes\ell}\otimes I)\ket{\Omega_{2^{\secp\ell}}}$ and classical descriptions of $S'_{n}$ for all $n\in[d]$ to get $b\in\bit$. 
    \end{enumerate} 

    \AlgOutput The algorithm outputs $b$.
\end{algorithm}

\begin{restatable}{theorem}{BreakPRU}\label{thm:break_PRU}
Consider $c(\secp)=O(\log\secp)$
Let $\cO$ be a fixed oracle defined in \cref{def:unitary_orcle}.
Let $G^\cO$ be a QPT algorithm that satisfies the correctness of $c(\secp)$-ancilla PRU.
For such $G^\cO$,
let $\{U_k\}_{k\in\cK_\secp}$ be the unitary implemented by $G^\cO$ on input $k\in\cK_\secp$, where $\cK_\secp$ denotes the key-space.
Then, for any polynomial $p$, there exists a QPT adversary $\cA^{(\cdot,\cdot)}$ such that
    \begin{align}
        \bigg|\Pr_{k\gets\cK_\secp}[1\gets\cA^{U_k,\cO}(1^\secp)]-\Pr_{U\gets\mu_{2^\secp}}[1\gets\cA^{U,\cO}(1^\secp)]\bigg|
        \ge 1-O\bigg(\frac{1}{p(\secp)}\bigg).
    \end{align}
    Moreover, $\cA^{(\cdot),\cO}$ queries the first oracle non-adaptively.
\end{restatable}

\ifnum\submission=1
We give an adversary in \cref{alg:break_PRU}.
Since the proof follows from a straightforward calculation, we give the proof in the supplemental material.
\fi

\ifnum\submission=0
\begin{proof}[Proof of \cref{thm:break_PRU}]
    We construct $\cA^{(\cdot,\cdot)}$ as in \cref{alg:break_PRU}.
    It is clear that $\cA^{(\cdot,\cdot)}$ is a QPT algorithm. Step~\ref{algstep:process_tomography} runs in QPT because $2^d = (\ell T p)^2 + 2^c \leq \poly(\secp)$, given that $c = O(\log \secp)$.  
    Steps~\ref{algstep:prepare_Choi_state} and \ref{algstep:measurement} also run in QPT, as established in \cref{lem:PRU_distinguisher}.

    Assume that the tomography is successful in the step \ref{algstep:process_tomography}, namely, we have 
    $\|\cS_{n}(\cdot)\cS_{n}^\dag-\cS'_{n}(\cdot)\cS'^\dag_{n}\|_\diamond\le\epsilon$ for all $n\in[d]$. For notational simplicity, let $E$ denote this event. Note that
    \begin{align}
        \Pr[E]\ge1-2^{-\secp}\ge1-\negl(\secp)\label{eq:prob_tomography_succeeds}
    \end{align}
    from \cref{thm:process_tomography_HKOT23}. Thus, it suffices to show that $\cA$ can distinguish $\{U_k\}_k$ from Haar random unitaries when the tomography succeeds.
    Recall that $\{V_k\}_k$ is a family of unitaries in the step \ref{algstep:measurement}, and $\cD^{(\cdot)}$ is a QPT algorithm in the step \ref{algstep:measurement} for $\{V_k\}_k$.
    Note that
    \begin{align}
        \Pr_{k\gets\cK_\secp}[1\gets\cA^{U_k,\cO}(1^\secp)|E]
        =\Pr\bigg[
         1\gets\cD^\cU\bigg(
          (\cM_{\{U_k\},\ell}\otimes \identitymap)(\ketbra{\Omega_{2^{\secp\ell}}}{\Omega_{2^{\secp\ell}}})
         \bigg)
        \bigg]\label{eq:probability_when_querying_PRU}
    \end{align}
    and
    \begin{align}
        \Pr_{U\gets\mu_{2^\secp}}[1\gets\cA^{U,\cO}(1^\secp)|E]
        =\Pr\bigg[
         1\gets\cD^\cU\bigg(
         (\cM_{\mu_{2^\secp,\ell}}\otimes \identitymap)(\ketbra{\Omega_{2^{\secp\ell}}}{\Omega_{2^{\secp\ell}}})
         \bigg)
        \bigg].\label{eq:probability_when_querying_Haar}
    \end{align}
    We prove our claim by the standard hybrid argument. 
    First, we replace $\rho_0\coloneqq(\cM_{\{U_k\},\ell}\otimes \identitymap)(\ketbra{\Omega_{2^{\secp\ell}}}{\Omega_{2^{\secp\ell}}})$ in \cref{eq:probability_when_querying_PRU} with $\rho_1\coloneqq(\cE_{\{V_k\},\ell}\otimes \identitymap)(\ketbra{\Omega_{2^{\secp\ell}}}{\Omega_{2^{\secp\ell}}})$. We obtain the following claim.
    
    \begin{claim}\label{claim:hyb0_to_hyb1}
        \begin{align}
            &\bigg|
             \Pr\bigg[
              1\gets\cD^\cU(\rho_0)
             \bigg]
             -\Pr\bigg[
              1\gets\cD^\cU(\rho_1)
             \bigg]
            \bigg|
            \le O\bigg(\frac{1}{p(\secp)}\bigg).
            \notag
        \end{align}
    \end{claim}

    \begin{proof}[Proof of \cref{claim:hyb0_to_hyb1}]
        From \cref{lem:approx_random_Choi_state}, we have
         \begin{align}
         \frac{1}{2}\|
          \rho_0-
          \rho_1
         \|_1
         \le O(\ell T\epsilon)+O\bigg(\frac{2^{c/2}\ell T}{2^{d/2}}\bigg)
         \le O\bigg(\frac{1}{p(\secp)}\bigg).
         \notag
         \end{align}
         Here, $\epsilon=\frac{1}{\ell Tp}$ and $d=2\log(\ell Tp)+c$. 
    \end{proof}

    Next, we replace $\rho_1$ with 
    $\rho_2\coloneqq(\cM_{\mu_{2^\secp},\ell}\otimes \identitymap)(\ketbra{\Omega_{2^{\secp\ell}}}{\Omega_{2^{\secp\ell}}})$. 
    We have the following claim. Because its proof is straightforward from \cref{lem:PRU_distinguisher}, we omit it.

    \begin{claim}\label{claim:hyb1_to_hyb2}
        \begin{align}
            &\bigg|\Pr\bigg[1\gets\cD^\cU(\rho_1)\bigg]-
              \Pr\bigg[1\gets\cD^\cU(\rho_2))\bigg]
            \bigg|
            \ge 1-\negl(\secp).
            \notag
        \end{align}
    \end{claim}

    Combing \cref{eq:prob_tomography_succeeds,eq:probability_when_querying_PRU,eq:probability_when_querying_Haar} with \cref{claim:hyb0_to_hyb1,claim:hyb1_to_hyb2}, we have
    \begin{align}
        &\bigg|\Pr_{k\gets\cK_\secp}[1\gets\cA^{U_k,\cO}(1^\secp)]-\Pr_{U\gets\mu_{2^\secp}}[1\gets\cA^{U,\cO}(1^\secp)]\bigg|
        \notag\\
        \ge&\Pr[E]\bigg|\Pr_{k\gets\cK_\secp}[1\gets\cA^{U_k,\cO}(1^\secp)|E]-\Pr_{U\gets\mu_{2^\secp}}[1\gets\cA^{U,\cO}(1^\secp)|E]\bigg|
        -2\Pr[\Bar{E}]
        \notag\\
        \ge&(1-\negl(\secp))
        \bigg(
         1-O\bigg(\frac{1}{p(\secp)}\bigg)-\negl(\secp)
        \bigg)
        -\negl(\secp)
        \notag\\
        \ge& 1-O\bigg(\frac{1}{p(\secp)}\bigg),
    \end{align}
    which concludes the proof.
\end{proof}
\fi

\subsection[Proof of Lemma~\ref{lem:approx_random_Choi_state}]{Proof of \cref{lem:approx_random_Choi_state}}
\label{subsec:proof_of_approx_random_Choi_state}

In this subsection, we prove \cref{lem:approx_random_Choi_state}. 
To show it, we remove queries to $\cS_n$ for all small $n$. 
To this end, we need the following lemma.
\ifnum\submission=1
Since the proof follows from a straightforward calculation, we give the proof in the supplemental material.
\fi

\begin{restatable}{lemma}{SwapChoi}\label{lem:swap_unitary_is_almost_identity_for_Choi_state;ancilla}
    Let $c'\in\N$.
    Let $2n+1\le\secp+c'$.
    Consider $\cS_n$ defined in~\cref{def:unitary_orcle}. 
    Then, for any $U,V\in\Unitaries(2^{\secp+c'})$,
    \begin{align}
        \frac{1}{2}\|\ketbra{\psi}{\psi}-\ketbra{\phi}{\phi}\|_1 \le O\bigg(\frac{2^{c'/2}}{2^{n/2}}\bigg),
    \end{align}
    where 
    \begin{align}
        \ket{\psi}_{\regA\regA'\regB} & \coloneqq ( U ( \cS_{n} \otimes I ) V )_{\regA\regB} \otimes I_{\regA'} \cdot \ket{\Omega_{2^\secp}}_{\regA\regA'}\ket{0^{c'}}_{\regB}, \\
        \ket{\phi}_{\regA\regA'\regB} & \coloneqq ( U V )_{\regA\regB} \otimes I_{\regA'} \cdot \ket{\Omega_{2^\secp}}_{\regA\regA'} \ket{0^{c'}}_{\regB}.
    \end{align}
    Here, $\regA,\regA'$ are $\secp$-qubit registers, respectively, and $\regB$ is a $c'$-qubit register.
\end{restatable}

\ifnum\submission=0
\begin{proof}[Proof of \cref{lem:swap_unitary_is_almost_identity_for_Choi_state;ancilla}]
    Define the following states:
    \begin{align}
        \ket{\psi'}_{\regA\regB\regA'\regB'} & \coloneqq ( U ( \cS_{n} \otimes I ) V )_{\regA\regB} \otimes I_{\regA'\regB'} \cdot \ket{\Omega_{2^{\secp+c'}}}_{\regA\regB\regA'\regB'}, \\
        \ket{\phi'}_{\regA\regB\regA'\regB'} & \coloneqq ( U V )_{\regA\regB} \otimes I_{\regA'\regB'} \cdot \ket{\Omega_{2^{\secp+c'}}}_{\regA\regB\regA'\regB'},
    \end{align}
    where $\regB'$ is a $c'$-qubit register and $\ket{\Omega_{2^{\secp+c'}}}$ is across $(\regA\regB,\regA'\regB')$.

    We can prove the following inequality, which we will prove later.
    \begin{align}
        \|\ket{\psi'} - \ket{\phi'}\| \le O \bigg(\frac{1}{2^{n/2}}\bigg).
        \label{eq:closeness}
    \end{align}
    Moreover, by definition, it holds that
    \begin{align}
        \ket{\psi}_{\regA\regA'\regB} & = \sqrt{2^{c'}} \cdot I_{\regA\regA'\regB} \otimes \bra{0^{c'}}_{\regB'} \cdot \ket{\psi'}_{\regA\regB\regA'\regB'}, 
        \notag \\
        \ket{\phi}_{\regA\regA'\regB} & = \sqrt{2^{c'}} \cdot I_{\regA\regA'\regB} \otimes \bra{0^{c'}}_{\regB'} \cdot \ket{\phi'}_{\regA\regB\regA'\regB'}.
        \label{eq:post_selection}
    \end{align}
    Putting them together, we obtain
    \begin{align}
        & \, \frac{1}{2}\|\ketbra{\psi}{\psi}-\ketbra{\phi}{\phi}\|_1 
        \notag \\
        \le & \, \|\ket{\psi} - \ket{\phi}\| 
        \tag{Since trace distance between pure states is bounded by their Euclidean distance} \\
        = & \, \sqrt{2^{c'}} \cdot \| I \otimes \bra{0^{c'}} \cdot (\ket{\psi'} - \ket{\phi'})\| 
        \tag{By~\cref{eq:post_selection}} \\
        \le & \, \sqrt{2^{c'}} \cdot \| I \otimes \bra{0^{c'}} \|_{\infty} \cdot \|\ket{\psi'} - \ket{\phi'}\| 
        \tag{By~$\|{A \ket{v}} \| \le \|A\|_\infty \|\ket{v}\|$} \\
        = & \, O\bigg(\frac{2^{c'/2}}{2^{n/2}}\bigg)
        \tag{By~$\| I \otimes \bra{0^{c'}} \|_{\infty} = 1$ and~\cref{eq:closeness}}
    \end{align}
    as desired.

    To conclude the proof, we prove \cref{eq:closeness}.
    Define the projection 
    \begin{align}\label{eq:projection_subspace}
        \Pi_{n}\coloneqq\sum_{m\in\bit^{n}}\ketbra{m}{m}\otimes I_\bot^{n,m}
    \end{align} 
    onto the subspace on which $\cS_{n}$ acts as the identity. It then follows that
    \begin{align}
        (I - \cS_{n}) \Pi_n = 0.
    \end{align}
    First, we can see 
    $(\cS_{n}\otimes I^{\otimes 2\secp+2c'-2n-1})\ket{\Omega_{2^{\secp+c}}}$
    is close to $\ket{\Omega_{2^{\secp+c}}}$ in Euclidean distance as follows:
    \begin{align}
        & \|\ket{\Omega_{2^{\secp+c'}}} - (\cS_{n}\otimes I^{\otimes 2\secp+2c'-2n-1}) \ket{\Omega_{2^{\secp+c'}}}\| \\
        = & \|(I^{\otimes 2n+1} - \cS_{n})\otimes I^{\otimes 2\secp+2c'-2n-1} \ket{\Omega_{2^{\secp+c'}}}\| \\
        = & \|(I^{\otimes 2n+1} - \cS_{n}) (I-\Pi_n) \otimes I^{\otimes 2\secp+2c'-2n-1} \ket{\Omega_{2^{\secp+c'}}}\| 
        \tag{By~\cref{eq:projection_subspace}} \\
        \le & \|(I^{\otimes 2n+1} - \cS_{n})\otimes I^{\otimes 2\secp+2c'-2n-1}\|_\infty \cdot \|(I^{\otimes 2n+1}-\Pi_n)\otimes I^{\otimes 2\secp+2c'-2n-1} \ket{\Omega_{2^{\secp+c'}}}\| 
        \tag{By~$\|A\ket{v}\| \le \|A\|_\infty \|\ket{v}\|$} \\
        \le & 2 \|(I^{\otimes 2n+1}-\Pi_n)\otimes I^{\otimes 2\secp+2c'-2n-1} \ket{\Omega_{2^{\secp+c'}}}\| 
        \tag{By the triangle inequality and the fact that unitaries have unit operator norm} \\
        =& 2\sqrt{\frac{1}{2^{\secp+c'}}\Tr[(I^{\otimes 2n+1}-\Pi_n)\otimes 2^{\secp+c'-2n-1}]}
        \tag{By $||(A\otimes I)\ket{\Omega_D}||^2=\frac{1}{D}\Tr[A^\dag A]$ for any $A\in\Unitaries(D)$}
        \\
        =& 2\sqrt{\frac{1}{2^{2n+1}}\Tr[I^{\otimes 2n+1}-\Pi_n]}
        \tag{By $\Tr[A\otimes B]=\Tr[A]\Tr[B]$ forn any matrix $A$ and $B$}
        \\
        =& 2\sqrt{\frac{1}{2^{2n+1}}\sum_{m\in\bit^n}\Tr[\ketbra{m}{m}\otimes(I^{\otimes n+1}-I^{n,m}_\bot)]}
        \tag{By \cref{eq:projection_subspace}}
        \\
        \le & O(2^{n/2}). 
    \end{align}
    Here, in the last line, we have used that $\Tr[I^{\otimes n+1}-I^{n,m}_\bot]=2$.
    Since $\ket{\psi'}_{\regA\regB\regA'\regB'}=(U_{\regA\regB}\otimes V^\top_{\regA'\regB'})((\cS_{n}\otimes I)_{\regA\regB}\otimes I_{\regA'\regB'})\ket{\Omega_{2^{\secp+c}}}_{\regA\regB\regA'\regB'}$ and $\ket{\phi'}_{\regA\regB\regA'\regB'}=(U_{\regA\regB}\otimes V^\top_{\regA'\regB'})\ket{\Omega_{2^{\secp+c}}}_{\regA\regB\regA'\regB'}$, we obtain \cref{eq:closeness}, which concludes the proof.
\end{proof}
\fi

Having this lemma, we are ready to show \cref{lem:approx_random_Choi_state}.
\ifnum\submission=0
For the reader's convenience, we restate it now:

\ChoiStates*
\fi

\begin{proof}[Proof of \cref{lem:approx_random_Choi_state}]
    For each $k \in \cK_\secp$, we may view $G^\cO(k,\cdot)$ as acting as follows: first, it prepares $\ket{0^c}$ in the ancilla register, then applies a $(\secp + c)$-qubit unitary $W_k$ to the input and ancilla qubits by querying $\cO$, and finally discards the ancilla register.
    
    We define a $(\secp+c)$-qubit unitary $V_k$ as follows. 
    Take the same circuit as in the implementation of $W_k$, except at the points where it queries $\cS_{n}$. 
    Whenever the circuit queries $\cS_{n}$,
          \begin{itemize}
              \item if $n\in[d]$, apply $\cS'_{n}$ by using its classical description;
              \item if $n\in[\secp+c]/[d]$, do not apply any unitary circuit. 
          \end{itemize}
    The unitary $V_k$ is then the unitary implemented by these modified procedures. 
    Thus, it is clear that $V_k$ is QPT implementable with classical descriptions of $\cS'_n$ for all $n\in[d]$ and query access to the $\unitaryPSPACE$-complete oracle $\cU$.
    
    We prove~\cref{eq:goal;approx_random_Choi_state;ancilla} by the standard hybrid argument.
    Define a unitary $\widetilde{V}_k$ as follows. Take the same circuit as in the implementation of $W_k$, except at the points where it queries $\cS_{n}$. Whenever the circuit queries $\cS_{n}$,
          \begin{itemize}
              \item if $n\in[d]$, apply $\cS_{n}$;
              \item if $n\in[\secp+c]/[d]$, do not apply any unitary. 
          \end{itemize}
    The unitary $\widetilde{V}_k$ is then the unitary implemented by these modified procedures.
    Define a CPTP map $\cE_{\{\widetilde{V_k}\}_k,\ell}$ acting on $\secp\ell$ qubits in the same manner as $\cE_{\{V_k\}_k,\ell}$.
    We can show
    \begin{align}
        &\bigg\|
         (\cE_{\{V_k\},\ell}\otimes \identitymap)(\ketbra{\Omega_{2^{\secp\ell}}}{\Omega_{2^{\secp\ell}}})
         -(\cE_{\{\widetilde{V_k}\},\ell}\otimes \identitymap)(\ketbra{\Omega_{2^{\secp\ell}}}{\Omega_{2^{\secp\ell}}})
         \bigg\|_1\le O(\ell T\epsilon)
        \label{eq:hyb0;approx_random_Choi_state;ancilla}
    \end{align}
    and
    \begin{align}
        \bigg\|
        (\cE_{\{\widetilde{V_k}\},\ell}\otimes \identitymap)(\ketbra{\Omega_{2^{\secp\ell}}}{\Omega_{2^{\secp\ell}}})
        -(\cM_{\{U_k\},\ell}\otimes \identitymap)(\ketbra{\Omega_{2^{\secp\ell}}}{\Omega_{2^{\secp\ell}}})
        \bigg\|_1\le O\bigg(\frac{2^{c/2}\ell T}{2^{d/2}}\bigg).
        \label{eq:hyb1;approx_random_Choi_state;ancilla}
    \end{align}
    We will give the proofs later. From~\cref{eq:hyb0;approx_random_Choi_state;ancilla,eq:hyb1;approx_random_Choi_state;ancilla} and the triangle inequality, we obtain~\cref{eq:goal;approx_random_Choi_state;ancilla}.
    
    First, we prove \cref{eq:hyb0;approx_random_Choi_state;ancilla}
    For each $k\in\cK_\secp$, the difference between $V_k$ and $\widetilde{V}_k$ lies only whether we apply $\cS'_{n}$ or $\cS_{n}$ for all $n\in[d]$. Since the number of queries to the swap unitaries is at most $T$, we have
    \begin{align}
        \|V_k(\cdot)V_k^{\dag}-\widetilde{V}_k(\cdot)\widetilde{V}_k^\dag\|_\diamond\le T\epsilon\label{eq:U'_k_is_close_to_tilde_U_k}
    \end{align}
    for all $k\in\cK_\secp$, which implies \cref{eq:hyb0;approx_random_Choi_state;ancilla}.
    
    Next, we prove~\cref{eq:hyb1;approx_random_Choi_state;ancilla}. For each $k\in\cK_\secp$, the difference between $W_k$ and $\widetilde{V}_k$ lies only whether we apply $\cS_{n}$ or not for all $n\in[\secp]/[d]$.Note that the number of queries to the swap unitaries is at most $T$.
    Thus, from~\cref{lem:swap_unitary_is_almost_identity_for_Choi_state;ancilla} with $c'=c$, we have
    \begin{align}
        \bigg\|
         (\widetilde{V}_k\otimes I)(\ketbra{0^c}{0^c}\otimes\ketbra{\Omega_{2^{\secp}}}{\Omega_{2^{\secp}}})(\widetilde{V}_k\otimes I)^\dag
         -(W_k\otimes I)(\ketbra{0^c}{0^c}\otimes\ketbra{\Omega_{2^{\secp}}}{\Omega_{2^{\secp}}})(W_k\otimes I)^\dag
        \bigg\|_1\le O\bigg(\frac{2^{c/2}T}{2^{d/2}}\bigg) 
    \end{align}
    for all $k\in\cK_\secp$.
    By combing this inequality and $\|\rho^{\otimes\ell}-\sigma^{\otimes\ell}\|_1\le\ell\|\rho-\sigma\|_1$ for any state $\rho$ and $\sigma$\footnote{The latter follows from the triangle inequality.}, we have
    \begin{align}
        &\bigg\|
         (\widetilde{V}^{\otimes\ell}_{k,\regB\regA}\otimes \identitymap_{\regA'})(\ketbra{0^c}{0^c}^{\otimes\ell}_\regB\otimes\ketbra{\Omega_{2^{\secp\ell}}}{\Omega_{2^{\secp\ell}}}_{\regA\regA'})
         \notag\\
         &-(W^{\otimes\ell}_{k,\regB\regA}\otimes \identitymap_{\regA'})
         (\ketbra{0^c}{0^c}^{\otimes\ell}_\regB\otimes\ketbra{\Omega_{2^{\secp\ell}}}{\Omega_{2^{\secp\ell}}}_{\regA\regA'})
        \bigg\|_1\le O\bigg(\frac{2^{c/2}\ell T}{2^{d/2}}\bigg) 
        \label{eq:V_k_vs_W_k}
    \end{align}
    for all $k\in\cK_\secp$, where $\regA$  and $\regA'$ are $\secp\ell$-qubit register, and $\regB$ is a $c\ell$-qubit rgister.
    Here, note that
    \begin{align}
        \Tr_\regB[(W^{\otimes\ell}_{k,\regB\regA}\otimes \identitymap_{\regA'})
         (\ketbra{0^c}{0^c}^{\otimes\ell}_\regB\otimes\ketbra{\Omega_{2^{\secp\ell}}}{\Omega_{2^{\secp\ell}}}_{\regA\regA'})]
         =U^{\otimes\ell}_{k,\regA}(\ketbra{\Omega_{2^{\secp\ell}}}{\Omega_{2^{\secp\ell}}}_{\regA\regA'}))U^{\dag\otimes\ell}_{k,\regA}
    \end{align}
    since $W_k$ is a purification of $G^\cO(k,\cdot)$.
    Therefore, we obtain~\cref{eq:hyb1;approx_random_Choi_state;ancilla} from \cref{eq:V_k_vs_W_k}.
\end{proof}

\subsection[Proof of Lemma~\ref{lem:PRU_distinguisher}]{Proof of \cref{lem:PRU_distinguisher}}
\label{subsec:proof_of_distinguisher}

To show \cref{lem:PRU_distinguisher}, we need some lemmas.
First lemma ensures that we can implement the block-encoding of $(\cE_{\{V_k\},\ell}\otimes \identitymap)(\ketbra{\Omega_{2^{\secp\ell}}}{\Omega_{2^{\secp\ell}}})$ by querying $\cU$ and by using the classical descriptions of $\cS_n'$.
\ifnum\submission=1
Since the proof follows from a straightforward calculation, we give it in the supplemental material.
\fi

\begin{restatable}{lemma}{BlockEncoding}\label{lem:block-encoding_of_random_Choi_states}
    Let $\{V_k\}_{k\in\cK_\secp}$ be a family of $(\secp+c)$-qubit unitaries that are QPT implementable with classical descriptions of $\cS'_n$ for all $n\in[d]$ and 
    with the query access to $\cU$, where $\cU$ is the $\unitaryPSPACE$ complete problem in \cref{lem:unitaryPSPACE_has_a_complete_problem}.
    Let $\ell(\secp)\coloneqq\lceil\log|\cK_\secp|\rceil$. Then, for any polynomial $p$, there exists a unitary circuit $V_\secp$ satisfying the following:
    \begin{itemize}
        \item $V_\secp$ is QPT implementable with classical descriptions of $\cS'_n$ for all $n\in[d]$ and with the query access to $\cU$.
        \item $V_\secp$ is a $(1,2^{-p(\secp)},\poly(\secp))$-block encoding of $(\cE_{\{V_k\},\ell}\otimes \identitymap)(\ketbra{\Omega_{2^{\secp\ell}}}{\Omega_{2^{\secp\ell}}})$.
    \end{itemize}
\end{restatable}

\ifnum\submission=0
\begin{proof}[Proof of \cref{lem:block-encoding_of_random_Choi_states}]
    Let $x$ denote the concatenation of classical descriptions of $\cS'_n$ for all $n\in[d]$.
    Let $A$ and $B$ be unitaries such that $BA$ is a purification unitary of 
    $(\cM_{\{U'_k\},\ell}\otimes\identitymap)(\ketbra{\Omega_{2^{\secp\ell}}}{\Omega_{2^{\secp\ell}}})$.
    In other words,
    tracing out some qubits of $BA|0...0\rangle$ is
    equal to
    $(\cM_{\{V_k\},\ell}\otimes\identitymap)(\ketbra{\Omega_{2^{\secp\ell}}}{\Omega_{2^{\secp\ell}}})$.
   $A$ is a unitary that maps $|0...0\rangle$ to $|x\rangle$.
   $B$ is the following unitary:
    \begin{enumerate}
        \item First, map $\ket{0...0}\ket{x}\mapsto\frac{1}{\sqrt{|\cK_\secp|}}\sum_{k\in\cK_\secp}\ket{0^c}^{\otimes\ell}\ket{\Omega_{2^{\secp\ell}}}\ket{0...0}\ket{k}\ket{x}$.
        \item 
        Then, map $
        \frac{1}{\sqrt{|\cK_\secp|}}
        \sum_{k\in\cK_\secp}\ket{0^c}^{\otimes\ell}\ket{\Omega_{2^{\secp\ell}}}\ket{0...0}\ket{k}\ket{x}
        \mapsto\frac{1}{\sqrt{|\cK_\secp|}}
        \sum_{k\in\cK_\secp}(V_k^{\otimes\ell}\otimes I)(\ket{0^c}^{\otimes\ell}\ket{\Omega_{2^{\secp\ell}}})\ket{0...0}\ket{k}\ket{x}$.
    \end{enumerate}
    Clearly, $A$ is QPT implementable given $x$. 
    $B$ can be approximately QPT implementable with an exponentially-small error by querying $\cU$, 
    because of the following reason:
    The first step of $B$ is QPT implementable.
    For the second step of $B$, we have only to show that each controlled-$V_k$ is approximately QPT implementable by querying $\cU$.
    In fact, first, $V_k$ is QPT implementable on input $x,k$ and by querying $\cU$.
    Second, in order to implement the controlled-$V_k$, we need the controlled-$\cU$.
    The controlled-$\cU$ is in $\unitaryPSPACE$ from \cref{remark:contrlizatin_of_pureUnitaryPSPACE}, therefore it
    is approximately QPT implementable with an exponentially-small error by querying $\cU$.
    
    Thus, for any polynomial $p$, there exists a QPT algorithm $\cB$ that implements $B$ with error $2^{-p(\secp)}$ by querying $\cU$. Namely, it satisfies
    \begin{align}
        \|\cB^\cU(\cdot)-B(\cdot)B^\dag\|_\diamond\le2^{-p(\secp)}.
    \end{align}
    Since $\cB$ is a QPT algorithm, it queries $\cU$ at most polynomially many times. Thus, by postponing all intermediate measurements, we can assume that $\cB$ applies a QPT unitary $C$ by querying $\cU$.
    Thus, we have
    \begin{align}
        \|\Tr_{\regY}[C_{\regX\regY}((\cdot)_{\regX}\otimes\ketbra{0...0}{0...0}_{\regY})C^\dag_{\regX\regY}]-(B(\cdot)B^\dag)_{\regX}\|_\diamond\le2^{-p(\secp)},\label{eq:approx_of_purification_for_random_Choi_states}
    \end{align}
    where $\regX$ and $\regY$ denote the main register and the ancilla register, respectively.
    Suppose that $\regA$ and $\regA'$ are $\secp\ell$-qubit registers, and $\regB$ is a $c\ell$-qubit register.
    We decompose $\regX$ as $\regX_0\coloneqq\regB\regA\regA'$ and $\regX_1$, where $\regX_0$ is the first $2\secp\ell$ qubits, and $\regX_1$ is the other qubits.
    From \cref{eq:approx_of_purification_for_random_Choi_states} and $\Tr_{\regX_1}[(BA\ketbra{0...0}{0...0}A^\dag B^\dag)_{\regX}]=(\cM_{\{V_k\},\ell,\regB\regA}\otimes\identitymap_{\regA'})(\ketbra{0^c}{0^c}^{\otimes\ell}_\regB\otimes\ketbra{\Omega_{2^{\secp\ell}}}{\Omega_{2^{\secp\ell}}}_{\regA\regA'})$, we have
    \begin{align}
        &\|\Tr_{\regB\regX_1\regY}[C_{\regX\regY}A_{\regX}(\ketbra{0...0}{0...0}_{\regX\regY})A^\dag_\regX C^\dag_{\regX\regY}]
        -(\cE_{\{V_k\},\ell,\regA}\otimes \identitymap_{\regA'})(\ketbra{\Omega_{2^{\secp\ell}}}{\Omega_{2^{\secp\ell}}}_{\regA\regA'})\|_1
        \le2^{-p(\secp)}.\label{eq:approx_of_random_Choi_states}
    \end{align}

    Now we are ready to construct block-encoding of $\Tr_{\regB}[(\cM_{\{V_k\},\ell,\regB\regA}\otimes\identitymap)(\ketbra{0^c}{0^c}^{\otimes\ell}_\regB\otimes\ketbra{\Omega_{2^{\secp\ell}}}{\Omega_{2^{\secp\ell}}})_{\regA\regA'}]$. From \cref{lem:block-encoding_of_state}, there exists a $(1,0,\poly(\secp))$-block encoding unitary $V_\secp$ of
    \begin{align}
        \sigma_{\regA\regA'}\coloneqq\Tr_{\regB\regX_1\regY}[C_{\regX\regY}A_{\regX}(\ketbra{0...0}{0...0}_{\regX\regY})A^\dag_\regX C^\dag_{\regX\regY}].
    \end{align}
    From \cref{lem:block-encoding_of_state} $V_\secp$ can be realized with a single use of $C_{\regX\regY}A_{\regX}$ and its inverse, and $\poly(\secp)$ two-qubit gates. 
    Therefore, $V_\secp$ is QPT implementable with $x$ and with the query access to $\cU$. Moreover, $V_\secp$ satisfies
    \begin{align}
        &\|((\bra{0^{\poly(\secp)}}\otimes I)V_{\secp}(\ket{0^{\poly(\secp)}}\otimes I))_{\regA\regA'}
        -(\cE_{\{V_k\},\ell,\regA}\otimes \identitymap_{\regA'})(\ketbra{\Omega_{2^{\secp\ell}}}{\Omega_{2^{\secp\ell}}}_{\regA\regA'})
        \|_\infty
        \notag\\
        =&
        \|\sigma_{\regA\regA'}
        -(\cE_{\{V_k\},\ell,\regA}\otimes \identitymap_{\regA'})(\ketbra{\Omega_{2^{\secp\ell}}}{\Omega_{2^{\secp\ell}}}_{\regA\regA'})\|_\infty
        \tag{Since $V_\secp$ is an $(1,0,\poly(\secp))$-block encoding of $\sigma$}\\
        \le&
        \|\sigma_{\regA\regA'}
        -(\cE_{\{V_k\},\ell,\regA}\otimes \identitymap_{\regA'})(\ketbra{\Omega_{2^{\secp\ell}}}{\Omega_{2^{\secp\ell}}}_{\regA\regA'})\|_1
        \tag{By $\|A\|_\infty\le\|A\|_1$}\\
        \le&2^{-p(\secp)},
        \tag{By \cref{eq:approx_of_random_Choi_states}}
    \end{align}
    which implies that $V_{\secp}$ is a $(1,2^{-p(\secp)},\poly(\secp))$-block encoding of $(\cE_{\{V_k\},\ell,\regA}\otimes \identitymap_{\regA'})(\ketbra{\Omega_{2^{\secp\ell}}}{\Omega_{2^{\secp\ell}}}_{\regA\regA'})$.
\end{proof}
\fi

The next lemma ensures that $(\cM_{\mu_{2^\secp},\ell,\regA}\otimes\identitymap_{\regA'})(\ketbra{\Omega_{2^{\secp\ell}}}{\Omega_{2^{\secp\ell}}}_{\regA\regA'})$ has negligible overlap with the support of $(\cE_{\{V_k\},\ell}\otimes \identitymap)(\ketbra{\Omega_{2^{\secp\ell}}}{\Omega_{2^{\secp\ell}}})$.
\ifnum\submission=1
Since the proof follows from a straightforward calculation, we give it in the supplemental material.
\fi

\begin{restatable}{lemma}{SmallOverlap}\label{lem:Haar_Choi_has_negligible_overlap}
    Suppose that $c(\secp)=O(\log\secp)$.
    Let $Q$ be the projection onto the support of $(\cE_{\{V_k\},\ell}\otimes \identitymap)(\ketbra{\Omega_{2^{\secp\ell}}}{\Omega_{2^{\secp\ell}}})$. 
    Then,
    \begin{align}
        \Tr[Q(\cM_{\mu_{2^\secp},\ell}\otimes\identitymap)(\ketbra{\Omega_{2^{\secp\ell}}}{\Omega_{2^{\secp\ell}}})]
        \le\negl(\secp).
    \end{align}
\end{restatable}

\ifnum\submission=0
\begin{proof}[Proof of \cref{lem:Haar_Choi_has_negligible_overlap}]
    First, we prove that $\Tr[Q]\le2^{(1+c)\ell}$.
    Note that
    \begin{align}
        &(\cE_{\{V_k\},\ell,\regA}\otimes \identitymap_{\regA'})(\ketbra{\Omega_{2^{\secp\ell}}}{\Omega_{2^{\secp\ell}}}_{\regA\regA'})
        =\frac{1}{|\cK_\secp|}\sum_{k\in\cK_\secp}
        \Tr_\regB[(V^{\otimes\ell}_{k,\regB\regA}\otimes \identitymap_{\regA'})(\ketbra{0^c}{0^c}^{\otimes\ell}_\regB\otimes\ketbra{\Omega_{2^{\secp\ell}}}{\Omega_{2^{\secp\ell}}}_{\regA\regA'})],
    \end{align}
    where $\regA$ and $\regA'$ are $\secp\ell$-qubit registers, and $\regB$ is a $c\ell$-qubit register.
    For each $k$, the rank of $\Tr_\regB[(V^{\otimes\ell}_{k,\regB\regA}\otimes \identitymap_{\regA'})(\ketbra{0^c}{0^c}^{\otimes\ell}_\regB\otimes\ketbra{\Omega_{2^{\secp\ell}}}{\Omega_{2^{\secp\ell}}}_{\regA\regA'})]$ is at most $\min\{2^{c\ell},2^{2\secp\ell}\}=2^{c\ell}$ since $(V^{\otimes\ell}_{k,\regB\regA}\otimes \identitymap_{\regA'})(\ketbra{0^c}{0^c}^{\otimes\ell}_\regB\otimes\ketbra{\Omega_{2^{\secp\ell}}}{\Omega_{2^{\secp\ell}}}_{\regA\regA'})$ is pure.
    Thus, the rank of $Q$ is at most $2^{c\ell}\cdot|\cK_\secp|\le2^{(1+c)\ell}$, which implies $\Tr[Q]\le2^{(1+c)\ell}$.
    
    Having this, 
    \begin{align}
        \Tr[Q(\cM_{\mu_{2^\secp},\ell}\otimes \identitymap)(\ketbra{\Omega_{2^{\secp\ell}}}{\Omega_{2^{\secp\ell}}})]
        \le&
        \Tr[Q\Exp_{\ket{\psi}\gets\sigma_{2^{2\secp}}}\ketbra{\psi}{\psi}^{\otimes\ell}]
        +O\bigg(\frac{\ell^2}{2^\secp}\bigg)
        \tag{By \cref{lem:Haar_state_vs_Haar_choi_state}}
        \\
        =&
        \frac{\Tr[Q\Pi_{\symetric}]}{\binom{2^{2\secp}+\ell-1}{\ell}}+\negl(\secp)
        \notag
        \\
        \le&\frac{2^{(1+c)\ell}}{\binom{2^{2\secp}+\ell-1}{\ell}}+\negl(\secp)
        \tag{By $\Tr[Q\Pi_\symetric]\le\Tr[Q]\le2^{(1+c)\ell}$}\\
        \le&O\bigg(
        \frac{2^{(1+c)\ell}(\ell!)}{2^{2\secp\ell}}
        \bigg)
        +\negl(\secp)
        \notag\\
        \le&O\bigg(
        \frac{2^{(1+c)\ell} \ell^{\ell+1/2}e^{-\ell}}{2^{2\secp\ell}}
        \bigg)
        +\negl(\secp)
        \tag{By the Stirling's formula, $\ell!\le \ell^{\ell+1/2}e^{-\ell+1}$}\\
        =&O\bigg(\ell^{1/2}
        \bigg(\frac{2^{1+c}e^{-1} \ell}{2^{2\secp}}\bigg)^\ell
        \bigg)
        +\negl(\secp)
        \notag\\
        =&O\bigg(\ell^{1/2}
        \bigg(\frac{\poly(\secp)}{2^{2\secp}}\bigg)^\ell
        \bigg)
        +\negl(\secp)
        \tag{By $c(\secp)=O(\log\secp)$}\\
        \le&\negl(\secp),\label{eq:xi_has_negligible_overlap_with_Q}
    \end{align}
    which concludes the proof.
\end{proof}
\fi

The following lemma gives us an algorithm that distinguishes between two states if they are statistically far.

\begin{lemma}\label{lem:distinguisher}
    Suppose that $\rho$ and $\xi$ are $n(\secp)=\poly(\secp)$-qubit states, and satisfy the following:
    \begin{itemize}
        \item For any polynomial $p$, there exists a QPT implementable unitary $V_\secp$ with classical advice $c$ and with query access to $\cU$ such that $V_\secp$ is a $(1,2^{-p(\secp)},\poly(\secp))$-block encoding of $\rho$.
        \item For the projection $Q$ onto the support of $\rho$, $\Tr[Q\xi]\le\negl(\secp)$.
    \end{itemize}
    Then, there exists a QPT algorithm $\cD^\cU$ that, on input $c$, distinguishes $\rho$ from $\xi$ with advantage at least $1-\negl(\secp)$.
\end{lemma}

Before giving the proof, we show  \cref{lem:PRU_distinguisher} with \cref{lem:distinguisher}.
We restate it here for the reader's convenience. 
\PRUdistinguisher*

\begin{proof}[Proof of \cref{lem:PRU_distinguisher}]
    Let $\rho=(\cE_{\{V_k\},\ell}\otimes \identitymap)(\ketbra{\Omega_{2^{\secp\ell}}}{\Omega_{2^{\secp\ell}}})$ and 
    $\xi=(\cM_{\mu_{2^\secp},\ell}\otimes\identitymap)(\ketbra{\Omega_{2^{\secp\ell}}}{\Omega_{2^{\secp\ell}}})$.
    From \cref{lem:block-encoding_of_random_Choi_states,lem:Haar_Choi_has_negligible_overlap}, they satisfy the condition in \cref{lem:distinguisher}.
    Therefore, we obtain \cref{lem:PRU_distinguisher} by applying \cref{lem:distinguisher}, which concludes the proof.
\end{proof}

Finally, we give the proof of \cref{lem:distinguisher}.

\begin{proof}[Proof of \cref{lem:distinguisher}]
    
    Suppose that $\rho$ and $\sigma$ are $n$-qubit state, where $n=\poly(\secp)$.
    Let $V_\secp$ be a $(1,2^{-p(\secp)},\poly(\secp))$-block encoding of $\rho$, where we chose a polynomial $p$ later.
    We construct $\cD^\cU$ from this $V_\secp$. Before that, consider the singular value discrimination algorithm in \cref{thm:singular_value_discrimination} with $M=(\ketbra{0^{\poly(\secp)}}{0^{\poly(\secp)}}\otimes I)V_\secp(\ketbra{0^{\poly(\secp)}}{0^{\poly(\secp)}}\otimes I)$, $\eta=2^{-\secp},a=2^{-3n}$ and $b=2^{-2n}$ by querying $\cU$.
    Based on this, we construct $\cD^\cU$ as follows: 
    \begin{enumerate}
        \item Take $n$-qubit state $\sigma$ as an input.
        \item Simulate the above singular value discrimination algorithm on input $\ketbra{0^{\poly(\secp)}}{0^{\poly(\secp)}}\otimes\sigma$.
        (Because this simulation queries $\cU$, it causes at most a negligible error.)
    \end{enumerate}

    Our goal is to show $\cD^\cU$ distinguishes $\rho$ from $\xi$.
    To use the singular value discrimination algorithm (\cref{thm:singular_value_discrimination}), the input state has to be inside of the promise with high probability.
    In other words, the input state must have a sufficiently large overlap with the subspace $W_0$ or $W_1$.
    Here, $W_0$ is the subspace spanned by the all right singular value vectors of $(\ketbra{0^{\poly(\secp)}}{0^{\poly(\secp)}}\otimes I)V_\secp(\ketbra{0^{\poly(\secp)}}{0^{\poly(\secp)}}\otimes I)$ with singular value is at most $a=2^{-3n}$, $W_1$ is the subspace spanned by the all right singular value vectors of $(\ketbra{0^{\poly(\secp)}}{0^{\poly(\secp)}}\otimes I)V_\secp(\ketbra{0^{\poly(\secp)}}{0^{\poly(\secp)}}\otimes I)$ with singular value is at least $b=2^{-2n}$.
    In the following, we show that $\rho$ 
    has at least $1-\negl(\secp)$ overlap with $W_1$
    and that $\xi$ 
    has at least $1-\negl(\secp)$ overlap with $W_0$.

    \paragraph{Large overlap with $W_1$.}
    First, we show that $\rho$ has a large overlap with $W_1$.
    This follows from that $V_\secp$ is a $(1,2^{-p(\secp)},\poly(\secp))$-block encoding of $\rho$.
    The formal statement is the following.

    \begin{restatable}{lemma}{RSVsubspace}\label{lem:RSV_subspace_and_support}
        Let $\Pi_{\ge\epsilon}$ be the projection onto the subspace spanned by right singular vectors of 
        $(\bra{0^{\poly(\secp)}}\otimes I)V_\secp(\ket{0^{\poly(\secp)}}\otimes I)$ 
        with singular value at least $\epsilon$. If $V_\secp$ is a $(1,2^{-p(\secp)},\poly(\secp))$-block encoding of $\rho$, then 
        \begin{align}
            \Tr[\Pi_{\ge\epsilon}\rho]\ge1-2^{n-p+1}-2^{n}\epsilon.
        \end{align}
    \end{restatable}

    \ifnum\submission=0
    We give its proof later. 
    \fi
    \ifnum\submission=1
    We provide the proof in the supplemental materials.
    \fi
    We define
    \begin{align}
        \rho'\coloneqq\frac{\Pi_{\ge2^{-2n}}\rho\Pi_{\ge2^{-2n}}}{\Tr[\Pi_{\ge2^{-2n}}\rho]},
        \label{eq:def_of_sigma}
    \end{align}
    where $\Pi_{\ge2^{-2n}}$ is the projection in \cref{lem:RSV_subspace_and_support} with $\epsilon=2^{-2n}$. Therefore, we have
    \begin{align}
        &\bigg|
         \Pr[1\gets\cD^\cU(\rho)]
         -\Pr[1\gets\cD^\cU(\rho')]
        \bigg|
        \notag\\
        \le&
        \|
         \rho
         -\rho'
        \|_1
        \notag\\
        \le&\sqrt{2^{n-p+1}+2^{-n}}
        \tag{By \cref{lem:RSV_subspace_and_support} and the gentle measurement lemma (\cref{lem:gentle_measurement})}\\
        \le&\sqrt{2^{n-p+1}+\negl(\secp)}.\label{eq:postselection_of_random_Choi_states_for_PRU}
    \end{align}

    \paragraph{Large overlap with $W_0$.}

    Next, we show $\xi$ has a large overlap with $W_0$.
    From the assumption, we have 
    \begin{align}
        \Tr[Q\xi]\le\negl(\secp),\label{eq:xi_has_negligible_overlap_with_Q}
    \end{align}
    where $Q$ is the projection onto the support of $\rho$.
    We define
    \begin{align}
        \xi'\coloneqq\frac{(I-Q)\xi(I-Q)}{\Tr[(I-Q)\xi]}.
        \label{eq:def_of_tau}
    \end{align}
    With \cref{eq:xi_has_negligible_overlap_with_Q} and the gentle measurement lemma (\cref{lem:gentle_measurement}), we have
    \begin{align}
        \bigg|
         \Pr[1\gets\cD^\cU(\xi)]-
         \Pr[1\gets\cD^\cU(\xi')]
         \bigg|
        &\le
        \|
         \xi-
         \xi'
        \|_1
        =\negl(\secp).\label{eq:first_postselection_of_Haar_random_Choi_states}
    \end{align}
    Moreover, we can show that $\xi'$ has a large overlap with $W_0$ from the following lemma.
    \ifnum\submission=0
    We give its proof later. 
    \fi
    \ifnum\submission=1
    We provide the proof in the supplemental materials.
    \fi

    \begin{restatable}{lemma}{smallRSV}\label{lem:negligible_overlap_with_large_RSVs}
        Let $\ket{\psi}$ be a state such that $Q\ket{\psi}=0$, where $Q$ is the projection onto the support of 
        $\rho$. Let $\Pi_{\ge\epsilon}$ be the projection in \cref{lem:RSV_subspace_and_support}.  
        If $V_\secp$ is a $(1,2^{-p(\secp)},\poly(\secp))$-block encoding of $\rho$, then, $\|\Pi_{\ge\epsilon}\ket{\psi}\|\le2^{-p}\epsilon^{-1}$. 
    \end{restatable}
    
    Since $\Tr[Q\xi']=0$ by its definition (\cref{eq:def_of_tau}), by applying \cref{lem:negligible_overlap_with_large_RSVs} with $\epsilon=2^{-3n}$ we have
    \begin{align}
        \Tr[\Pi_{\ge2^{-3n}}\xi']\le2^{-2p+6n}.
    \end{align}
    Thus, let us define 
    \begin{align}
        \xi''\coloneqq\frac{(I-\Pi_{\ge2^{-3n}})\xi'(I-\Pi_{\ge2^{-3n}})}{\Tr[(I-\Pi_{\ge2^{-3n}})\tau]}.
    \end{align}
    From the above inequality and the gentle measurement lemma (\cref{lem:gentle_measurement}), we have
    \begin{align}
        \bigg|
         \Pr[1\gets\cD^\cU(\xi')]-\Pr[1\gets\cD^\cU(\xi'')]
        \bigg|
        \le\|\xi'-\xi''\|_1
        \le2^{-p+3n}.\label{eq:second_postselection_of_Haar_random_Choi_states}
    \end{align}

    \paragraph{Combining All Components.}
    Now we are ready to show $\cD^\cU$ distinguishes $\rho$ from $\xi$.
    Recall that $\cD^\cU$ simulates the singular value discrimination algorithm in \cref{thm:singular_value_discrimination} with $M=(\ketbra{0^{\poly(\secp)}}{0^{\poly(\secp)}}\otimes I)V_\secp(\ketbra{0^{\poly(\secp)}}{0^{\poly(\secp)}}\otimes I)$, $\eta=2^{-\secp},a=2^{-3n}$ and $b=2^{-2n}$ by querying $\cU$. Note that $\rho'$ is on $W_1$ and $\xi''$ is on $W_0$.
    Thus, from \cref{thm:singular_value_discrimination}, we have
    \begin{align}
        \bigg|
         \Pr[1\gets\cD^\cU(\rho')]-\Pr[1\gets\cD^\cU(\xi'')]
        \bigg|\ge1-\negl(\secp).
    \end{align}
    Moreover, by choosing $p(\secp)=4n(\secp)$, we have
    \begin{align}
        \bigg|
         \Pr[1\gets\cD^\cU(\rho)]
         -\Pr[1\gets\cD^\cU(\rho')]
        \bigg|
        \le&\sqrt{2^{-3n+1}+\negl(\secp)}
        \tag{By \cref{eq:postselection_of_random_Choi_states_for_PRU}}\\
        \le&\negl(\secp)
    \end{align}
    and
    \begin{align}
        \bigg|
         \Pr[1\gets\cD^\cU(\xi)]-
         \Pr[1\gets\cD^\cU(\xi'')]
        \bigg|
        \le&
        \negl(\secp)+2^{-n}
        \tag{By \cref{eq:first_postselection_of_Haar_random_Choi_states,eq:second_postselection_of_Haar_random_Choi_states}}\\
        \le&\negl(\secp).
    \end{align}
    With these inequalities at hand, we have
    \begin{align}
        &\bigg|
         \Pr[1\gets\cD^\cU(\rho)]
         -\Pr[1\gets\cD^\cU(\xi)]
        \bigg|
        \notag\\
        \ge&1-\negl(\secp),
    \end{align}
    which concludes the proof.
\end{proof}

\ifnum\submission=0
We give proofs of \cref{lem:RSV_subspace_and_support,lem:negligible_overlap_with_large_RSVs} to complete the proof. First, we show \cref{lem:RSV_subspace_and_support}. We restate it here for the reader's convenience.

\RSVsubspace*

\begin{proof}[Proof of \cref{lem:RSV_subspace_and_support}]
    We have
    \begin{align}
        &\Tr[\rho\Pi_{\ge\epsilon}]
        \\
        =&\bigg|
         \Tr[(\bra{0^{\poly(\secp)}}\otimes I)V_\secp(\ket{0^{\poly(\secp)}}\otimes I)\Pi_{\ge\epsilon}]+
        \Tr[(\rho-(\bra{0^{\poly(\secp)}}\otimes I)V_\secp(\ket{0^{\poly(\secp)}}\otimes I))\Pi_{\ge\epsilon}]
        \bigg|
        \\
        \ge&\bigg|
         \Tr[(\bra{0^{\poly(\secp)}}\otimes I)V_\secp(\ket{0^{\poly(\secp)}}\otimes I)\Pi_{\ge\epsilon}]
        \bigg|
         -\bigg\|(\rho-(\bra{0^{\poly(\secp)}}\otimes I)V_\secp(\ket{0^{\poly(\secp)}}\otimes I))\Pi_{\ge\epsilon}
         \bigg\|_1
        \tag{By the triangle inequality and $|\Tr[A]|\le\|A\|_1$}\\
        \ge&
        \bigg|
         \Tr[(\bra{0^{\poly(\secp)}}\otimes I)V_\secp(\ket{0^{\poly(\secp)}}\otimes I)\Pi_{\ge\epsilon}]
        \bigg|
        -\bigg\|\rho-(\bra{0^{\poly(\secp)}}\otimes I)V_\secp(\ket{0^{\poly(\secp)}}\otimes I)
        \bigg\|_1.
        \tag{By H\"{o}lder's inequality (\cref{lem:Holder}) and $\|\Pi_{\ge\epsilon}\|_\infty=1$}
   \end{align}
    We can show
    \begin{align}
        \bigg|
         \Tr[(\bra{0^{\poly(\secp)}}\otimes I)V_\secp(\ket{0^{\poly(\secp)}}\otimes I)\Pi_{\ge\epsilon}]
        \bigg|
        \ge1-2^{n-p}-2^{n}\epsilon\label{eq:bound_of_trace_with_Pi}
    \end{align}
    and 
    \begin{align}
        \|\rho-(\bra{0^{\poly(\secp)}}\otimes I)V_\secp(\ket{0^{\poly(\secp)}}\otimes I)\|_1
        \le2^{n-p}.\label{eq:bound_of_trace_distance}
    \end{align}
    We give their proofs later. With these inequalities at hand, we obtain \cref{lem:RSV_subspace_and_support}.

    To conclude the proof, we give proofs of \cref{eq:bound_of_trace_with_Pi,eq:bound_of_trace_distance}. We can show the latter as follows:
    \begin{align}
        &\|
       \rho
        -(\bra{0^{\poly(\secp)}}\otimes I)V_\secp(\ket{0^{\poly(\secp)}}\otimes I)
        \|_1
        \\
        \le&2^{n}
        \|
        \rho
        -(\bra{0^{\poly(\secp)}}\otimes I)V_\secp(\ket{0^{\poly(\secp)}}\otimes I)
        \|_\infty
        \tag{By $\|A\|_1\le d\|A\|_\infty$ for any $A\in\Linear(d)$}
        \\
        \le&2^{n-p}.
        \tag{Since $V_\secp$ is a $(1,2^{-p},\poly(\secp))$-block encoding of $\rho$}
    \end{align}
    We can show the former as follows.
    \begin{align}
        &\bigg|
         \Tr[(\bra{0^{\poly(\secp)}}\otimes I)V_\secp(\ket{0^{\poly(\secp)}}\otimes I))\Pi_{\ge\epsilon}]
        \bigg|
        \\
        =&
        \bigg|
         \Tr[\bra{0^{\poly(\secp)}}\otimes I)V_\secp(\ket{0^{\poly(\secp)}}\otimes I)]
        -
         \Tr[(\bra{0^{\poly(\secp)}}\otimes I)V_\secp(\ket{0^{\poly(\secp)}}\otimes I)(I-\Pi_{\ge\epsilon})]
        \bigg|
        \\
        \ge&
        \bigg|
         \Tr[\bra{0^{\poly(\secp)}}\otimes I)V_\secp(\ket{0^{\poly(\secp)}}\otimes I)]
        \bigg|
        -
        \bigg\|
         (\bra{0^{\poly(\secp)}}\otimes I)V_\secp(\ket{0^{\poly(\secp)}}\otimes I)(I-\Pi_{\ge\epsilon})
        \bigg\|_1,\label{eq:ineq1_for_bound_of_trace_with_Pi}
    \end{align}
    where we have used the triangle inequality and $|\Tr[A]|\le\|A\|_1=\|-A\|_1$ in the inequality. The first term can be estimated as follows:
    \begin{align}
        &\bigg|
         \Tr[\bra{0^{\poly(\secp)}}\otimes I)V_\secp(\ket{0^{\poly(\secp)}}\otimes I)]
        \bigg|
        \notag\\
        =&
        \bigg|
         \Tr[\rho]
         -
         \Tr[\rho-
         \bra{0^{\poly(\secp)}}\otimes I)V_\secp(\ket{0^{\poly(\secp)}}\otimes I)]
        \bigg|
        \notag\\
        \ge&
        |\Tr[\rho]|-
        \bigg\|
         \rho-
         \bra{0^{\poly(\secp)}}\otimes I)V_\secp(\ket{0^{\poly(\secp)}}\otimes I)
        \bigg\|_1
        \tag{By the triangle inequality and $|\Tr[A]|\le\|A\|_1$}
        \\
        \ge&1-2^{n-p},\label{eq:ineq2_for_bound_of_trace_with_Pi}
    \end{align}
    where we have used \cref{eq:bound_of_trace_distance} in the last inequality. To estimate the second term in \cref{eq:ineq1_for_bound_of_trace_with_Pi}, recall that $\Pi_{\ge\epsilon}$ is the projection onto the subspace spanned by right singular vectors of $\bra{0^{\poly(\secp)}}\otimes I)V_\secp(\ket{0^{\poly(\secp)}}\otimes I)$ whose singular values are at least $\epsilon$. 
    Thus, $I-\Pi_{\ge\epsilon}$ is the projection onto the subspace spanned by right singular vectors of $\bra{0^{\poly(\secp)}}\otimes I)V_\secp(\ket{0^{\poly(\secp)}}\otimes I)$ whose singular values are less than $\epsilon$.
    In addition to that, note that the number of singular values of $\bra{0^{\poly(\secp)}}\otimes I)V_\secp(\ket{0^{\poly(\secp)}}\otimes I)$ is at most $2^{n}$ since $\bra{0^{\poly(\secp)}}\otimes I)V_\secp(\ket{0^{\poly(\secp)}}\otimes I)$ is an operator acting on $n$ qubits.
    With these observations, we have
    \begin{align}
        \bigg\|
         (\bra{0^{\poly(\secp)}}\otimes I)V_\secp(\ket{0^{\poly(\secp)}}\otimes I)(I-\Pi_{\ge\epsilon})
        \bigg\|_1\le2^{n}\epsilon\label{eq:ineq3_for_bound_of_trace_with_Pi}
    \end{align}
    since $\|A\|_1$ is equivalent to the sum of its singular values.
    From \cref{eq:ineq1_for_bound_of_trace_with_Pi,eq:ineq2_for_bound_of_trace_with_Pi,eq:ineq3_for_bound_of_trace_with_Pi}, we obtain \cref{eq:bound_of_trace_with_Pi}.
\end{proof}

Finally, we show \cref{lem:negligible_overlap_with_large_RSVs}. We also restate it here.

\smallRSV*

\begin{proof}[Proof of \cref{lem:negligible_overlap_with_large_RSVs}]
    Since $Q$ is the projection onto $\rho$, we have $\rho\ket{\psi}=\rho Q\ket{\psi}=0$. Then, we have
    \begin{align}
        \|
         \bra{0^{\poly(\secp)}}\otimes I)V_\secp(\ket{0^{\poly(\secp)}}\otimes I)\ket{\psi}
        \|
        =&\big\|
         \big(\rho-
         \bra{0^{\poly(\secp)}}\otimes I)V_\secp(\ket{0^{\poly(\secp)}}\otimes I)\big)\ket{\psi}
        \big\|
        \\
        \le&
        \|\rho-
         \bra{0^{\poly(\secp)}}\otimes I)V_\secp(\ket{0^{\poly(\secp)}}\otimes I)
         \|_\infty
         \\
         \le&2^{-p},\label{eq:ineq1_for_negligible_overlap_with_large_RSVs}
    \end{align}
    where the last inequality follows from the assumption that $V_\secp$ is a $(1,2^{-p},\poly(\secp))$-blcok encoding of $\rho$.
    Let 
    \begin{align}
        (\bra{0^{\poly(\secp)}}\otimes I)V_\secp(\ket{0^{\poly(\secp)}}\otimes I)=\sum_i a_i\ketbra{w_i}{v_i}
    \end{align}
    be the singular value decomposition. Namely, $\{\ket{w_i}\}_i$ and $\{\ket{v_i}\}_i$ are sets of orthonormal states, and all $a_i$ are positive real numbers.
    Since each $\ket{v_i}$ is the right singular vector of $\bra{0^{\poly(\secp)}}\otimes I)V_\secp(\ket{0^{\poly(\secp)}}\otimes I)$ whose singular value is $a_i$, we have $\Pi_{\ge\epsilon}=\sum_{i:a_i\ge\epsilon}\ketbra{v_i}{v_i}$.
    Then, we have 
    \begin{align}
        \|
         \bra{0^{\poly(\secp)}}\otimes I)V_\secp(\ket{0^{\poly(\secp)}}\otimes I)\ket{\psi}
        \|
        =&
        \bigg\|
         \sum_i a_i\ket{w_i}\braket{v_i|\psi}
        \bigg\|
        \\
        =&
        \sqrt{
        \sum_{i}a_i^2|\braket{v_i|\psi}|^2
        }
        \\
        \ge&
        \sqrt{
        \sum_{i:a_i\ge\epsilon}a_i^2|\braket{v_i|\psi}|^2
        }
        \\
        \ge&\epsilon\sqrt{
        \sum_{i:a_i\ge\epsilon}|\braket{v_i|\psi}|^2
        }
        \\
        =&\epsilon\|\Pi_{\ge\epsilon}\ket{\psi}\|,\label{eq:ineq2_for_negligible_overlap_with_large_RSVs}
    \end{align}
    where in the last inequality we have used $\Pi_{\ge\epsilon}=\sum_{i:a_i\ge\epsilon}\ketbra{v_i}{v_i}$.
    From \cref{eq:ineq1_for_negligible_overlap_with_large_RSVs,eq:ineq2_for_negligible_overlap_with_large_RSVs},
    we have $\epsilon\|\Pi_{\ge\epsilon}\ket{\psi}\|\le2^{-p}$, which implies 
    $\|\Pi_{\ge\epsilon}\ket{\psi}\|\le2^{-p}\epsilon^{-1}$.
\end{proof}
\fi

\section{Oracle Separation Between PRIs with Short Stretch and PRFSGs}
\label{sec:PRFSG_vs_PRI}

\ifnum\submission=0
In this section, we prove the following.
\fi
\ifnum\submission=1
By modifying the proof in \cref{sec:PRUs}, we can prove the following.
We give it in the supplemental material.
\fi

\begin{theorem}\label{thm:PRI_vs_PRFSG}
    Then, with probability $1$ over the choice of $\cO$ defined in \cref{def:unitary_orcle}, the following are satisfied:
    \begin{itemize}
        \item Quantumly-accessible adaptively-secure PRFSGs exist relative to $\cO$.
        \item Non-adaptive, $O(\log\secp)$-ancilla PRIs with $O(\log\secp)$ stretch do not exist relative to $\cO$.
    \end{itemize}
\end{theorem}

\begin{remark}
    \cref{thm:PRI_vs_PRFSG} is stronger than \cref{thm:main} because, for any $s$ and $c$, non-adaptive, $c$-ancilla PRIs with $s$ stretch are constructed from non-adaptive, $c$-ancilla PRUs in a black-box manner.
\end{remark}

\ifnum\submission=0
\ifnum\submission=1
\section{Proof of \cref{thm:PRI_vs_PRFSG}}
In this section, we prove \cref{thm:PRI_vs_PRFSG}.
\fi
Since we have already proved the first item in \cref{thm:PRFSGs_relative_to_unitary_oracle}, it suffices to prove the second item by constructing an adversary breaking PRIs.
We give the adversary in a similar way as in \cref{alg:break_PRU}.
Let $G^\cO$ be a QPT algorithm that satisfies the correctness of $c$-ancilla PRIs with $s$ stretch.
In particular, we consider the case when $c(\secp)=O(\log\secp)$ and $s(\secp)=O(\log\secp)$.
For such $G^\cO$,
let $\{\cI_k\}_{k\in\cK_\secp}$ be the isometry implemented by $G^\cO$ on input $k\in\cK_\secp$, where $\cK_\secp$ denotes the key-space.
We define the following map:
\begin{align}
    \cM_{\{\cI_k\},\ell}(\cdot)&=\Exp_{k\gets\cK_\secp}\cI^{\otimes\ell}_k(\cdot)\cI_k^{\dag\otimes\ell}.
\end{align}

To construct a PRI adversary, we need two lemmas as in \cref{sec:PRUs}.
These lemmas can be obtained by modifying \cref{lem:approx_random_Choi_state} and \cref{lem:distinguisher} for PRIs. We give their proofs later.

\begin{restatable}{lemma}{PRIChoiStates}\label{lem:approx_PRI_Choi_state}
    Let $T(\secp)$ be a polynomial.
    Let $\epsilon>0$ and $c,d,s\in\N$.
    Let $\{\cI_k\}_{k\in\cK_\secp}$ be an ensemble of isometries mapping $\secp$ qubits to $\secp+s$ qubits, where each $\cI_k$ is QPT implementable with $s+c$ ancilla qubits by making $T$ queries to $\cO$ defined in \cref{def:unitary_orcle}. 
    For all $n\in[d]$, let $\cS'_{n}$ be any unitary satisfying 
    \begin{align}
        \|\cS_{n}(\cdot)\cS_{n}^\dag-\cS_{n}'(\cdot)\cS_{n}^{'\dag}\|_\diamond\le\epsilon.
    \end{align}
    Then, for any polynomial $\ell$, there exists a family 
    $\{V_k\}_{k\in\cK_\secp}$ 
    of $(\secp+s+c)$-qubit unitaries 
    such that each $V_k$ is QPT implementable with classical descriptions of $\cS'_n$ for all $n\in[d]$ and query access to the $\unitaryPSPACE$-complete oracle $\cU$ such that it satisfies
    \begin{align}
        \bigg\|
         (\cM_{\{\cI_k\},\ell}\otimes \identitymap)(\ketbra{\Omega_{2^{\secp\ell}}}{\Omega_{2^{\secp\ell}}})-
         (\cF_{\{V_k\},\ell}\otimes \identitymap)(\ketbra{\Omega_{2^{\secp\ell}}}{\Omega_{2^{\secp\ell}}})
        \bigg\|_1\le O(\ell T\epsilon)
        +O\bigg(\frac{2^{s+3c/2}\ell T}{2^{d/2}}\bigg).
        \label{eq:goal;approx_random_Choi_state}
    \end{align}
    Here, $\cF_{\{V_k\},\ell}$ is a CPTP map from $\secp\ell$ qubits to $(\secp+s)\ell$ qubits defined as follows:
    \begin{align}
        \cF_{\{V_k\},\ell}((\cdot)_\regA)\coloneqq
        \Tr_{\regC}
        \bigg[
         \Exp_{k\gets\cK_\secp}V^{\otimes\ell}_{k,\regA\regB\regC}((\cdot)_\regA\otimes\ketbra{0^{s}}{0^s}^{\otimes\ell}_{\regB}\otimes\ketbra{0^{c}}{0^c}^{\otimes\ell}_{\regC})
         V^{\dag\otimes\ell}_{k,\regA\regB\regC}
        \bigg],
    \end{align}
    where $\regA$ is a $\secp\ell$-qubit register, $\regB$ is a $s\ell$-qubit register, and $\regC$ is a $c\ell$-qubit register.
\end{restatable}

\ifnum\submission=0
For the next lemma, we define the Haar random isometry map as follows:

\begin{definition}[Haar Random Isometry Map]\label{def:Haar_isometry_map}
    We define\footnote{Here, we chose $\ket{0^s}$ as an input state. This does not lose any generality due to the right invariance of the Haar measure.}
    \begin{align}
     \cI_{\secp\to \secp+s,\ell}(\rho_\regA)\coloneqq
     \Exp_{U\gets\mu_{2^{\secp+s}}} 
     \bigg(
      \bigotimes_{i\in[\ell]}U_{\regA_i\regB_i}
     \bigg)
     (\rho_\regA\otimes\ketbra{0^s}{0^s}^{\otimes\ell}_{\regB})
     \bigg(
      \bigotimes_{i\in[\ell]}U_{\regA_i\regB_i}
     \bigg)^\dag,
    \end{align}
    where $\regA\coloneqq\bigotimes_{i\in[\ell]}\regA_i$ and $\regB\coloneqq\bigotimes_{i\in[\ell]}\regB_i$, where, for each $i\in[\ell]$, $\regA_i$ and $\regB_i$ are $\secp$-qubit register and $s$-qubit register, respectively.
\end{definition}

Now we are ready to give the following lemma.
\fi
\ifnum\submission=1
Regarding the next lemma, recall that $\cI_{\secp\to\secp+s}$ is the Haar random isometry map defined in \cref{def:Haar_isometry_map}.
\fi

\begin{restatable}{lemma}{PRIdistinguisher}\label{lem:PRI_distinguisher}
    Suppose that $c(\secp)=O(\log\secp)$.
    Let $\{V_k\}_{k\in\cK_\secp}$ be the family of $(\secp+s+c)$-qubit unitary, and $\cF_{\{V_k\}_k,\ell}$ be the CPTP map in \cref{lem:approx_PRI_Choi_state}.
    Let $\ell(\secp)\coloneqq\lceil\log|\cK_\secp|\rceil$.
    Then, there exists a QPT algorithm $\cD^{\cU}$ that, on input classical descriptions of $\cS'_n$ for all $n\in[d]$, distinguishes $(\cF_{\{V_k\},\ell}\otimes\identitymap)(\ketbra{\Omega_{2^{\secp\ell}}}{\Omega_{2^{\secp\ell}}})$ from $(\cI_{\secp\to\secp+s,\ell}\otimes\identitymap)(\ketbra{\Omega_{2^{\secp\ell}}}{\Omega_{2^{\secp\ell}}})$
    with advantage at least $1-\negl(\secp)$.
\end{restatable}

Based on these lemmas, we construct a PRI adversary as shown in \cref{alg:break_PRI}.
The {\color{red}red-highlighted lines} in \cref{alg:break_PRI} show the differences from \cref{alg:break_PRU}.

\begin{algorithm}
    \caption{Adversary distinguishing $\{\cI_k\}_{k\in\cK_\secp}$ from Haar random isometry relative to $(\cS,\cU)$.}
    \label{alg:break_PRI} 
    \vspace{2mm}
    \AlgOracle 
    The algorithm has query access to
    \begin{itemize}
        \item {\color{red}an isometry $\cI$ mapping $\secp$ qubits to $\secp+s(\secp)$ qubits, which is whether $\cI_k$ or a Haar random isometry.}
        
        \item the oracle $\cO=(\cS,\cU)$ and its inverse defined in \cref{def:unitary_orcle};
        
    \end{itemize}
    \AlgInput The algorithm takes the security parameter $1^\secp$ as input.

    Define $\ell\coloneqq\lceil\log|\cK_\secp|\rceil$ and {\color{red}$d\coloneqq 2\log(\ell Tp)+3c+2s$}, where $p$ is a polynomial, and $T$ is the number of queries to $\cO$ to implement {\color{red}$\cI_k$}.

    \vspace{2mm}
    \begin{enumerate}
        \item 
        \label{algstep:process_tomography_PRI}
        For $n\in[d]$, run the process tomography algorithm in \cref{thm:process_tomography_HKOT23} on inputs $\epsilon\coloneqq\frac{1}{\ell Tp}$ and $\eta\coloneqq2^{-\secp-1}$ for $\cS_{n}$ to get a classical description of $\cS'_{n}$. Note that $\|\cS_{n}(\cdot)\cS_{n}^\dag-\cS'_{n}(\cdot)\cS'^\dag_{n}\|_\diamond\le\epsilon$ holds with probability at least $1-2^{-\secp}$ over the randomness of the process tomography algorithm.
        
        \item {\color{red}Prepare $(\cI^{\otimes\ell}\otimes I)\ket{\Omega_{2^{\secp\ell}}}$ by querying $\cI$.}

        \item 
        {\color{red}Let $\{V_k\}_{k\in\cK_\secp}$ be a family of unitaries in \cref{lem:approx_PRI_Choi_state}.} Note that each $V_k$ is QPT implementable with access to $\cU$ and classical descriptions of $S'_{n}$ for all $n\in[d]$, where such classical descriptions are obtained in the step \ref{algstep:process_tomography_PRI}.
        Let $\cD^{(\cdot)}$ be a QPT algorithm in {\color{red}\cref{lem:PRI_distinguisher} for $\{V_k\}_{k\in\cK_\secp}$}. 
        By querying $\cU$, run $\cD^\cU$ on input {\color{red}$(\cI^{\otimes\ell}\otimes I)\ket{\Omega_{2^{\secp\ell}}}$} and classical descriptions of $S'_{n}$ for all $n\in[d]$ to get $b\in\bit$. 
    \end{enumerate} 

    \AlgOutput The algorithm outputs $b$.
\end{algorithm}

The following \cref{thm:break_PRI} implies \cref{thm:PRI_vs_PRFSG}.
We omit the proof of \cref{thm:break_PRI} as it is identical to that of \cref{thm:break_PRU}, except that \cref{lem:approx_random_Choi_state,lem:distinguisher} are respectively replaced by \cref{lem:approx_PRI_Choi_state,lem:PRI_distinguisher}.

\begin{theorem}\label{thm:break_PRI}
Suppose that $c(\secp)=O(\log\secp)$ and $s(\secp)=O(\log\secp)$.
Let $\cO$ be a fixed oracle defined in \cref{def:unitary_orcle}.
Let $G^\cO$ be a QPT algorithm that satisfies the correctness of $c$-ancilla PRIs with $s$ stretch.
For such $G^\cO$,
let $\{\cI_k\}_{k\in\cK_\secp}$ be the isometry mapping $\secp$ qubits to $\secp+s$ qubits implemented by $G^\cO$ on input $k\in\cK_\secp$, where $\cK_\secp$ denotes the key-space.
Then, for any polynomial $p$, the QPT adversary $\cA^{(\cdot,\cdot)}$ defined in \cref{alg:break_PRI} satisfies
    \begin{align}
        \bigg|\Pr_{k\gets\cK_\secp}[1\gets\cA^{\cI_k,\cO}(1^\secp)]-\Pr_{U\gets\mu_{2^{\secp+s}}}[1\gets\cA^{\cI_U,\cO}(1^\secp)]\bigg|
        \ge 1-O\bigg(\frac{1}{p(\secp)}\bigg),
    \end{align}
    where, for each $U\in\Unitaries(2^{\secp+s(\secp)})$, $\cI_U$ is the isometry that maps $\secp$-qubit state $\ket{\psi}$ to $(\secp+s(\secp))$-qubit state $U(\ket{\psi}\ket{0^s})$. 
    Here, $\cA^{(\cdot),\cO}$ queries not only $\cO$ but also its inverse. Moreover, $\cA^{(\cdot),\cO}$ queries the first oracle non-adaptively.
\end{theorem}

\begin{remark}
    The parameter $d$ is chosen so that $\ell T2^{s+3c/2 - d/2} = O(1/p)$, representing the error term in \cref{lem:approx_PRI_Choi_state}.
    \cref{alg:break_PRI} is QPT because $2^d = (\ell Tp)^2 \cdot 2^{2s+3c} \le \poly(\secp)$, where the inequality holds under the assumption that $c(\secp)=O(\log\secp)$ and $s(\secp) = O(\log \secp)$.
    In contrast, if $c(\secp)=\omega(\log\secp)$ or $s(\secp) = \omega(\log \secp)$, then \cref{alg:break_PRI} would no longer be QPT.
\end{remark}

\subsection[Proof of Lemma~\ref{lem:approx_PRI_Choi_state}]{Proof of \cref{lem:approx_PRI_Choi_state}}

In this subsection, we prove \cref{lem:approx_PRI_Choi_state}. 
The proof strategy is the same as that of \cref{lem:approx_random_Choi_state}.

\begin{proof}[Proof of \cref{lem:approx_PRI_Choi_state}]
     Let $G^\cO(k,\cdot)$ be an algorithm implementing $\cI_k$.
     For each $k \in \cK_\secp$, we may view $G^\cO(k,\cdot)$ as acting as follows: first, it prepares $\ket{0^s}\ket{0^c}$ in the ancilla register, then applies a $(\secp +s+ c)$-qubit unitary $W_k$ to the input and ancilla qubits by querying $\cO$, and finally discards the last $c$ qubits in the ancilla register.
    
    We define a $(\secp+s+c)$-qubit unitary $V_k$ as follows. 
    Take the same circuit as in the implementation of $W_k$, except at the points where it queries $\cS_{n}$. 
    Whenever the circuit queries $\cS_{n}$,
          \begin{itemize}
              \item if $n\in[d]$, apply $\cS'_{n}$ by using its classical description;
              \item if $n\in[\secp+s+c]/[d]$, do not apply any unitary circuit. 
          \end{itemize}
    We prove \cref{eq:goal;approx_random_Choi_state} by the standard hybrid argument.
    For each $k\in\cK_\secp$, we define a $(\secp+s+c)$-qubit unitary $\widetilde{V}_k$ as follows: apply the same unitary as $W_k$ except for querying $\cS_{n}$. When querying to $\cS_{n}$,
          \begin{itemize}
              \item if $n\in[d]$, apply $\cS_{n}$;
              \item if $n\in[\secp+s+c]/[d]$, do not apply any unitary. 
          \end{itemize}
    Define a CPTP map $\cF_{\{\widetilde{V_k}\}_k,\ell}$ from $\secp\ell$ qubits to $(\secp+s)\ell$ qubits in the same manner as $\cF_{\{V_k\}_k,\ell}$.
    We can show
    \begin{align}
        \bigg\|
         (\cF_{\{V_k\},\ell}\otimes \identitymap)(\ketbra{\Omega_{2^{\secp\ell}}}{\Omega_{2^{\secp\ell}}})
         -(\cF_{\{\widetilde{V}_k\},\ell}\otimes \identitymap)(\ketbra{\Omega_{2^{\secp\ell}}}{\Omega_{2^{\secp\ell}}})
         \bigg\|_1\le O(\ell T\epsilon)
        \label{eq:hyb0;approx_PRI_Choi_state}
    \end{align}
    and
    \begin{align}
        \bigg\|
        (\cF_{\{\widetilde{V}_k\},\ell}\otimes \identitymap)(\ketbra{\Omega_{2^{\secp\ell}}}{\Omega_{2^{\secp\ell}}})
        -(\cM_{\{\cI_k\},\ell}\otimes \identitymap)(\ketbra{\Omega_{2^{\secp\ell}}}{\Omega_{2^{\secp\ell}}})
        \bigg\|_1\le O(\ell T2^{s+3c/2-d/2}).
        \label{eq:hyb1;approx_PRI_Choi_state}
    \end{align}
    We will give the proofs later. From \cref{eq:hyb0;approx_PRI_Choi_state,eq:hyb1;approx_PRI_Choi_state} and the triangle inequality, we obtain \cref{eq:goal;approx_random_Choi_state}.

    First, we prove \cref{eq:hyb0;approx_PRI_Choi_state}. For each $k\in\cK_\secp$, the difference between $V_k$ and $\widetilde{V}_k$ lies only whether we apply $\cS'_{n}$ or $\cS_{n}$ for all $n\in[d]$. Since the number of queries to the swap unitaries is at most $T$, we have
    \begin{align}
        \|V_k(\cdot)V_k^{\dag}-\widetilde{V}_k(\cdot)\widetilde{V}_k^\dag\|_\diamond\le T\epsilon
    \end{align}
    for all $k\in\cK_\secp$, which implies \cref{eq:hyb0;approx_PRI_Choi_state}.

    For each $k$, define a map $\cF_k$ as follows:
    \begin{align}
        \cF_k(\rho_{\regX})\coloneqq\Tr_{\regZ}[V_{k,\regX\regY\regZ}(\rho_\regX\otimes\ketbra{0^s}{0^s}_\regY\otimes\ketbra{0^c}{0^c}_\regZ)V_{k,\regX\regY\regZ}^\dag],
    \end{align}
    where $\regX$ is a $\secp$-qubit register, $\regY$ is a $s$-qubit register, and $\regZ$ is a $c$-qubit register.
    Suppose that we have
    \begin{align}
        \bigg\|
         (\cF_{k,\regX\to\regX\regY}\otimes \identitymap_{\regX'})(\ketbra{\Omega_{2^{\secp}}}{\Omega_{2^{\secp}}}_{\regX\regX'})
         -(\cI_{k,\regX\to\regX\regY}\otimes \identitymap_{\regX'})(\ketbra{\Omega_{2^{\secp}}}{\Omega_{2^{\secp}}}_{\regX\regX'})
         \bigg\|_1
        \le O(T2^{s+3c/2-d/2})\label{eq:hyb3;approx_PRI_Choi_state}
    \end{align}
    for all $k\in\cK_\secp$, where $\regX'$ is a $\secp$-qubit regoister.
    Since \cref{eq:hyb3;approx_PRI_Choi_state} implies \cref{eq:hyb1;approx_PRI_Choi_state}, it suffices to prove \cref{eq:hyb3;approx_PRI_Choi_state}.
    From the $c$-ancilla correctness condition, we can view that $\cI_k$ is implemented as follows:\footnote{Here, $A$ and each $V_i$ depend on $k$, but we omit the subscript of $k$ for notational simplicity.}
    \begin{enumerate}
        \item Prepare $\ket{0^{s}}_\regY\ket{0^c}_\regZ$ on an ancilla qubits, and apply $(\secp+s+c)$-qubit unitary on $\regX\regY\regZ$. This is equal to applying an isometry $A$ mapping $\regX$ to $\regX\regY\regZ$.
        \item For each $i\in[T]$, perform the following:
        Apply $\cT_i\coloneqq\cS_{n_i}$ onto some $2n_i+1$ qubits of $\regX\regY\regZ$, where $2n_i+1\le\secp+s+c$.
        Then, apply a $(\secp+s+c)$-qubit unitary $B_i$ on $\regX\regY\regZ$.
        \item Discard $\regZ$.
    \end{enumerate}
    From the above observation, we define the following hybrids for each $i\in[T]$:\footnote{Here $\prod_jB_j$ means $B_T\cdots B_1$. The order is important because each operation is not commutative in general.}

    \begin{align}
        \ket{\psi_i}_{\regX\regY\regZ\regX'} = \left(\prod_{j=i+1}^{T} \left(B_j\cT'_{j}\right)_{\regX\regY\regZ}\prod_{j=1}^{i} \left(B_j\cT_{j}\right)_{\regX\regY\regZ}A_{\regX\to\regX\regY\regZ} \otimes I_{\regX'}\right)\ket{\Omega_{2^\secp}}_{\regX\regX'},
    \end{align}
    where $\cT'_j$ is defined as follows:
    \begin{itemize}
        \item if $n_i\in[d]$, it is exactly the same as $\cT_i$.
        \item if $n_i\in[\secp+s+c]/[d]$, it is the identity.
    \end{itemize}
    Then, we can prove the following for all $i\in[T]$:
    \begin{align}
        \|\ketbra{\psi_{i-1}}{\psi_{i-1}}-\ketbra{\psi_i}{\psi_i}\|_1\le O(2^{s+3c/2-d/2})
        \label{eq:hyb4;approx_PRI_Choi_state}
    \end{align}
    Since $\Tr_\regZ[\ketbra{\psi_0}{\psi_0}_{\regX\regY\regZ\regX'}]=(\cF_{k,\regX\to\regX\regY}\otimes \identitymap_{\regX'})(\ketbra{\Omega_{2^{\secp}}}{\Omega_{2^{\secp}}}_{\regX\regX'})$, and $\Tr_\regZ[\ketbra{\psi_T}{\psi_T}_{\regX\regY\regZ\regX'}]=(\cI_{k,\regX\to\regX\regY}\otimes \identitymap_{\regX'})(\ketbra{\Omega_{2^{\secp}}}{\Omega_{2^{\secp}}}_{\regX\regX'})$, we obtain \cref{eq:hyb3;approx_PRI_Choi_state} from \cref{eq:hyb4;approx_PRI_Choi_state}.
    Thus, it remains to prove \cref{eq:hyb4;approx_PRI_Choi_state}.
    To this end, we need the following claim.
    
    \begin{claim}\label{claim:choi_state_on_isometry}
        Let $A$ be an isometry which maps $\ell_{\text{in}}$ qubits to $\ell_{\text{out}}$ qubits.
        Then,
        \begin{align}
            (A_{\regE}\otimes I_{\regF})\ket{\Omega_{2^{\ell_{\text{in}}}}}_{\regE\regF} = \sqrt{2^{\ell_{\text{out}}-\ell_{\text{in}}}}(I_{\regE'}\otimes A^{\top}_{\regF'})\ket{\Omega_{2^{\ell_{\text{out}}}}}_{\regE'\regF'},
        \end{align}
        where $\top$ denotes the transpose.
        Here, $\regE$ and $\regF$ are $\ell_{\text{in}}$-qubit registers, and $\regE'$ and $\regF'$ are $\ell_{\text{out}}$-qubit registers.
    \end{claim}
    
    \cref{claim:choi_state_on_isometry} follows from a straightforward calculation.
    We give the proof later.
    We obtain \cref{eq:hyb4;approx_PRI_Choi_state} as follows.
    \begin{align}
        &\|\ketbra{\psi_{i-1}}{\psi_{i-1}}- \ketbra{\psi_{i}}{\psi_{i}} \|_{1}  
        \notag\\
        =& 
        \left\|
         \left(
          \left(
           \prod_{j=i}^{T} B_j\cT'_j
           \cdot
           \prod_{j=1}^{i-1} B_j\cT_j
           \cdot A
          \right)
         \otimes \identitymap
         \right)
         (\ketbra{\Omega_{2^\secp}}{\Omega_{2^\secp}}) 
         - 
         \left(
          \left(
           \prod_{j=i+1}^{T} B_j\cT'_j
           \cdot
           \prod_{j=1}^{i} B_j\cT_j
           \cdot A
          \right)
          \otimes \identitymap
         \right)
        (\ketbra{\Omega_{2^\secp}}{\Omega_{2^\secp}})
        \right\|_{1} 
        \notag\\
        =& 
        \left\|
         \left(
          \left(
           \prod_{j=i+1}^{T} B_j\cT_{j}
           \cdot V_{i}(\cT'_{i}-\cT_{i})
           \cdot
           \prod_{j=1}^{i-1} B_j\cT_{j}
           \cdot A
          \right)
           \otimes \identitymap
         \right)
         (\ketbra{\Omega_{2^\secp}}{\Omega_{2^\secp}})
        \right\|_{1} 
        \notag\\
        =& 2^{s+c}
        \left\|
         \left(
          \left(
           \prod_{j=i+1}^{T} B_j\cT_{j}
           \cdot V_{i}(\cT'_{i}-\cT_{i})
          \right)
          \otimes 
          \left(
           \prod_{j=1}^{t-i-1} B_j\cT_{j}
           \cdot A
          \right)^\top
         \right)
         (\ketbra{\Omega_{2^{\secp+s+c}}}{\Omega_{2^{\secp+s+c}}})
         \right\|_{1} 
        \tag{By \cref{claim:choi_state_on_isometry} with $\ell_{\text{in}}=\secp$ and $\ell_{\text{out}}=\secp+s+c$}\\
        \leq& 2^{s+c}
        \left\|
         \left(
          (\cT'_{i}-\cT_{i}) \otimes \identitymap
         \right)
         (\ketbra{\Omega_{2^{\secp+s+c}}}{\Omega_{2^{\secp+s+c}}})
        \right\|_{1} 
        \tag{By H\"{o}lder's inequality \cref{lem:Holder} with $\|A\|_\infty=\|A^\top\|_\infty\le1$ for any isometry $A$}\\
        \leq& O(2^{s+c3/2-d/2}),
    \end{align}
    where, in the last line, we have used \cref{lem:swap_unitary_is_almost_identity_for_Choi_state;ancilla} with $c'=s+c$.
    which concludes the proof.
\end{proof}

Finally, we give the proof of \cref{claim:choi_state_on_isometry}.

\begin{proof}[Proo of \cref{claim:choi_state_on_isometry}]
    Let $$A = \sum_{\substack{x\in\set{0,1}^{\ell_{\text{in}}}\\ y\in\set{0,1}^{\ell_{\text{out}}}}} \alpha_{x,y}\ketbra{y}{x},$$
    then 
    \begin{align}
        (A_{\regE}\otimes I_{\regF})\ket{\Omega_{\ell_{\text{in}}}}_{\regE\regF} &= \frac{1}{\sqrt{2^{\ell_{\text{in}}}}}(A_{\regE}\otimes I_{\regF})\sum_{\substack{x\in\set{0,1}^{\ell_{\text{in}}}}} \ket{x}_{\regE}\ket{x}_{\regF}\\
        &= \frac{1}{\sqrt{2^{\ell_{\text{in}}}}}\sum_{\substack{x\in\set{0,1}^{\ell_{\text{in}}}\\ y\in\set{0,1}^{\ell_{\text{out}}}}} \alpha_{x,y}\ket{y}_{\regE}\ket{x}_{\regF},
    \end{align}
    and 
    \begin{align}
        (I_{\regE'}\otimes A^{\top}_{\regF'})\ket{\Omega_{\ell_{\text{out}}}}_{\regE'\regF'} &= \frac{1}{\sqrt{2^{\ell_{\text{out}}}}}(I_{\regE'}\otimes A^{\top}_{\regF'})\sum_{\substack{y\in\set{0,1}^{\ell_{\text{out}}}}} \ket{y}_{\regE'}\ket{y}_{\regF'}\\
        &=\frac{1}{\sqrt{2^{\ell_{\text{out}}}}}\sum_{\substack{x\in\set{0,1}^{\ell_{\text{in}}}\\ y\in\set{0,1}^{\ell_{\text{out}}}}} \alpha_{x,y}\ket{y}_{\regE'}\ket{x}_{\regF'}.
    \end{align}
    Hence, we have 
    $$(A_{\regE}\otimes I_{\regF})\ket{\Omega_{\ell_{\text{in}}}}_{\regE\regF} = \sqrt{2^{\ell_{\text{out}}-\ell_{\text{in}}}}(I_{\regE'}\otimes A^{\top}_{\regF'})\ket{\Omega_{\ell_{\text{out}}}}_{\regE'\regF'}.$$
\end{proof}

\subsection[Proof of Lemma~\ref{lem:PRI_distinguisher}]{Proof of \cref{lem:PRI_distinguisher}}

In this subsection, we prove \cref{lem:PRI_distinguisher}.
To this end, it suffices to apply \cref{lem:distinguisher} with $\rho=(\cF_{\{V_k\},\ell}\otimes\identitymap)(\ketbra{\Omega_{2^{\secp\ell}}}{\Omega_{2^{\secp\ell}}})$ and $\xi=(\cI_{\secp\to\secp+s,\ell}\otimes\identitymap)(\ketbra{\Omega_{2^{\secp\ell}}}{\Omega_{2^{\secp\ell}}})$.
Before that, we need to certify that they satisfy the condition in \cref{lem:distinguisher}.
First, we need the following, which corresponds to \cref{lem:block-encoding_of_random_Choi_states}.

\begin{lemma}\label{lem:block-encoding_of_PRI_Choi_states}
    Let $\{V_k\}_{k\in\cK_\secp}$ be a family of $(\secp+s+c)$-qubit unitaries that are QPT implementable with classical descriptions of $\cS'_n$ for all $n\in[d]$ and 
    with the query access to $\cU$, where $\cU$ is the $\unitaryPSPACE$ complete problem in \cref{lem:unitaryPSPACE_has_a_complete_problem}.
    Let $\ell(\secp)\coloneqq\lceil\log|\cK_\secp|\rceil$. Then, for any polynomial $p$, there exists a unitary circuit $V_\secp$ satisfying the following:
    \begin{itemize}
        \item $V_\secp$ is QPT implementable with classical descriptions of $\cS'_n$ for all $n\in[d]$ and with the query access to $\cU$.
        \item $V_\secp$ is a $(1,2^{-p(\secp)},\poly(\secp))$-block encoding of $(\cF_{\{V_k\},\ell}\otimes\identitymap)(\ketbra{\Omega_{2^{\secp\ell}}}{\Omega_{2^{\secp\ell}}})$.
    \end{itemize}
\end{lemma}

Since the proof is the same as that of \cref{lem:block-encoding_of_random_Choi_states}, we omit it.
Next, we show that $(\cI_{\secp\to\secp+s,\ell}\otimes\identitymap)(\ketbra{\Omega_{2^{\secp\ell}}}{\Omega_{2^{\secp\ell}}})$ has negligible overlap with the support of $(\cF_{\{V_k\},\ell}\otimes\identitymap)(\ketbra{\Omega_{2^{\secp\ell}}}{\Omega_{2^{\secp\ell}}})$ formalized as follows:

\begin{lemma}\label{lem:Haar_isometry_Choi_has_negligible_overlap}
    Suppose that $c(\secp)=O(\log\secp)$.
    Let $Q'$ be the projection onto the support of $(\cF_{\{V_k\},\ell}\otimes\identitymap)(\ketbra{\Omega_{2^{\secp\ell}}}{\Omega_{2^{\secp\ell}}})$. 
    Then,
    \begin{align}
        \Tr[Q'(\cI_{\secp\to\secp+s,\ell}\otimes\identitymap)(\ketbra{\Omega_{2^{\secp\ell}}}{\Omega_{2^{\secp\ell}}})]
        \le\negl(\secp).
    \end{align}
\end{lemma}

We need the following lemma for the proof of \cref{lem:Haar_isometry_Choi_has_negligible_overlap}.

\begin{lemma}\label{lem:Haar_isometry_Choi}
    Let $\secp,s,\ell\in\N$ such that $2^{\secp+s}\ge \ell^{2}$. Then, we have
    \begin{align}
        \bigg\|
         (\cI_{\secp\to\secp+s,\ell}\otimes\identitymap)(\ketbra{\Omega_{2^{\secp\ell}}}{\Omega_{2^{\secp\ell}}})
         -\Exp_{\ket{\psi}\gets\sigma_{2^{2\secp+s}}}\ketbra{\psi}{\psi}^{\otimes\ell}
        \bigg\|_1
        \le O\left(\frac{\ell^2}{2^{\secp+s}}\right),
    \end{align}
    where $\cI_{\secp\to\secp+s,\ell}$ is the Haar random isometry map defined in \cref{def:Haar_isometry_map}
\end{lemma}

Before proving \cref{lem:Haar_isometry_Choi}, we give the proof of \cref{lem:Haar_isometry_Choi_has_negligible_overlap} assuming \cref{lem:Haar_isometry_Choi}.

\begin{proof}[Proof of \cref{lem:Haar_isometry_Choi_has_negligible_overlap}]
    First, we prove that $\Tr[Q']\le2^{(1+c)\ell}$.
    Note that
    \begin{align}
        &(\cF_{\{V_k\},\ell,\regA\to\regA\regB}\otimes\identitymap_{\regA'})(\ketbra{\Omega_{2^{\secp\ell}}}{\Omega_{2^{\secp\ell}}}_{\regA\regA'})
        \notag\\
        =&\frac{1}{|\cK_\secp|}\sum_{k\in\cK_\secp}
        \Tr_\regC[(V^{\otimes\ell}_{k,\regC\regB\regA}\otimes \identitymap_{\regA'})(\ketbra{0^c}{0^c}^{\otimes\ell}_\regC\otimes\ketbra{0^s}{0^s}^{\otimes\ell}_\regB\ketbra{\Omega_{2^{\secp\ell}}}{\Omega_{2^{\secp\ell}}}_{\regA\regA'})].
    \end{align}
    For each $k$, the rank of $\Tr_\regC[(V^{\otimes\ell}_{k,\regC\regB\regA}\otimes \identitymap_{\regA'})(\ketbra{0^c}{0^c}^{\otimes\ell}_\regC\otimes\ketbra{0^s}{0^s}^{\otimes\ell}_\regB\ketbra{\Omega_{2^{\secp\ell}}}{\Omega_{2^{\secp\ell}}}_{\regA\regA'})]$ is at most $\min\{2^{c\ell},2^{(2\secp+s)\ell}\}=2^{c\ell}$ since $(V^{\otimes\ell}_{k,\regC\regB\regA}\otimes \identitymap_{\regA'})(\ketbra{0^c}{0^c}^{\otimes\ell}_\regC\otimes\ketbra{0^s}{0^s}^{\otimes\ell}_\regB\ketbra{\Omega_{2^{\secp\ell}}}{\Omega_{2^{\secp\ell}}}_{\regA\regA'})$ is pure.
    Thus, the rank of $Q'$ is at most $2^{c\ell}\cdot|\cK_\secp|\le2^{(1+c)\ell}$, which implies $\Tr[Q']\le2^{(1+c)\ell}$.
    
    Note that the rank of $Q'$ is at most $2^\ell$ since it is the same as the rank of $(\cM_{\{\cI'_k\},\ell}\otimes\identitymap)(\ketbra{\Omega_{2^{\secp\ell}}}{\Omega_{2^{\secp\ell}}})$.
    Therefore, we have
    \begin{align}
        \Tr[Q'(\cI_{\secp\to\secp+s,\ell}\otimes\identitymap)(\ketbra{\Omega_{2^{\secp\ell}}}{\Omega_{2^{\secp\ell}}})]
        \le& 
        \Tr
        \bigg[
         Q'\Exp_{\ket{\psi}\gets\sigma_{2^{2\secp+s}}}\ketbra{\psi}{\psi}^{\otimes\ell}
        \bigg]
        +\negl(\secp)
        \tag{By \cref{lem:Haar_isometry_Choi}}\\
        =&\frac{\Tr[Q'\Pi_{\symetric}]}{\binom{2^{2\secp+s}+\ell-1}{\ell}}+\negl(\secp)
        \tag{By \cref{lem:Haar_states}}\\
        \le&\frac{2^{(1+c)\ell}}{\binom{2^{2\secp+s}+\ell-1}{\ell}}+\negl(\secp)
        \tag{By $\Tr[Q'\Pi_{\symetric}]\le\Tr[Q']\le2^{(1+c)\ell}$}\\
        \le&O\bigg(
        \frac{2^{(1+c)\ell}(\ell!)}{2^{(2\secp+s)\ell}}
        \bigg)
        +\negl(\secp)
        \notag\\
        \le&O\bigg(
        \frac{2^{(1+c)\ell} \ell^{\ell+1/2}e^{-\ell}}{2^{(2\secp+s)\ell}}
        \bigg)
        +\negl(\secp)
        \tag{By the Stirling's formula}\\
        =&O\bigg(\ell^{1/2}
        \bigg(\frac{2^ce^{-1} \ell}{2^{2\secp+s}}\bigg)^\ell
        \bigg)
        +\negl(\secp)
        \notag\\
        \le&O\bigg(\ell^{1/2}
        \bigg(\frac{\poly(\secp)}{2^{2\secp+s}}\bigg)^\ell
        \bigg)
        +\negl(\secp)
        \tag{By $c(\secp)=O(\log\secp)$}\\
        \le&\negl(\secp),
    \end{align}
    which concludes the proof.
\end{proof}

For the proof of \cref{lem:Haar_isometry_Choi}, we use the following.

\begin{lemma}[Haar Twirl Approximation, \cite{schuster2024random,HhaYam24}]\label{lem:Haar_approx}
    Let $n,\ell\in\N$ such that $2^n\ge \ell^{2}$. Let $\regA$ be a $n\ell$-qubit register, and $\regA'$ be some fixed register. Then, for any quantum state $\rho$ on the registers $\regA\regA'$,
    \begin{align}
        \left\|
        (\cM_{\mu_{2^n},\ell,\regA}\otimes\identitymap_{\regA'})(\rho_{\regA\regA'})-\sum_{\pi\in S_\ell}\frac{1}{2^{n\ell}}R_{\pi,\regA}\otimes\Tr_{\regA}[(R_{\pi,\regA}^\dag\otimes I_{\regA'})\rho_{\regA\regA'}]
        \right\|_1
        \le O\left(\frac{\ell^2}{2^n}\right).
    \end{align}
    Here, $S_\ell$ denotes the permutation group over $\ell$ elements, and
    for $\pi\in S_\ell$, $R_\pi$ is the permutation unitary such that $R_\pi\ket{x_1,...,x_\ell}=\ket{x_{\pi^{-1}(1)},...,x_{\pi^{-1}(\ell)}}$ for all $x_1,...,x_\ell\in\bit^n$.
\end{lemma}

From \cref{lem:Haar_approx} and a straightforward calculation, we can prove \cref{lem:Haar_isometry_Choi}.

\begin{proof}[Proof of \cref{lem:Haar_isometry_Choi}]
    We define $\regA$ and $\regA'$ are $\secp\ell$-qubit register, and $\regB$ is $s\ell$-qubit register.
    Let the output registers of $\cI_{\secp\to\secp+s,\ell,\regA}$ be $\regA$ and $\regB$.
    We prove \cref{lem:Haar_isometry_Choi} via the following hybrids:
    \begin{align}
        &\rho_{0,\regB\regA\regA'}\coloneqq
        (\cI_{\secp\to\secp+s,\ell,\regA}\otimes\identitymap_{\regA'})(\ketbra{\Omega_{2^{\secp\ell}}}{\Omega_{2^{\secp\ell}}}_{\regA\regA'})\\
        &\rho_{1,\regB\regA\regA'}\coloneqq
        \sum_{\pi\in S_\ell}\frac{1}{2^{(\secp+s)\ell}}R_{\pi,\regB\regA}\otimes\Tr_{\regB\regA}[(R_{\pi,\regB\regA}^\dag\otimes I_{\regA'})(\ketbra{0^{s\ell}}{0^{s\ell}}_{\regB}\otimes\ketbra{\Omega_{2^{\secp\ell}}}{\Omega_{2^{\secp\ell}}}_{\regA\regA'})]\\
        &\rho_{2,\regB\regA\regA'}\coloneqq
        \Exp_{\ket{\psi}\gets\sigma_{2^{2\secp+s}}}\ketbra{\psi}{\psi}^{\otimes\ell}_{\regB\regA\regA'}.
    \end{align}
    Since we have $\|\rho_0-\rho_1\|_1\le O(\ell^2/2^{\secp+s})$ from \cref{lem:Haar_approx} with $n=\secp+s$, it suffices to prove $\|\rho_1-\rho_2\|_1\le O(\ell^2/2^{\secp+s})$.
    Note that we have
    \begin{align}
        &(R_{\pi,\regB\regA}^\dag\otimes I_{\regA'})(\ketbra{0^{s\ell}}{0^{s\ell}}_{\regB}\otimes\ketbra{\Omega_{2^{\secp\ell}}}{\Omega_{2^{\secp\ell}}}_{\regA\regA'})
        \notag\\
        =&R^\dag_{\pi,\regB}\ketbra{0^{s\ell}}{0^{s\ell}}_{\regB}\otimes(R_{\pi,\regA}^\dag\otimes I_{\regA'})\ketbra{\Omega_{2^{\secp\ell}}}{\Omega_{2^{\secp\ell}}}_{\regA\regA'}
        \tag{$R^\dag_{\pi,\regB\regA}=R^\dag_{\pi,\regB}\otimes R^\dag_{\pi,\regA}$}\\
        =&\ketbra{0^{s\ell}}{0^{s\ell}}_{\regB}\otimes(R_{\pi,\regA}^\dag\otimes I_{\regA'})\ketbra{\Omega_{2^{\secp\ell}}}{\Omega_{2^{\secp\ell}}}_{\regA\regA'}
        \tag{$(0^s,...,0^s)$ is invariant under any permutation $\pi\in S_\ell$}\\
        =&\ketbra{0^{s\ell}}{0^{s\ell}}_{\regB}\otimes(I_\regA\otimes R_{\pi,\regA'})\ketbra{\Omega_{2^{\secp\ell}}}{\Omega_{2^{\secp\ell}}}_{\regA\regA'}.\label{eq:lem:Haar_isometry_Choi1}
    \end{align}
    In the last line, we used the fact that $R_\pi^\dag = R_\pi^\top$ for any $\pi \in S_\ell$, where $\top$ denotes the transpose.
    Thus, we have
    \begin{align}
        &R_{\pi,\regB\regA}\otimes\Tr_{\regB\regA}[(R_{\pi,\regB\regA}^\dag\otimes I_{\regB})(\ketbra{0^{s\ell}}{0^{s\ell}}_{\regB}\otimes\ketbra{\Omega_{2^{\secp\ell}}}{\Omega_{2^{\secp\ell}}}_{\regA\regA'})]
        \notag\\
        =&R_{\pi,\regB\regA}\otimes\Tr_{\regA}[(I_\regA\otimes R_{\pi,\regA'})\ketbra{\Omega_{2^{\secp\ell}}}{\Omega_{2^{\secp\ell}}}_{\regA\regA'}]
        \tag{By \cref{eq:lem:Haar_isometry_Choi1}}\\
        =&\frac{1}{2^{\secp\ell}}R_{\pi,\regB\regA}\otimes R_{\pi,\regA'}
        \notag\\
        =&\frac{1}{2^{\secp\ell}}R_{\pi,\regB\regA\regA'},\label{eq:lem:Haar_isometry_Choi2}
    \end{align}
    where we have used $R_{\pi,\regB\regA}\otimes R_{\pi,\regA'}=R_{\pi,\regB\regA\regA'}$ in the last line.
    From \cref{eq:lem:Haar_isometry_Choi2,lem:Haar_states}, we have
    \begin{align}
        \rho_{1,\regB\regA\regA'}
        =\sum_{\pi\in S_\ell}\frac{1}{2^{(2\secp+s)\ell}}R_{\pi,\regB\regA\regA'}
        =\frac{\ell!}{2^{(2\secp+s)\ell}}\binom{2^{2\secp+s}+\ell-1}{\ell}\Exp_{\ket{\psi}\gets\sigma_{2^{2\secp+s}}}\ketbra{\psi}{\psi}^{\otimes\ell}_{\regB\regA\regA'}.
    \end{align}
    By a straightforward calculation, we have $\frac{\ell!}{2^{(2\secp+s)\ell}}\binom{2^{2\secp+s}+\ell-1}{\ell}=1+O(\ell^2/2^{2\secp+s})$. 
    Therefore, since $\rho_2=\Exp_{\ket{\psi}\gets\sigma_{2^{2\secp+s}}}\ketbra{\psi}{\psi}^{\otimes\ell}$, we obtain $\|\rho_1-\rho_2\|_1\le O(\ell^2/2^{2\secp+s})\le O(\ell^2/2^{\secp+s})$, which concludes the proof.
\end{proof}

From \cref{lem:block-encoding_of_PRI_Choi_states,lem:Haar_isometry_Choi_has_negligible_overlap}, we are ready to prove \cref{lem:PRI_distinguisher}.
We restate it here.

\PRIdistinguisher*

\begin{proof}[Proof of \cref{lem:PRI_distinguisher}]
    Let $\rho=(\cF_{\{V_k\},\ell}\otimes\identitymap)(\ketbra{\Omega_{2^{\secp\ell}}}{\Omega_{2^{\secp\ell}}})$.
    From \cref{lem:block-encoding_of_PRI_Choi_states,lem:Haar_isometry_Choi_has_negligible_overlap}, they satisfy the condition in \cref{lem:distinguisher}.
    Therefore, we obtain \cref{lem:PRI_distinguisher} by applying \cref{lem:distinguisher}, which concludes the proof.
\end{proof}
\fi
\section{Oracle Separation Between PRIs with Short Stretch and PRIs with Large Stretch}
\label{sec:PRI_vs_PRI}

In this section, we prove the following theorem.

\begin{theorem}\label{thm:PRI_vs_PRI}
    Let $s(\secp)=O(\log\secp)$ and $t(\secp)=\Omega(\secp)$. Then, there exists a unitary oracle $\cO'$ relative to which
    \begin{itemize}
        \item adaptive PRIs with $t$ stretch exist, but
        \item non-adaptive and ancilla-free PRIs with $s$ stretch do not exist.
    \end{itemize}
    Here, it is allowed to query the inverse of $\cO'$.
\end{theorem}

Since PRIs with $s=0$ stretch are PRUs, we get the following theorem immediately by choosing $s=0$.

\begin{theorem}
    Let $t(\secp)=\Omega(\secp)$.
    Then, there exists a unitary oracle $\cO'$ relative to which
    \begin{itemize}
        \item adaptive PRIs with $t$ stretch exist, but
        \item non-adaptive and ancilla-free PRUs do not exist.
    \end{itemize}
     Here, it is allowed to query the inverse of $\cO'$.
\end{theorem}

\subsection{Separation Oracle}
\label{subsec:PRI_oracle}

We define a separation oracle.

\begin{definition}[Haar Random Isometry Unitary Oracle]\label{def:HRI_orcle}
    For a function $t:\N\to\N$, we define a unitary oracle $\HRI_t$ as follows:
        For each $n\in\N$ and $m\in\bit^n$, sample $(n+t(n))$-qubit unitary $U_{n,m}\in\Unitaries(2^{n+t(n)})$ from the Haar measure $\mu_{2^{n+t(n)}}$. Then, define the $(n+t(n)+1)$-qubit unitary
        \begin{align}
            \HRI_{t,n,m}\coloneqq\ketbra{0}{1}\otimes(\ketbra{0^t}{0^t}\otimes I_n)U^\dag_{n,k}+\ketbra{1}{0}\otimes U_{n,m}(\ketbra{0^t}{0^t}\otimes I_n)+I_\bot^{t,n,m},
        \end{align}
        where, $I_\bot^{t,n,m}$ is the identity on the subspace orthogonal to $\SpanSpace\{\ket{0}\otimes\ket{0^t}\ket{x},\ket{1}\otimes U_{n,m}(\ket{0^t}\ket{x})\}_{x\in\bit^n}$.
        We define $\HRI_t\coloneqq\{\HRI_{t,n}\}_{n\in\N}$, where $\HRI_{t,n}\coloneqq\{\HRI_{t,n,m}\}_{m\in\bit^n}$.
\end{definition}

\begin{remark}
    We have $\HRI_{t,n,m}=\HRI_{t,n,m}^\dag$ for any function $t$, $n\in\N$ and $m\in\bit^n$ regardless of the choice of $U_{n,m}$.
\end{remark}

Next, we construct a PRI with $t$ stretch relative to this oracle.

\begin{definition}[PRI Relative to HRI Oracle]\label{const:PRI}
    Let $t:\N\to\N$ be a function. We define a QPT algorithm $G^{\HRI_t}$ as follows:
    \begin{enumerate}
        \item Let $k\in\bit^\secp$ and $\secp$-qubit register $\regA$ be inputs.
        \item Prepare $\ket{0}_{\regX}\ket{0^{t}}_{\regY}$ on ancilla register $\regX\regY$.
        \item Apply $I_{\regX}\otimes U_{\secp,k\regY\regA}$ by querying the registers $\regX,\regY$ and $\regA$ to $\HRI_{t,\secp,k}$. 
        \item Output the registers $\regY$ and $\regA$.
    \end{enumerate} 
\end{definition}

\begin{remark}
    If the input state is a pure state $\ket{\psi}_{\regA}$, the output state is $G^{\HRI_t}(k,\ket{\psi}_{\regA})=U_{\secp,k\regY\regA}(\ket{0^t}_{\regY}\ket{\psi}_{\regA})$.
\end{remark}

Our goal is to prove the following theorem, which implies \cref{thm:PRI_vs_PRI}.

\begin{theorem}\label{thm:PRI_vs_PRI_wrt_HRI_oracle}
    Let $\cU$ be a $\unitaryPSPACE$ complete problem in \cref{lem:unitaryPSPACE_has_a_complete_problem}.
    Let $s(\secp)=O(\log\secp)$ and $t(\secp)=\Omega(\secp)$.
    Then, with probability $1$ over the choice of $\HRI_t$ defined in \cref{def:HRI_orcle}, the following are satisfied:
    \begin{enumerate}
        \item $G^{\HRI_t}$ in \cref{const:PRI} is an adaptive secure PRI with $t$ stretch relative to $(\HRI_t,\cU)$
        \item Non-adaptive PRIs with $s$ stertch do not exist relative to $(\HRI_t,\cU)$.
    \end{enumerate}
\end{theorem}
Since the proof strategy of the first item in \cref{thm:PRI_vs_PRI_wrt_HRI_oracle} is the same as \cref{thm:PRFSGs_relative_to_unitary_oracle}, we omit it.

\subsection{Breaking PRIs with Short Stretch}
\label{subsec:break_PRI_wrt_HRI_oracle}

In this subsection, we prove the second item in \cref{thm:PRI_vs_PRI_wrt_HRI_oracle}.
\ifnum\submission=0
We give an adversary in a similar way as in \cref{alg:break_PRU,alg:break_PRI}.
\fi
\ifnum\submission=1
We give an adversary in a similar way as in \cref{alg:break_PRU}
\fi
Suppose that $\cU$ is the $\unitaryPSPACE$-complete problem.
Let $G^{\HRI_t,\cU}$ be a QPT algorithm that satisfies the correctness of ancilla-free PRIs with $s$ stretch.
For such $G^{\HRI_t,\cU}$,
let $\{\cI_k\}_{k\in\cK_\secp}$ be the isometry implemented by $G^{\HRI_t,\cU}$ on input $k\in\cK_\secp$, where $\cK_\secp$ denotes the key-space.
We define the following map:
\begin{align}
    \cM_{\{\cI_k\},\ell}(\cdot)&=\Exp_{k\gets\cK_\secp}\cI^{\otimes\ell}_k(\cdot)\cI_k^{\dag\otimes\ell}.
\end{align}

To construct a PRI adversary, we need two lemmas as in \cref{sec:PRUs,sec:PRFSG_vs_PRI}.
These lemmas can be obtained by modifying \cref{lem:approx_PRI_Choi_state} and \cref{lem:PRI_distinguisher} for the $\HRI_t$ oracle.
\ifnum\submission=0
We give their proofs later.
\fi
\ifnum\submission=1
We give the proofs in the supplemental material.
\fi

\begin{restatable}{lemma}{PRIChoiStates}\label{lem:approx_PRI_Choi_state_wrt_HRI_oracle}
    Let $T(\secp)$ be a polynomial.
    Let $\epsilon>0$ and $d\in\N$.
    Let $\{\cI_k\}_{k\in\cK_\secp}$ be an ensemble of isometries mapping $\secp$ qubits to $\secp+s$ qubits, where each $\cI_k$ is QPT implementable with $s$ ancilla qubits by making $T$ queries to $(\HRI_{t},\cU)$. 
    For all $n\in[d]$ and $m\in\bit^n$, let $\HRI'_{t,n,k}$ be any unitary satisfying 
    \begin{align}
        \|\HRI_{t,n,m}(\cdot)\HRI_{t,n,m}^\dag-\HRI_{t,n,m}'(\cdot)\HRI_{t,n,m}^{'\dag}\|_\diamond\le\epsilon.
    \end{align}
    Then, for any polynomial $\ell$, there exists a family 
    $\{V_k\}_{k\in\cK_\secp}$ 
    of $(\secp+s+c)$-qubit unitaries 
    such that each $V_k$ is QPT implementable with classical descriptions of $\HRI'_{t,n,m}$ for all $n\in[d],m\in\bit^n$ and query access to $\cU$ such that it satisfies
    \begin{align}
        \bigg\|
         (\cM_{\{\cI_k\},\ell}\otimes \identitymap)(\ketbra{\Omega_{2^{\secp\ell}}}{\Omega_{2^{\secp\ell}}})-
         (\cF_{\{V_k\},\ell}\otimes \identitymap)(\ketbra{\Omega_{2^{\secp\ell}}}{\Omega_{2^{\secp\ell}}})
        \bigg\|_1\le O(\ell T\epsilon)
        +O\bigg(\frac{2^{s+3c/2}\ell T}{2^{t(d)/2}}\bigg).
        \label{eq:goal;approx_random_Choi_state_wrt_HRI_oracle}
    \end{align}
    Here, $\cF_{\{V_k\},\ell}$ is a CPTP map from $\secp\ell$ qubits to $(\secp+s)\ell$ qubits defined as follows:
    \begin{align}
        \cF_{\{V_k\},\ell}((\cdot)_\regA)\coloneqq
        \Tr_{\regC}
        \bigg[
         \Exp_{k\gets\cK_\secp}V^{\otimes\ell}_{k,\regA\regB\regC}((\cdot)_\regA\otimes\ketbra{0^{s}}{0^s}^{\otimes\ell}_{\regB}\otimes\ketbra{0^{c}}{0^c}^{\otimes\ell}_{\regC})
         V^{\dag\otimes\ell}_{k,\regA\regB\regC}
        \bigg],
    \end{align}
    where $\regA$ is a $\secp\ell$-qubit register, $\regB$ is a $s\ell$-qubit register, and $\regC$ is a $c\ell$-qubit register.
\end{restatable}

\ifnum\submission=0
Regarding the next lemma, recall that $\cI_{\secp\to\secp+s}$ is the Haar random isometry map defined in \cref{def:Haar_isometry_map}.
\fi
\ifnum\submission=1
Before introducing a new lemma, we define the following map.

\begin{definition}[Haar Random Isometry Map]\label{def:Haar_isometry_map}
    We define\footnote{Here, we chose $\ket{0^s}$ as an input state. This does not lose any generality due to the right invariance of the Haar measure.}
    \begin{align}
     \cI_{\secp\to \secp+s,\ell}(\rho_\regA)\coloneqq
     \Exp_{U\gets\mu_{2^{\secp+s}}} 
     \bigg(
      \bigotimes_{i\in[\ell]}U_{\regA_i\regB_i}
     \bigg)
     (\rho_\regA\otimes\ketbra{0^s}{0^s}^{\otimes\ell}_{\regB})
     \bigg(
      \bigotimes_{i\in[\ell]}U_{\regA_i\regB_i}
     \bigg)^\dag,
    \end{align}
    where $\regA\coloneqq\bigotimes_{i\in[\ell]}\regA_i$ and $\regB\coloneqq\bigotimes_{i\in[\ell]}\regB_i$, where, for each $i\in[\ell]$, $\regA_i$ and $\regB_i$ are $\secp$-qubit register and $s$-qubit register, respectively.
\end{definition}

Now we are ready to give the following lemma.
\fi

\begin{lemma}\label{lem:PRI_distinguisher_wrt_HRI_oracle}
    Suppose that $c(\secp)=O(\log\secp)$.
    Let $\{V_k\}_{k\in\cK_\secp}$ be a family of $(\secp+s+c)$-qubit unitaries, and $\cF_{\{V_k\}_k,\ell}$ be the CPTP map in \cref{lem:approx_PRI_Choi_state_wrt_HRI_oracle}.
    Let $\ell(\secp)\coloneqq\lceil\log|\cK_\secp|\rceil$.
    Then, there exists a QPT algorithm $\cD^{\cU}$ that, on input classical descriptions of $\HRI'_{t,n,m}$ for all $n\in[d]$ and $m\in\bit^n$, distinguishes $(\cF_{\{V_k\},\ell}\otimes\identitymap)(\ketbra{\Omega_{2^{\secp\ell}}}{\Omega_{2^{\secp\ell}}})$ from $(\cI_{\secp\to\secp+s,\ell}\otimes\identitymap)(\ketbra{\Omega_{2^{\secp\ell}}}{\Omega_{2^{\secp\ell}}})$
    with advantage at least $1-\negl(\secp)$.
\end{lemma}

From these lemmas, we construct a PRI adversary as shown in \cref{alg:break_PRI_wrt_HRI_oracle}.
\ifnum\submission=0
The {\color{red}red-highlighted lines} in \cref{alg:break_PRI_wrt_HRI_oracle} indicate the differences from \cref{alg:break_PRI}.
\fi
\ifnum\submission=1
The {\color{red}red-highlighted lines} in \cref{alg:break_PRI_wrt_HRI_oracle} indicate the differences from \cref{alg:break_PRU}.
\fi

\begin{algorithm}
    \caption{Adversary distinguishing $\{\cI_k\}_{k\in\cK_\secp}$ from Haar random isometry relative to $(\HRI,\cU)$.}
    \label{alg:break_PRI_wrt_HRI_oracle} 
    \vspace{2mm}
    \AlgOracle 
    The algorithm has query access to
    \begin{itemize}
        \item an isometry $\cI$ from $\secp$ qubits to $\secp+s(\secp)$ qubits, which is whether $\cI_k$ or a Haar random isometry.
        
        \item {\color{red} the oracle $\cO'=(\HRI_t,\cU)$ and its inverse, where $\HRI_t$ is defined in \cref{def:HRI_orcle}.}
        
    \end{itemize}
    \AlgInput The algorithm takes the security parameter $1^\secp$ as input.

    Define $\ell\coloneqq\lceil\log|\cK_\secp|\rceil$ and
    {\color{red}$d\coloneqq (2\log(\ell Tp)+2s+3c)^{\frac{1}{a}}$}
    where $p$ is a polynomial, and $T$ is the number of queries to {\color{red}$\cO'$} to implement $\cI_k$.

    \vspace{2mm}
    \begin{enumerate}
        \item 
        \label{algstep:process_tomography_PRI_wrt_HRI_oracle}
        For $n\in[d]$ {\color{red}and $m\in\bit^n$}, run the process tomography algorithm in \cref{thm:process_tomography_HKOT23} on inputs $\epsilon\coloneqq\frac{1}{\ell Tp}$ and $\eta\coloneqq2^{-\secp-1}$ for {\color{red}$\HRI_{t,n,m}$} to get a classical description of {\color{red}$\HRI'_{t,n,m}$}. Note that {\color{red}$\|\HRI_{t,n,m}(\cdot)\HRI_{t,n,m}^\dag-\HRI'_{t,n,m}(\cdot)\HRI'^\dag_{t,n,m}\|_\diamond\le\epsilon$} holds with probability at least $1-2^{-\secp}$ over the randomness of the process tomography algorithm.
        
        \item Prepare $(\cI^{\otimes\ell}\otimes I)\ket{\Omega_{2^{\secp\ell}}}$ by querying $\cI$.

        \item 
        Let $\{\cI'_k\}_{k\in\cK_\secp}$ be a family of isometries in {\color{red}\cref{lem:approx_PRI_Choi_state_wrt_HRI_oracle}.} Note that each $\cI'_k$ is QPT implementable with access to $\cU$ and classical descriptions of {\color{red}$\HRI'_{t,n,m}$} for all $n\in[d]$ {\color{red} and $m\in\bit^n$}, where such classical descriptions are obtained in the step \ref{algstep:process_tomography_PRI_wrt_HRI_oracle}.
        Let $\cD^{(\cdot)}$ be a QPT algorithm in {\color{red}\cref{lem:PRI_distinguisher_wrt_HRI_oracle}} for $\{\cI'_k\}_{k\in\cK_\secp}$. 
        By querying $\cU$, run $\cD^\cU$ on input $(\cI^{\otimes\ell}\otimes I)\ket{\Omega_{2^{\secp\ell}}}$ and classical descriptions of {\color{red}$\HRI'_{t,n,m}$} for all $n\in[d]$ {\color{red}and $m\in\bit^n$} to get $b\in\bit$. 
    \end{enumerate} 

    \AlgOutput The algorithm outputs $b$.
\end{algorithm}

\begin{remark}
    $d$ is chosen so that it satisfies $\ell T 2^{s+3c/2 - t(d)/2} = O(1/p)$ and ensures that the process tomography algorithm runs in QPT.
    \cref{alg:break_PRI_wrt_HRI_oracle} is QPT because, if $n\in[d]$, the dimension of $\HRI_{t,n,m}$ is at most $2^{d+t(d)+1}=O(2^{2t(d)})=O((\ell Tp)^4\cdot2^{4s+6c})\le\poly(\secp)$, where the inequality follows from $s(\secp)\le O(\log\secp)$ and $c(\secp)\le O(\log\secp)$.
\end{remark}

Now we are ready to construct the PRI adversary in \cref{alg:break_PRI_wrt_HRI_oracle}.
The following \cref{thm:break_PRI_wrt_HRI_oracle} implies the second statement in \cref{thm:PRI_vs_PRI_wrt_HRI_oracle}.

\begin{theorem}\label{thm:break_PRI_wrt_HRI_oracle}
Suppose that $s=O(\log\secp)$ and $t(\secp)=\Theta(\secp^a)$ for some constant $a\ge1$.
Let $\cO'\coloneqq(\HRI_t,\cU)$.
Let $G^{\cO'}$ be a QPT algorithm that satisfies the correctness of ancilla-free PRIs with $s$ stretch.
For such $G^{\cO'}$,
let $\{\cI_k\}_{k\in\cK_\secp}$ be the isometry mapping $\secp$ qubits to $\secp+s$ qubits implemented by $G^{\cO'}$ on input $k\in\cK_\secp$, where $\cK_\secp$ denotes the key-space.
Then, for any polynomial $p$, the QPT adversary $\cA^{(\cdot,\cdot)}$ defined in \cref{alg:break_PRI_wrt_HRI_oracle} satisfies
    \ifnum\submission=0
    \begin{align}
        \bigg|\Pr_{k\gets\cK_\secp}[1\gets\cA^{\cI_k,\cO'}(1^\secp)]-\Pr_{U\gets\mu_{2^{\secp+s}}}[1\gets\cA^{\cI_U,\cO'}(1^\secp)]\bigg|
        \ge 1-O\bigg(\frac{1}{p(\secp)}\bigg),
    \end{align}
    \fi
    \ifnum\submission=1
    $|\Pr_{k\gets\cK_\secp}[1\gets\cA^{\cI_k,\cO'}(1^\secp)]-\Pr_{U\gets\mu_{2^{\secp+s}}}[1\gets\cA^{\cI_U,\cO'}(1^\secp)]|
        \ge 1-O(1/p(\secp))$,
    \fi
    where, for each $U\in\Unitaries(2^{\secp+s(\secp)})$, $\cI_U$ is the isometry that maps $\secp$-qubit state $\ket{\psi}$ to $(\secp+s(\secp))$-qubit state $U(\ket{\psi}\ket{0^s})$. 
    Here, $\cA^{(\cdot),\cO'}$ queries not only $\cO'$ but also its inverse. Moreover, $\cA^{(\cdot),\cO'}$ queries the first oracle non-adaptively.
\end{theorem}

\ifnum\submission=0
\ifnum\submission=1
\section{Omitted Proofs in \cref{sec:PRI_vs_PRI}}

To conclude the proof of \cref{thm:break_PRI_wrt_HRI_oracle}, it remains to prove \cref{lem:approx_PRI_Choi_state_wrt_HRI_oracle,lem:PRI_distinguisher_wrt_HRI_oracle}.
Since the proof of \cref{lem:PRI_distinguisher_wrt_HRI_oracle} is the same as that of \cref{lem:PRI_distinguisher}, we omit it.
Thus, it suffices to prove \cref{lem:approx_PRI_Choi_state_wrt_HRI_oracle}.
\fi
\ifnum\submission=0
\subsection[Proof of Lemma~\ref{lem:approx_PRI_Choi_state_wrt_HRI_oracle}]{Proof of \cref{lem:approx_PRI_Choi_state_wrt_HRI_oracle}}

To conclude the proof, it remains to prove \cref{lem:approx_PRI_Choi_state_wrt_HRI_oracle,lem:PRI_distinguisher_wrt_HRI_oracle}.
Since the proof of \cref{lem:PRI_distinguisher_wrt_HRI_oracle} is the same as that of \cref{lem:PRI_distinguisher}, we omit it.
Thus, it suffices to prove \cref{lem:approx_PRI_Choi_state_wrt_HRI_oracle}.
\fi

The proof strategy is the same as that of \cref{lem:approx_random_Choi_state,lem:approx_PRI_Choi_state}.
First, we need the following lemma.

\begin{lemma}\label{lem:HRI_is_almost_identity_for_Choi_state}
    Let $n+t(n)+1\le\secp+c'$.
    Suppose that $\regA$ and $\regA'$ are $\secp$-qubit registers, and $\regB$ is an $s$-qubit register.
    Then, for any $U,V\in\Unitaries(2^{\secp+c'})$ and $m\in\bit^n$,
    \begin{align}
        \frac{1}{2}\|\ketbra{\psi}{\psi}-\ketbra{\phi}{\phi}\|_1\le O(2^{c'-t(n)/2}),
    \end{align}
    where 
    \begin{align}
        \ket{\psi}_{\regB\regA\regA'}&\coloneqq((U(\HRI_{t,n,m}\otimes I)V)_{\regB\regA}\otimes I_{\regA'})\ket{0^s}_\regB\ket{\Omega_{2^\secp}}_{\regA\regA'},\\
        \ket{\phi}_{\regB\regA\regA'}&\coloneqq((UV)_{\regB\regA}\otimes I_{\regA'})\ket{0^s}_\regB\ket{\Omega_{2^\secp}}_{\regA\regA'}
    \end{align}
    and $\HRI_{t,n,m}$ is defined in \cref{def:HRI_orcle}.
\end{lemma}

\begin{proof}[Proof of \cref{lem:HRI_is_almost_identity_for_Choi_state}]
    First, we can see 
    $(\HRI_{t,n,m}\otimes I^{\otimes 2\secp+2c'-n-t(n)-1})\ket{\Omega_{2^{\secp+c'}}}$
    is close to $\ket{\Omega_{2^{\secp+c'}}}$ in the trace norm as follows:
    \begin{align}
        &\frac{1}{2}\bigg\|(\HRI_{t,n,m}\otimes I^{\otimes 2\secp+2c'-n-t(n)-1})\ketbra{\Omega_{2^{\secp+c'}}}{\Omega_{2^{\secp+c'}}}(\HRI_{t,n,m}\otimes I^{\otimes 2\secp+2c'-n-t(n)-1})^\dag-\ketbra{\Omega_{2^{\secp+c'}}}{\Omega_{2^{\secp+c'}}}\bigg\|_1
        \notag\\
        =&\sqrt{1-|\bra{\Omega_{2^{\secp+c'}}}(\HRI_{t,n,m}\otimes I^{\otimes 2\secp+2c'-n-t(n)-1})\ket{\Omega_{2^{\secp+c'}}}|^2}
        \tag{By $\frac{1}{2}\|\ketbra{\alpha}{\alpha}-\ketbra{\beta}{\beta}\|_1=\sqrt{1-|\braket{\alpha|\beta}|^2}$}\\
        =&\sqrt{1-\frac{1}{2^{2\secp+2c'}}|\Tr[\HRI_{t,n,m}\otimes I^{\otimes \secp+c'-n-t(n)-1}]|^2}
        \notag\\
        =&\sqrt{1-\frac{1}{2^{2n+2t(n)+2}}|\Tr[\HRI_{t,n,m}]|^2}
        \tag{By $\Tr[A\otimes B]=\Tr[A]\Tr[B]$}\\
        =&\sqrt{1-\frac{(2^{n+t(n)+1}-2^{n+1})^2}{2^{2n+2t(n)+2}}}
        \tag{By $\Tr[\HRI_{t,n,m}]=2^{n+t(n)+1}-2^{n+1}$}\\
        \le&O\bigg(\frac{1}{2^{t(n)/2}}\bigg).\label{eq:HRI_is_almost_identity_for_Choi_state1}
    \end{align}
    Let us define an isometry $V'_{\regB\regA}\coloneqq V_{\regB\regA}(\ket{0^s}_\regB\otimes I_\regA)$ mapping $\regA$ to $\regB\regA$.
    With this $V'$, we have 
    \begin{align}
        \ket{\psi}_{\regB\regA\regA'}&=((U(\HRI_{t,n,m}\otimes I))_{\regB\regA}\cdot V'_{\regA}\otimes I_{\regA'})\ket{\Omega_{2^\secp}}_{\regA\regA'}
        \notag\\
        \ket{\phi}_{\regB\regA\regA'}&=(U_{\regB\regA}\cdot V'_{\regA}\otimes I_{\regA'})\ket{\Omega_{2^\secp}}_{\regA\regA'}
        \label{eq:HRI_is_almost_identity_for_Choi_state2}
    \end{align}
    Then, from \cref{eq:HRI_is_almost_identity_for_Choi_state1,eq:HRI_is_almost_identity_for_Choi_state2}, we obtain our statement as follows:
    \begin{align}
        &\frac{1}{2}\|\ketbra{\psi}{\psi}-\ketbra{\phi}{\phi}\|_1
        \notag\\
        =&\frac{1}{2}\bigg\|
         ((U(\HRI_{t,n,m}\otimes I))_{\regB\regA}\cdot V'_\regA\otimes \identitymap_{\regA'})(\Omega_{2^\secp,\regA\regA'})
         -(U_{\regB\regA}\cdot V'_\regA\otimes \identitymap_{\regA'})(\Omega_{2^\secp,\regA\regA'})
        \bigg\|_1
        \tag{By \cref{eq:HRI_is_almost_identity_for_Choi_state2}}\\
        =&
        \frac{2^{c'}}{2}\bigg\|
         ((U(\HRI_{t,n,m}\otimes I))_{\regB\regA}\otimes V'^\top_{\regB'\regA'})(\Omega_{2^{\secp+c'},\regB\regA\regB'\regA'})
         -(U_{\regB\regA}\otimes V'^\top_{\regB'\regA'})(\Omega_{2^{\secp+c'},\regB\regA\regB'\regA'})
        \bigg\|_1
        \tag{By \cref{claim:choi_state_on_isometry} and define $\regB'$ as an $c'$-qubit register}\\
        =&\frac{2^{c'}}{2}\bigg\|
         ((\HRI_{t,n,m}\otimes I)_{\regB\regA}\otimes \identitymap_{\regB'\regA'})(\Omega_{2^{\secp+c'},\regB\regA\regB'\regA'})
         -(\Omega_{2^{\secp+c'},\regB\regA\regB'\regA'})
        \bigg\|_1
        \tag{By H\"{o}lder's inequality (\cref{lem:Holder}) and $\|V'\|_\infty=\|V'^\top\|_\infty\le 1$ since $V'$ is an isometry}\\
        \le&O(2^{c'-t(n)/2}),\tag{By \cref{eq:HRI_is_almost_identity_for_Choi_state1}}
    \end{align}
    where $V'^\top$ denotes the transpose of $V'$.
\end{proof}

Now we are ready to prove \cref{lem:approx_PRI_Choi_state_wrt_HRI_oracle}.

\begin{proof}[Proof of \cref{lem:approx_PRI_Choi_state_wrt_HRI_oracle}]
    The proof strategy is the same as in \cref{lem:approx_PRI_Choi_state}.  
    The difference lies in applying \cref{lem:HRI_is_almost_identity_for_Choi_state} instead of \cref{lem:swap_unitary_is_almost_identity_for_Choi_state;ancilla}.  
    Thus, it suffices to replace the dependence on $d$ with the dependence on $t(d)$.
    Therefore, there exists a family 
    $\{V_k\}_{k\in\cK_\secp}$ 
    of $(\secp+s+c)$-qubit unitaries 
    such that each $V_k$ is QPT implementable with classical descriptions of $\HRI'_{t,n,m}$ for all $n\in[d]$ and $m\in\bit^n$, and query access to the $\unitaryPSPACE$-complete oracle $\cU$ such that it satisfies
    \begin{align}
        \bigg\|
         (\cM_{\{\cI_k\},\ell}\otimes \identitymap)(\ketbra{\Omega_{2^{\secp\ell}}}{\Omega_{2^{\secp\ell}}})-
         (\cF_{\{V_k\},\ell}\otimes \identitymap)(\ketbra{\Omega_{2^{\secp\ell}}}{\Omega_{2^{\secp\ell}}})
        \bigg\|_1\le O(\ell T\epsilon)
        +O\bigg(\frac{2^{s+3c/2}\ell T}{2^{t(d)/2}}\bigg),
    \end{align}
    which concludes the proof.
\end{proof}

\fi

\ifnum\anonymous=1
\else
\paragraph{Acknowledgments.}
The authors thank Prabhanjan Ananth for the helpful discussion.
S.Y. gratefully acknowledges Prabhanjan Ananth for hosting his visit to UCSB.
The authors gratefully acknowledge the anonymous reviewers of TQC and CRYPTO for identifying a bug in an earlier draft of this paper.
TM is supported by
JST CREST JPMJCR23I3,
JST Moonshot R\verb|&|D JPMJMS2061-5-1-1, 
JST FOREST, 
MEXT QLEAP, 
the Grant-in Aid for Transformative Research Areas (A) 21H05183,
and 
the Grant-in-Aid for Scientific Research (A) No.22H00522.
AG and YTL are supported by the National Science Foundation under the grants 
FET-2329938, 
CAREER-2341004,
and FET-2530160.
\fi

\ifnum\submission=0
\bibliographystyle{alpha} 
\else
\bibliographystyle{splncs04}
\fi
\bibliography{abbrev3,crypto,reference}
\clearpage

\appendix
\ifnum\submission=1
\clearpage
\vspace{2em}
\begin{center}
\textbf{\Large Supplementary Materials}
\end{center}

\section{Omitted Proofs in \cref{sec:PRUs}}

In this section, we give the omitted proofs in \cref{sec:PRUs}.

\subsection{Proof of \cref{thm:break_PRU}}

In this subsection, we give the proof of \cref{thm:break_PRU}.

\BreakPRU*

\begin{proof}[Proof of \cref{thm:break_PRU}]
    We construct $\cA^{(\cdot,\cdot)}$ as in \cref{alg:break_PRU}.
    It is clear that $\cA^{(\cdot,\cdot)}$ is a QPT algorithm. Step~\ref{algstep:process_tomography} runs in QPT because $2^d = (\ell T p)^2 + 2^c \leq \poly(\secp)$, given that $c = O(\log \secp)$.  
    Steps~\ref{algstep:prepare_Choi_state} and \ref{algstep:measurement} also run in QPT, as established in \cref{lem:PRU_distinguisher}.

    Assume that the tomography is successful in the step \ref{algstep:process_tomography}, namely, we have 
    $\|\cS_{n}(\cdot)\cS_{n}^\dag-\cS'_{n}(\cdot)\cS'^\dag_{n}\|_\diamond\le\epsilon$ for all $n\in[d]$. For notational simplicity, let $E$ denote this event. Note that
    \begin{align}
        \Pr[E]\ge1-2^{-\secp}\ge1-\negl(\secp)\label{eq:prob_tomography_succeeds}
    \end{align}
    from \cref{thm:process_tomography_HKOT23}. Thus, it suffices to show that $\cA$ can distinguish $\{U_k\}_k$ from Haar random unitaries when the tomography succeeds.
    Recall that $\{V_k\}_k$ is a family of unitaries in the step \ref{algstep:measurement}, and $\cD^{(\cdot)}$ is a QPT algorithm in the step \ref{algstep:measurement} for $\{V_k\}_k$.
    Note that
    \begin{align}
        \Pr_{k\gets\cK_\secp}[1\gets\cA^{U_k,\cO}(1^\secp)|E]
        =\Pr\bigg[
         1\gets\cD^\cU\bigg(
          (\cM_{\{U_k\},\ell}\otimes \identitymap)(\ketbra{\Omega_{2^{\secp\ell}}}{\Omega_{2^{\secp\ell}}})
         \bigg)
        \bigg]\label{eq:probability_when_querying_PRU}
    \end{align}
    and
    \begin{align}
        \Pr_{U\gets\mu_{2^\secp}}[1\gets\cA^{U,\cO}(1^\secp)|E]
        =\Pr\bigg[
         1\gets\cD^\cU\bigg(
         (\cM_{\mu_{2^\secp,\ell}}\otimes \identitymap)(\ketbra{\Omega_{2^{\secp\ell}}}{\Omega_{2^{\secp\ell}}})
         \bigg)
        \bigg].\label{eq:probability_when_querying_Haar}
    \end{align}
    We prove our claim by the standard hybrid argument. 
    First, we replace $\rho_0\coloneqq(\cM_{\{U_k\},\ell}\otimes \identitymap)(\ketbra{\Omega_{2^{\secp\ell}}}{\Omega_{2^{\secp\ell}}})$ in \cref{eq:probability_when_querying_PRU} with $\rho_1\coloneqq(\cE_{\{V_k\},\ell}\otimes \identitymap)(\ketbra{\Omega_{2^{\secp\ell}}}{\Omega_{2^{\secp\ell}}})$. We obtain the following claim.
    
    \begin{claim}\label{claim:hyb0_to_hyb1}
        \begin{align}
            &\bigg|
             \Pr\bigg[
              1\gets\cD^\cU(\rho_0)
             \bigg]
             -\Pr\bigg[
              1\gets\cD^\cU(\rho_1)
             \bigg]
            \bigg|
            \le O\bigg(\frac{1}{p(\secp)}\bigg).
            \notag
        \end{align}
    \end{claim}

    \begin{proof}[Proof of \cref{claim:hyb0_to_hyb1}]
        From \cref{lem:approx_random_Choi_state}, we have
         \begin{align}
         \frac{1}{2}\|
          \rho_0-
          \rho_1
         \|_1
         \le O(\ell T\epsilon)+O\bigg(\frac{2^{c/2}\ell T}{2^{d/2}}\bigg)
         \le O\bigg(\frac{1}{p(\secp)}\bigg).
         \notag
         \end{align}
         Here, $\epsilon=\frac{1}{\ell Tp}$ and $d=2\log(\ell Tp)+c$. 
    \end{proof}

    Next, we replace $\rho_1$ with 
    $\rho_2\coloneqq(\cM_{\mu_{2^\secp},\ell}\otimes \identitymap)(\ketbra{\Omega_{2^{\secp\ell}}}{\Omega_{2^{\secp\ell}}})$. 
    We have the following claim. Because its proof is straightforward from \cref{lem:PRU_distinguisher}, we omit it.

    \begin{claim}\label{claim:hyb1_to_hyb2}
        \begin{align}
            &\bigg|\Pr\bigg[1\gets\cD^\cU(\rho_1)\bigg]-
              \Pr\bigg[1\gets\cD^\cU(\rho_2))\bigg]
            \bigg|
            \ge 1-\negl(\secp).
            \notag
        \end{align}
    \end{claim}

    Combing \cref{eq:prob_tomography_succeeds,eq:probability_when_querying_PRU,eq:probability_when_querying_Haar} with \cref{claim:hyb0_to_hyb1,claim:hyb1_to_hyb2}, we have
    \begin{align}
        &\bigg|\Pr_{k\gets\cK_\secp}[1\gets\cA^{U_k,\cO}(1^\secp)]-\Pr_{U\gets\mu_{2^\secp}}[1\gets\cA^{U,\cO}(1^\secp)]\bigg|
        \notag\\
        \ge&\Pr[E]\bigg|\Pr_{k\gets\cK_\secp}[1\gets\cA^{U_k,\cO}(1^\secp)|E]-\Pr_{U\gets\mu_{2^\secp}}[1\gets\cA^{U,\cO}(1^\secp)|E]\bigg|
        -2\Pr[\Bar{E}]
        \notag\\
        \ge&(1-\negl(\secp))
        \bigg(
         1-O\bigg(\frac{1}{p(\secp)}\bigg)-\negl(\secp)
        \bigg)
        -\negl(\secp)
        \notag\\
        \ge& 1-O\bigg(\frac{1}{p(\secp)}\bigg),
    \end{align}
    which concludes the proof.
\end{proof}

\subsection{Proof of \cref{lem:swap_unitary_is_almost_identity_for_Choi_state;ancilla}}

In this section, we give the proof of \cref{lem:swap_unitary_is_almost_identity_for_Choi_state;ancilla}

\SwapChoi*

\begin{proof}[Proof of \cref{lem:swap_unitary_is_almost_identity_for_Choi_state;ancilla}]
    Define the following states:
    \begin{align}
        \ket{\psi'}_{\regA\regB\regA'\regB'} & \coloneqq ( U ( \cS_{n} \otimes I ) V )_{\regA\regB} \otimes I_{\regA'\regB'} \cdot \ket{\Omega_{2^{\secp+c'}}}_{\regA\regB\regA'\regB'}, \\
        \ket{\phi'}_{\regA\regB\regA'\regB'} & \coloneqq ( U V )_{\regA\regB} \otimes I_{\regA'\regB'} \cdot \ket{\Omega_{2^{\secp+c'}}}_{\regA\regB\regA'\regB'},
    \end{align}
    where $\regB'$ is a $c'$-qubit register and $\ket{\Omega_{2^{\secp+c'}}}$ is across $(\regA\regB,\regA'\regB')$.

    We can prove the following inequality, which we will prove later.
    \begin{align}
        \|\ket{\psi'} - \ket{\phi'}\| \le O \bigg(\frac{1}{2^{n/2}}\bigg).
        \label{eq:closeness}
    \end{align}
    Moreover, by definition, it holds that
    \begin{align}
        \ket{\psi}_{\regA\regA'\regB} & = \sqrt{2^{c'}} \cdot I_{\regA\regA'\regB} \otimes \bra{0^{c'}}_{\regB'} \cdot \ket{\psi'}_{\regA\regB\regA'\regB'}, 
        \notag \\
        \ket{\phi}_{\regA\regA'\regB} & = \sqrt{2^{c'}} \cdot I_{\regA\regA'\regB} \otimes \bra{0^{c'}}_{\regB'} \cdot \ket{\phi'}_{\regA\regB\regA'\regB'}.
        \label{eq:post_selection}
    \end{align}
    Putting them together, we obtain
    \begin{align}
        & \, \frac{1}{2}\|\ketbra{\psi}{\psi}-\ketbra{\phi}{\phi}\|_1 
        \notag \\
        \le & \, \|\ket{\psi} - \ket{\phi}\| 
        \tag{Since trace distance between pure states is bounded by their Euclidean distance} \\
        = & \, \sqrt{2^{c'}} \cdot \| I \otimes \bra{0^{c'}} \cdot (\ket{\psi'} - \ket{\phi'})\| 
        \tag{By~\cref{eq:post_selection}} \\
        \le & \, \sqrt{2^{c'}} \cdot \| I \otimes \bra{0^{c'}} \|_{\infty} \cdot \|\ket{\psi'} - \ket{\phi'}\| 
        \tag{By~$\|{A \ket{v}} \| \le \|A\|_\infty \|\ket{v}\|$} \\
        = & \, O\bigg(\frac{2^{c'/2}}{2^{n/2}}\bigg)
        \tag{By~$\| I \otimes \bra{0^{c'}} \|_{\infty} = 1$ and~\cref{eq:closeness}}
    \end{align}
    as desired.

    To conclude the proof, we prove \cref{eq:closeness}.
    Define the projection 
    \begin{align}\label{eq:projection_subspace}
        \Pi_{n}\coloneqq\sum_{m\in\bit^{n}}\ketbra{m}{m}\otimes I_\bot^{n,m}
    \end{align} 
    onto the subspace on which $\cS_{n}$ acts as the identity. It then follows that
    \begin{align}
        (I - \cS_{n}) \Pi_n = 0.
    \end{align}
    First, we can see 
    $(\cS_{n}\otimes I^{\otimes 2\secp+2c'-2n-1})\ket{\Omega_{2^{\secp+c}}}$
    is close to $\ket{\Omega_{2^{\secp+c}}}$ in Euclidean distance as follows:
    \begin{align}
        & \|\ket{\Omega_{2^{\secp+c'}}} - (\cS_{n}\otimes I^{\otimes 2\secp+2c'-2n-1}) \ket{\Omega_{2^{\secp+c'}}}\| \\
        = & \|(I^{\otimes 2n+1} - \cS_{n})\otimes I^{\otimes 2\secp+2c'-2n-1} \ket{\Omega_{2^{\secp+c'}}}\| \\
        = & \|(I^{\otimes 2n+1} - \cS_{n}) (I-\Pi_n) \otimes I^{\otimes 2\secp+2c'-2n-1} \ket{\Omega_{2^{\secp+c'}}}\| 
        \tag{By~\cref{eq:projection_subspace}} \\
        \le & \|(I^{\otimes 2n+1} - \cS_{n})\otimes I^{\otimes 2\secp+2c'-2n-1}\|_\infty \cdot \|(I^{\otimes 2n+1}-\Pi_n)\otimes I^{\otimes 2\secp+2c'-2n-1} \ket{\Omega_{2^{\secp+c'}}}\| 
        \tag{By~$\|A\ket{v}\| \le \|A\|_\infty \|\ket{v}\|$} \\
        \le & 2 \|(I^{\otimes 2n+1}-\Pi_n)\otimes I^{\otimes 2\secp+2c'-2n-1} \ket{\Omega_{2^{\secp+c'}}}\| 
        \tag{By the triangle inequality and the fact that unitaries have unit operator norm} \\
        =& 2\sqrt{\frac{1}{2^{\secp+c'}}\Tr[(I^{\otimes 2n+1}-\Pi_n)\otimes 2^{\secp+c'-2n-1}]}
        \tag{By $||(A\otimes I)\ket{\Omega_D}||^2=\frac{1}{D}\Tr[A^\dag A]$ for any $A\in\Unitaries(D)$}
        \\
        =& 2\sqrt{\frac{1}{2^{2n+1}}\Tr[I^{\otimes 2n+1}-\Pi_n]}
        \tag{By $\Tr[A\otimes B]=\Tr[A]\Tr[B]$ forn any matrix $A$ and $B$}
        \\
        =& 2\sqrt{\frac{1}{2^{2n+1}}\sum_{m\in\bit^n}\Tr[\ketbra{m}{m}\otimes(I^{\otimes n+1}-I^{n,m}_\bot)]}
        \tag{By \cref{eq:projection_subspace}}
        \\
        \le & O(2^{n/2}). 
    \end{align}
    Here, in the last line, we have used that $\Tr[I^{\otimes n+1}-I^{n,m}_\bot]=2$.
    Since $\ket{\psi'}_{\regA\regB\regA'\regB'}=(U_{\regA\regB}\otimes V^\top_{\regA'\regB'})((\cS_{n}\otimes I)_{\regA\regB}\otimes I_{\regA'\regB'})\ket{\Omega_{2^{\secp+c}}}_{\regA\regB\regA'\regB'}$ and $\ket{\phi'}_{\regA\regB\regA'\regB'}=(U_{\regA\regB}\otimes V^\top_{\regA'\regB'})\ket{\Omega_{2^{\secp+c}}}_{\regA\regB\regA'\regB'}$, we obtain \cref{eq:closeness}, which concludes the proof.
\end{proof}

\subsection{Proof of \cref{lem:block-encoding_of_random_Choi_states}}

In this subsection, we give the proof of \cref{lem:block-encoding_of_random_Choi_states}.
For the proof, we need the following lemma.

\begin{lemma}[Lemma 12 in \cite{ICALP:vApGil19}]\label{lem:block-encoding_of_state}
    Let $U$ be a unitary over registers $\regA$ and $\regB$, where $\regA$ and $\regB$ are $n$-qubit register and $m$-qubit register, respectively. Define $\rho_{\regA}\coloneqq\Tr_{\regB}[(U\ketbra{0...0}{0...0}U^\dag)_{\regA\regB}]$. Then, there exists a $(1,0,n+m)$-block encoding unitary $V$ of $\rho$, where $V$ is implementable with single use of $U$ and $U^\dag$, and $n+1$ two-qubit gates.
\end{lemma}

Now we are ready to prove \cref{lem:block-encoding_of_random_Choi_states}.
For the reader's convenience, we restate it here.

\BlockEncoding*

\begin{proof}[Proof of \cref{lem:block-encoding_of_random_Choi_states}]
    Let $x$ denote the concatenation of classical descriptions of $\cS'_n$ for all $n\in[d]$.
    Let $A$ and $B$ be unitaries such that $BA$ is a purification unitary of 
    $(\cM_{\{U'_k\},\ell}\otimes\identitymap)(\ketbra{\Omega_{2^{\secp\ell}}}{\Omega_{2^{\secp\ell}}})$.
    In other words,
    tracing out some qubits of $BA|0...0\rangle$ is
    equal to
    $(\cM_{\{V_k\},\ell}\otimes\identitymap)(\ketbra{\Omega_{2^{\secp\ell}}}{\Omega_{2^{\secp\ell}}})$.
   $A$ is a unitary that maps $|0...0\rangle$ to $|x\rangle$.
   $B$ is the following unitary:
    \begin{enumerate}
        \item First, map $\ket{0...0}\ket{x}\mapsto\frac{1}{\sqrt{|\cK_\secp|}}\sum_{k\in\cK_\secp}\ket{0^c}^{\otimes\ell}\ket{\Omega_{2^{\secp\ell}}}\ket{0...0}\ket{k}\ket{x}$.
        \item 
        Then, map $
        \frac{1}{\sqrt{|\cK_\secp|}}
        \sum_{k\in\cK_\secp}\ket{0^c}^{\otimes\ell}\ket{\Omega_{2^{\secp\ell}}}\ket{0...0}\ket{k}\ket{x}
        \mapsto\frac{1}{\sqrt{|\cK_\secp|}}
        \sum_{k\in\cK_\secp}(V_k^{\otimes\ell}\otimes I)(\ket{0^c}^{\otimes\ell}\ket{\Omega_{2^{\secp\ell}}})\ket{0...0}\ket{k}\ket{x}$.
    \end{enumerate}
    Clearly, $A$ is QPT implementable given $x$. 
    $B$ can be approximately QPT implementable with an exponentially-small error by querying $\cU$, 
    because of the following reason:
    The first step of $B$ is QPT implementable.
    For the second step of $B$, we have only to show that each controlled-$V_k$ is approximately QPT implementable by querying $\cU$.
    In fact, first, $V_k$ is QPT implementable on input $x,k$ and by querying $\cU$.
    Second, in order to implement the controlled-$V_k$, we need the controlled-$\cU$.
    The controlled-$\cU$ is in $\unitaryPSPACE$ from \cref{remark:contrlizatin_of_pureUnitaryPSPACE}, therefore it
    is approximately QPT implementable with an exponentially-small error by querying $\cU$.
    
    Thus, for any polynomial $p$, there exists a QPT algorithm $\cB$ that implements $B$ with error $2^{-p(\secp)}$ by querying $\cU$. Namely, it satisfies
    \begin{align}
        \|\cB^\cU(\cdot)-B(\cdot)B^\dag\|_\diamond\le2^{-p(\secp)}.
    \end{align}
    Since $\cB$ is a QPT algorithm, it queries $\cU$ at most polynomially many times. Thus, by postponing all intermediate measurements, we can assume that $\cB$ applies a QPT unitary $C$ by querying $\cU$.
    Thus, we have
    \begin{align}
        \|\Tr_{\regY}[C_{\regX\regY}((\cdot)_{\regX}\otimes\ketbra{0...0}{0...0}_{\regY})C^\dag_{\regX\regY}]-(B(\cdot)B^\dag)_{\regX}\|_\diamond\le2^{-p(\secp)},\label{eq:approx_of_purification_for_random_Choi_states}
    \end{align}
    where $\regX$ and $\regY$ denote the main register and the ancilla register, respectively.
    Suppose that $\regA$ and $\regA'$ are $\secp\ell$-qubit registers, and $\regB$ is a $c\ell$-qubit register.
    We decompose $\regX$ as $\regX_0\coloneqq\regB\regA\regA'$ and $\regX_1$, where $\regX_0$ is the first $2\secp\ell$ qubits, and $\regX_1$ is the other qubits.
    From \cref{eq:approx_of_purification_for_random_Choi_states} and $\Tr_{\regX_1}[(BA\ketbra{0...0}{0...0}A^\dag B^\dag)_{\regX}]=(\cM_{\{V_k\},\ell,\regB\regA}\otimes\identitymap_{\regA'})(\ketbra{0^c}{0^c}^{\otimes\ell}_\regB\otimes\ketbra{\Omega_{2^{\secp\ell}}}{\Omega_{2^{\secp\ell}}}_{\regA\regA'})$, we have
    \begin{align}
        &\|\Tr_{\regB\regX_1\regY}[C_{\regX\regY}A_{\regX}(\ketbra{0...0}{0...0}_{\regX\regY})A^\dag_\regX C^\dag_{\regX\regY}]
        -(\cE_{\{V_k\},\ell,\regA}\otimes \identitymap_{\regA'})(\ketbra{\Omega_{2^{\secp\ell}}}{\Omega_{2^{\secp\ell}}}_{\regA\regA'})\|_1
        \le2^{-p(\secp)}.\label{eq:approx_of_random_Choi_states}
    \end{align}

    Now we are ready to construct block-encoding of $\Tr_{\regB}[(\cM_{\{V_k\},\ell,\regB\regA}\otimes\identitymap)(\ketbra{0^c}{0^c}^{\otimes\ell}_\regB\otimes\ketbra{\Omega_{2^{\secp\ell}}}{\Omega_{2^{\secp\ell}}})_{\regA\regA'}]$. From \cref{lem:block-encoding_of_state}, there exists a $(1,0,\poly(\secp))$-block encoding unitary $V_\secp$ of
    \begin{align}
        \sigma_{\regA\regA'}\coloneqq\Tr_{\regB\regX_1\regY}[C_{\regX\regY}A_{\regX}(\ketbra{0...0}{0...0}_{\regX\regY})A^\dag_\regX C^\dag_{\regX\regY}].
    \end{align}
    From \cref{lem:block-encoding_of_state} $V_\secp$ can be realized with a single use of $C_{\regX\regY}A_{\regX}$ and its inverse, and $\poly(\secp)$ two-qubit gates. 
    Therefore, $V_\secp$ is QPT implementable with $x$ and with the query access to $\cU$. Moreover, $V_\secp$ satisfies
    \begin{align}
        &\|((\bra{0^{\poly(\secp)}}\otimes I)V_{\secp}(\ket{0^{\poly(\secp)}}\otimes I))_{\regA\regA'}
        -(\cE_{\{V_k\},\ell,\regA}\otimes \identitymap_{\regA'})(\ketbra{\Omega_{2^{\secp\ell}}}{\Omega_{2^{\secp\ell}}}_{\regA\regA'})
        \|_\infty
        \notag\\
        =&
        \|\sigma_{\regA\regA'}
        -(\cE_{\{V_k\},\ell,\regA}\otimes \identitymap_{\regA'})(\ketbra{\Omega_{2^{\secp\ell}}}{\Omega_{2^{\secp\ell}}}_{\regA\regA'})\|_\infty
        \tag{Since $V_\secp$ is an $(1,0,\poly(\secp))$-block encoding of $\sigma$}\\
        \le&
        \|\sigma_{\regA\regA'}
        -(\cE_{\{V_k\},\ell,\regA}\otimes \identitymap_{\regA'})(\ketbra{\Omega_{2^{\secp\ell}}}{\Omega_{2^{\secp\ell}}}_{\regA\regA'})\|_1
        \tag{By $\|A\|_\infty\le\|A\|_1$}\\
        \le&2^{-p(\secp)},
        \tag{By \cref{eq:approx_of_random_Choi_states}}
    \end{align}
    which implies that $V_{\secp}$ is a $(1,2^{-p(\secp)},\poly(\secp))$-block encoding of $(\cE_{\{V_k\},\ell,\regA}\otimes \identitymap_{\regA'})(\ketbra{\Omega_{2^{\secp\ell}}}{\Omega_{2^{\secp\ell}}}_{\regA\regA'})$.
\end{proof}

\subsection{Proof of \cref{lem:Haar_Choi_has_negligible_overlap}}

In this subsection, we give the proof of \cref{lem:Haar_Choi_has_negligible_overlap}.
To this end, we need the following lemma.

\begin{lemma}[Implicitly Shown in \cite{Harrow23}]\label{lem:Haar_state_vs_Haar_choi_state}
    Let $\ell,d\in\N$ such that $d\ge\ell^2$. Then,
    \begin{align}
         \bigg\|
         (\cM_{\mu_d,\ell}\otimes \identitymap)(\ketbra{\Omega_{d^\ell}}{\Omega_{d^\ell}})
         -
         \Exp_{\ket{\psi}\gets\sigma_{d^2}}\ketbra{\psi}{\psi}^{\otimes\ell}
         \bigg\|_1
         \le O\bigg(\frac{\ell^2}{d}\bigg)
    \end{align}
    where $\ket{\Omega_{d^\ell}}=\frac{1}{d^{\ell/2}}\sum_{x\in[d^\ell]}\ket{x}\ket{x}$ is the maximally entangled state.
\end{lemma}

Now we are ready to prove \cref{lem:Haar_Choi_has_negligible_overlap}.
We restate it here for the reader's convenience.

\SmallOverlap*

\begin{proof}[Proof of \cref{lem:Haar_Choi_has_negligible_overlap}]
    First, we prove that $\Tr[Q]\le2^{(1+c)\ell}$.
    Note that
    \begin{align}
        &(\cE_{\{V_k\},\ell,\regA}\otimes \identitymap_{\regA'})(\ketbra{\Omega_{2^{\secp\ell}}}{\Omega_{2^{\secp\ell}}}_{\regA\regA'})
        =\frac{1}{|\cK_\secp|}\sum_{k\in\cK_\secp}
        \Tr_\regB[(V^{\otimes\ell}_{k,\regB\regA}\otimes \identitymap_{\regA'})(\ketbra{0^c}{0^c}^{\otimes\ell}_\regB\otimes\ketbra{\Omega_{2^{\secp\ell}}}{\Omega_{2^{\secp\ell}}}_{\regA\regA'})],
    \end{align}
    where $\regA$ and $\regA'$ are $\secp\ell$-qubit registers, and $\regB$ is a $c\ell$-qubit register.
    For each $k$, the rank of $\Tr_\regB[(V^{\otimes\ell}_{k,\regB\regA}\otimes \identitymap_{\regA'})(\ketbra{0^c}{0^c}^{\otimes\ell}_\regB\otimes\ketbra{\Omega_{2^{\secp\ell}}}{\Omega_{2^{\secp\ell}}}_{\regA\regA'})]$ is at most $\min\{2^{c\ell},2^{2\secp\ell}\}=2^{c\ell}$ since $(V^{\otimes\ell}_{k,\regB\regA}\otimes \identitymap_{\regA'})(\ketbra{0^c}{0^c}^{\otimes\ell}_\regB\otimes\ketbra{\Omega_{2^{\secp\ell}}}{\Omega_{2^{\secp\ell}}}_{\regA\regA'})$ is pure.
    Thus, the rank of $Q$ is at most $2^{c\ell}\cdot|\cK_\secp|\le2^{(1+c)\ell}$, which implies $\Tr[Q]\le2^{(1+c)\ell}$.
    
    Having this, 
    \begin{align}
        \Tr[Q(\cM_{\mu_{2^\secp},\ell}\otimes \identitymap)(\ketbra{\Omega_{2^{\secp\ell}}}{\Omega_{2^{\secp\ell}}})]
        \le&
        \Tr[Q\Exp_{\ket{\psi}\gets\sigma_{2^{2\secp}}}\ketbra{\psi}{\psi}^{\otimes\ell}]
        +O\bigg(\frac{\ell^2}{2^\secp}\bigg)
        \tag{By \cref{lem:Haar_state_vs_Haar_choi_state}}
        \\
        =&
        \frac{\Tr[Q\Pi_{\symetric}]}{\binom{2^{2\secp}+\ell-1}{\ell}}+\negl(\secp)
        \notag
        \\
        \le&\frac{2^{(1+c)\ell}}{\binom{2^{2\secp}+\ell-1}{\ell}}+\negl(\secp)
        \tag{By $\Tr[Q\Pi_\symetric]\le\Tr[Q]\le2^{(1+c)\ell}$}\\
        \le&O\bigg(
        \frac{2^{(1+c)\ell}(\ell!)}{2^{2\secp\ell}}
        \bigg)
        +\negl(\secp)
        \notag\\
        \le&O\bigg(
        \frac{2^{(1+c)\ell} \ell^{\ell+1/2}e^{-\ell}}{2^{2\secp\ell}}
        \bigg)
        +\negl(\secp)
        \tag{By the Stirling's formula, $\ell!\le \ell^{\ell+1/2}e^{-\ell+1}$}\\
        =&O\bigg(\ell^{1/2}
        \bigg(\frac{2^{1+c}e^{-1} \ell}{2^{2\secp}}\bigg)^\ell
        \bigg)
        +\negl(\secp)
        \notag\\
        =&O\bigg(\ell^{1/2}
        \bigg(\frac{\poly(\secp)}{2^{2\secp}}\bigg)^\ell
        \bigg)
        +\negl(\secp)
        \tag{By $c(\secp)=O(\log\secp)$}\\
        \le&\negl(\secp),\label{eq:xi_has_negligible_overlap_with_Q}
    \end{align}
    which concludes the proof.
\end{proof}

\subsection{Proof of \cref{lem:RSV_subspace_and_support}}

In this subsection, we show \cref{lem:RSV_subspace_and_support}. We restate it here for the convenience.

\RSVsubspace*

\begin{proof}[Proof of \cref{lem:RSV_subspace_and_support}]
    For the notational simplicity, we define 
    $\rho\coloneqq
    (\cM_{\{U'_k\},\ell}\otimes\identitymap)(\ketbra{\Omega_{2^{\secp\ell}}}{\Omega_{2^{\secp\ell}}})$. Then, we have
    \begin{align}
        &\Tr[\rho\Pi_{\ge\epsilon}]
        \\
        =&\bigg|
         \Tr[(\bra{0^{\poly(\secp)}}\otimes I)V_\secp(\ket{0^{\poly(\secp)}}\otimes I)\Pi_{\ge\epsilon}]+
        \Tr[(\rho-(\bra{0^{\poly(\secp)}}\otimes I)V_\secp(\ket{0^{\poly(\secp)}}\otimes I))\Pi_{\ge\epsilon}]
        \bigg|
        \\
        \ge&\bigg|
         \Tr[(\bra{0^{\poly(\secp)}}\otimes I)V_\secp(\ket{0^{\poly(\secp)}}\otimes I)\Pi_{\ge\epsilon}]
        \bigg|
         -\bigg\|(\rho-(\bra{0^{\poly(\secp)}}\otimes I)V_\secp(\ket{0^{\poly(\secp)}}\otimes I))\Pi_{\ge\epsilon}
         \bigg\|_1
        \tag{By the triangle inequality and $|\Tr[A]|\le\|A\|_1$}\\
        \ge&
        \bigg|
         \Tr[(\bra{0^{\poly(\secp)}}\otimes I)V_\secp(\ket{0^{\poly(\secp)}}\otimes I)\Pi_{\ge\epsilon}]
        \bigg|
        -\bigg\|\rho-(\bra{0^{\poly(\secp)}}\otimes I)V_\secp(\ket{0^{\poly(\secp)}}\otimes I)
        \bigg\|_1.
        \tag{By H\"{o}lder's inequality (\cref{lem:Holder}) and $\|\Pi_{\ge\epsilon}\|_\infty=1$}
   \end{align}
    We can show
    \begin{align}
        \bigg|
         \Tr[(\bra{0^{\poly(\secp)}}\otimes I)V_\secp(\ket{0^{\poly(\secp)}}\otimes I)\Pi_{\ge\epsilon}]
        \bigg|
        \ge1-2^{2\secp\ell-p}-2^{2\secp\ell}\epsilon\label{eq:bound_of_trace_with_Pi}
    \end{align}
    and 
    \begin{align}
        \|\rho-(\bra{0^{\poly(\secp)}}\otimes I)V_\secp(\ket{0^{\poly(\secp)}}\otimes I)\|_1
        \le2^{\secp\ell-p}.\label{eq:bound_of_trace_distance}
    \end{align}
    We give their proofs later. With these inequalities at hand, we obtain \cref{lem:RSV_subspace_and_support}.

    To conclude the proof, we give proofs of \cref{eq:bound_of_trace_with_Pi,eq:bound_of_trace_distance}. We can show the latter as follows:
    \begin{align}
        &\|
       \rho
        -(\bra{0^{\poly(\secp)}}\otimes I)V_\secp(\ket{0^{\poly(\secp)}}\otimes I)
        \|_1
        \\
        \le&2^{2\secp\ell}
        \|
        \rho
        -(\bra{0^{\poly(\secp)}}\otimes I)V_\secp(\ket{0^{\poly(\secp)}}\otimes I)
        \|_\infty
        \tag{By $\|A\|_1\le d\|A\|_\infty$ for any $A\in\Linear(d)$}
        \\
        \le&2^{2\secp\ell-p}.
        \tag{Since $V_\secp$ is a $(1,2^{-p},\poly(\secp))$-block encoding of $\rho$}
    \end{align}
    We can show the former as follows.
    \begin{align}
        &\bigg|
         \Tr[(\bra{0^{\poly(\secp)}}\otimes I)V_\secp(\ket{0^{\poly(\secp)}}\otimes I))\Pi_{\ge\epsilon}]
        \bigg|
        \\
        =&
        \bigg|
         \Tr[\bra{0^{\poly(\secp)}}\otimes I)V_\secp(\ket{0^{\poly(\secp)}}\otimes I)]
        -
         \Tr[(\bra{0^{\poly(\secp)}}\otimes I)V_\secp(\ket{0^{\poly(\secp)}}\otimes I)(I-\Pi_{\ge\epsilon})]
        \bigg|
        \\
        \ge&
        \bigg|
         \Tr[\bra{0^{\poly(\secp)}}\otimes I)V_\secp(\ket{0^{\poly(\secp)}}\otimes I)]
        \bigg|
        -
        \bigg\|
         (\bra{0^{\poly(\secp)}}\otimes I)V_\secp(\ket{0^{\poly(\secp)}}\otimes I)(I-\Pi_{\ge\epsilon})
        \bigg\|_1,\label{eq:ineq1_for_bound_of_trace_with_Pi}
    \end{align}
    where we have used the triangle inequality and $|\Tr[A]|\le\|A\|_1=\|-A\|_1$ in the inequality. The first term can be estimated as follows:
    \begin{align}
        \bigg|
         \Tr[\bra{0^{\poly(\secp)}}\otimes I)V_\secp(\ket{0^{\poly(\secp)}}\otimes I)]
        \bigg|
        =&
        \bigg|
         \Tr[\rho]
         -
         \Tr[\rho-
         \bra{0^{\poly(\secp)}}\otimes I)V_\secp(\ket{0^{\poly(\secp)}}\otimes I)]
        \bigg|
        \notag\\
        \ge&
        |\Tr[\rho]|-
        \bigg\|
         \rho-
         \bra{0^{\poly(\secp)}}\otimes I)V_\secp(\ket{0^{\poly(\secp)}}\otimes I)
        \bigg\|_1
        \tag{By the triangle inequality and $|\Tr[A]|\le\|A\|_1$}
        \\
        \ge&1-2^{2\secp\ell-p},\label{eq:ineq2_for_bound_of_trace_with_Pi}
    \end{align}
    where we have used \cref{eq:bound_of_trace_distance} in the last inequality. To estimate the second term in \cref{eq:ineq1_for_bound_of_trace_with_Pi}, recall that $\Pi_{\ge\epsilon}$ is the projection onto the subspace spanned by right singular vectors of $\bra{0^{\poly(\secp)}}\otimes I)V_\secp(\ket{0^{\poly(\secp)}}\otimes I)$ whose singular values are at least $\epsilon$. 
    Thus, $I-\Pi_{\ge\epsilon}$ is the projection onto the subspace spanned by right singular vectors of $\bra{0^{\poly(\secp)}}\otimes I)V_\secp(\ket{0^{\poly(\secp)}}\otimes I)$ whose singular values are less than $\epsilon$.
    In addition to that, note that the number of singular values of $\bra{0^{\poly(\secp)}}\otimes I)V_\secp(\ket{0^{\poly(\secp)}}\otimes I)$ is at most $2^{2\secp\ell}$ since $\bra{0^{\poly(\secp)}}\otimes I)V_\secp(\ket{0^{\poly(\secp)}}\otimes I)$ is an operator acting on $2\secp\ell$ qubits.
    With these observations, we have
    \begin{align}
        \bigg\|
         (\bra{0^{\poly(\secp)}}\otimes I)V_\secp(\ket{0^{\poly(\secp)}}\otimes I)(I-\Pi_{\ge\epsilon})
        \bigg\|_1\le2^{2\secp\ell}\epsilon\label{eq:ineq3_for_bound_of_trace_with_Pi}
    \end{align}
    since $\|A\|_1$ is equivalent to the sum of its singular values.
    From \cref{eq:ineq1_for_bound_of_trace_with_Pi,eq:ineq2_for_bound_of_trace_with_Pi,eq:ineq3_for_bound_of_trace_with_Pi}, we obtain \cref{eq:bound_of_trace_with_Pi}.
\end{proof}

\subsection{Proof of \cref{lem:negligible_overlap_with_large_RSVs}}

Finally, we show \cref{lem:negligible_overlap_with_large_RSVs}. We also restate it here.

\smallRSV*

\begin{proof}[Proof of \cref{lem:negligible_overlap_with_large_RSVs}]
    For the notational simplicity, we define 
    $\rho\coloneqq
    (\cM_{\{U'_k\},\ell}\otimes\identitymap)(\ketbra{\Omega_{2^{\secp\ell}}}{\Omega_{2^{\secp\ell}}})$.  Recall that $Q$ is the projection onto $\rho$, which implies $\rho\ket{\psi}=\rho Q\ket{\psi}=0$. Then, we have
    \begin{align}
        \|
         \bra{0^{\poly(\secp)}}\otimes I)V_\secp(\ket{0^{\poly(\secp)}}\otimes I)\ket{\psi}
        \|
        =&\big\|
         \big(\rho-
         \bra{0^{\poly(\secp)}}\otimes I)V_\secp(\ket{0^{\poly(\secp)}}\otimes I)\big)\ket{\psi}
        \big\|
        \\
        \le&
        \|\rho-
         \bra{0^{\poly(\secp)}}\otimes I)V_\secp(\ket{0^{\poly(\secp)}}\otimes I)
         \|_\infty
         \\
         \le&2^{-p},\label{eq:ineq1_for_negligible_overlap_with_large_RSVs}
    \end{align}
    where the last inequality follows from the assumption that $V_\secp$ is a $(1,2^{-p},\poly(\secp))$-blcok encoding of $\rho$.
    Let 
    \begin{align}
        (\bra{0^{\poly(\secp)}}\otimes I)V_\secp(\ket{0^{\poly(\secp)}}\otimes I)=\sum_i a_i\ketbra{w_i}{v_i}
    \end{align}
    be the singular value decomposition. Namely, $\{\ket{w_i}\}_i$ and $\{\ket{v_i}\}_i$ are sets of orthonormal states, and all $a_i$ are positive real numbers.
    Since each $\ket{v_i}$ is the right singular vector of $\bra{0^{\poly(\secp)}}\otimes I)V_\secp(\ket{0^{\poly(\secp)}}\otimes I)$ whose singular value is $a_i$, we have $\Pi_{\ge\epsilon}=\sum_{i:a_i\ge\epsilon}\ketbra{v_i}{v_i}$.
    Then, we have 
    \begin{align}
        \|
         \bra{0^{\poly(\secp)}}\otimes I)V_\secp(\ket{0^{\poly(\secp)}}\otimes I)\ket{\psi}
        \|
        =&
        \bigg\|
         \sum_i a_i\ket{w_i}\braket{v_i|\psi}
        \bigg\|
        \\
        =&
        \sqrt{
        \sum_{i}a_i^2|\braket{v_i|\psi}|^2
        }
        \\
        \ge&
        \sqrt{
        \sum_{i:a_i\ge\epsilon}a_i^2|\braket{v_i|\psi}|^2
        }
        \\
        \ge&\epsilon\sqrt{
        \sum_{i:a_i\ge\epsilon}|\braket{v_i|\psi}|^2
        }
        \\
        =&\epsilon\|\Pi_{\ge\epsilon}\ket{\psi}\|,\label{eq:ineq2_for_negligible_overlap_with_large_RSVs}
    \end{align}
    where in the last inequality we have used $\Pi_{\ge\epsilon}=\sum_{i:a_i\ge\epsilon}\ketbra{v_i}{v_i}$.
    From \cref{eq:ineq1_for_negligible_overlap_with_large_RSVs,eq:ineq2_for_negligible_overlap_with_large_RSVs},
    we have $\epsilon\|\Pi_{\ge\epsilon}\ket{\psi}\|\le2^{-p}$, which implies 
    $\|\Pi_{\ge\epsilon}\ket{\psi}\|\le2^{-p}\epsilon^{-1}$.
\end{proof}

\fi
\section{Relationship Between Black-Box Construction and Oracle Separation}
\subsection{Impossibility of Black-Box Constructions from Oracle Separations}
\label{sec:black-box_construction}
\cite{TCC:ColMut24,ChenColSat24}
showed the relation between oracle separations and black-box constructions.
Therefore, by applying the proof for \cref{Intro_thm:PRU_vs_PRFSG}, we obtain \cref{Intro_thm:main1}.
Here for the convenience of readers,
we provide the proof.

\BlackBox*

\begin{proof}[Proof of \cref{Intro_thm:main1}]
    For the sake of contradiction, assume that there is a black-box construction of PRUs from PRFSGs.
    Then
     there exist QPT algorithms $C^{(\cdot,\cdot)}$ and $R^{(\cdot,\cdot)}$ such that
    \begin{enumerate}
        \item
        Black-box construction: For any PRFSG $G$ and any its unitary implementation $\tilde{G}$,\footnote{In general $G$ is a CPTP map. The CPTP map $G$ can be
        implemented by applying a unitary $\tilde{G}$ on a state and tracing out some qubits. A unitary implementation of $G$ is such a unitary $\tilde{G}$.} $C^{\tilde{G},\tilde{G}^\dag}$ satisfies correctness of non-adaptive PRUs.
        \item Black-box security reduction: For any PRFSG $G$,
        any its unitary implementation $\tilde{G}$,
        any adversary $\cA$ that breaks the security of $C^{\tilde{G},\tilde{G}^\dag}$, 
        and any unitary implementation $\tilde{\cA}$ of $\cA$, 
        it holds that $R^{\tilde{\cA},\tilde{\cA}^\dag}$ breaks the security of $G$.
    \end{enumerate}

    From \cref{Intro_thm:PRU_vs_PRFSG}, 
    there is a QPT algorithm $B^{\cO,\cO^\dagger}$ querying $\cO$ and $\cO^\dagger$
    such that $B^{\cO,\cO^\dagger}$ is a PRFSGs.
    Therefore $C^{\tilde{B}^{\cO,\cO^\dagger},(\tilde{B}^{\cO,\cO^\dagger})^\dagger}$
    satisfies the correctness of non-adaptive PRUs.
    This means that a QPT algorithm $D^{\cO,\cO^\dagger}$ querying $\cO$ and $\cO^\dagger$ satisfies correctness of non-adaptive PRU.
    However, because non-adaptive PRUs do not exist relative to $\cO$ and $\cO^\dagger$ from \cref{Intro_thm:PRU_vs_PRFSG},
    this should not be secure. Therefore there exists a QPT adversary $\cA^{\cO,\cO^\dagger}$ that breaks it.
    Then $R^{\tilde{\cA}^{\cO,\cO^\dagger},\tilde{\cA}^{\cO,\cO^\dagger}}$ breaks $B^{\cO,\cO^\dagger}$, which means that
    a QPT algorithm $E^{\cO,\cO^\dagger}$ breaks the PRFSGs $B^{\cO,\cO^\dagger}$, which is the contradiction.
\end{proof}

\subsection{Black-Box Construction Relative to Oracles}
\label{subsec:many_primitives}
In \cref{def:BB}, we defined a black-box construction from PRFSGs to PRUs. Similarly, we can define a black-box construction for the general cryptographic primitives as follows.

\begin{definition}[Black-Box Construction \cite{TCC:ColMut24,ChenColSat24}]
    We say that a primitive $\cQ$ can be constructed from a primitive $\cP$ in a black-box way
    if there exist QPT algorithms $C^{(\cdot,\cdot)}$ and $R^{(\cdot,\cdot)}$ such that
    \begin{enumerate}
        \item
        Black-box construction: For any QPT algorithm $G$ satisfying the correctness of $\cP$ and any its unitary implementation $\tilde{G}$, $C^{\tilde{G},\tilde{G}^\dag}$ satisfies the correctness of primitive $\cQ$.
        \item Black-box security reduction: For any QPT algorithm $G$ satisfying the correctness of $\cP$, any its unitary implementation $\tilde{G}$,
        any adversary $\cA$ that breaks the $\cP$'s security of $C^{\tilde{G},\tilde{G}^\dag}$, 
        and any unitary implementation $\tilde{\cA}$ of $\cA$, 
        it holds that $R^{\tilde{\cA},\tilde{\cA}^\dag}$ breaks the $\cQ$'s security of $G$.
    \end{enumerate}
\end{definition}

The following is shown in \cite{TCC:ColMut24,ChenColSat24}.

\begin{theorem}[\cite{TCC:ColMut24,ChenColSat24}]\label{thm:black-box_construction}
    Suppose that there exists a black-box construction from primitive $\cP$ to $\cQ$. Then, for any unitary $\cO$, if $\cP$ exist relative to $\cO$ and $\cO^\dag$, $\cQ$ also exist relative to $\cO$ and $\cO^\dag$.
\end{theorem}

By combing \cref{thm:black-box_construction} and \cref{thm:main}, we have the following.

\begin{theorem}
    With probability $1$ over the choice of $\cO$ defined in \cref{def:unitary_orcle}, PRSGs, IND-CPA SKE, EUF-CMA MAC with unclonable tags, UPSGs, private-key money scheme, OWSGs, OWpuzzs, and EFI exist relative to $\cO$ and $\cO^\dag$.
\end{theorem}

\begin{proof}
    From the previous works, there are black-box constructions from PRFSGs to them. Thus, from \cref{thm:black-box_construction} and \cref{thm:main}, we obtain the above claim.
\end{proof}

Therefore, from the above theorem, we have the following. We omit its proof because we can show it by the same argument in the proof of \cref{Intro_thm:main1}.

\begin{theorem}
    There is no black-box construction of non-adaptive and $O(\log\secp)$-ancilla PRUs from PRSGs, IND-CPA SKE, EUF-CMA MAC with unclonable tags, UPSGs, private-key money scheme, OWSGs, OWpuzzs, or EFI.
\end{theorem}

\end{CJK}
\end{document}